\documentclass[11pt,reqno]{amsart}
\usepackage[utf8]{inputenc}
\usepackage{braket,amsmath,amsfonts,amsthm,dsfont,youngtab,mathtools,url,graphicx,url,xcolor,subcaption,appendix,comment,arydshln,enumerate}

\usepackage[normalem]{ulem}

\usepackage[colorlinks=true,linkcolor=blue,urlcolor=blue,citecolor=blue,anchorcolor=green,pdfusetitle]{hyperref}

\usepackage{fullpage}

\newtheorem{lemma}{Lemma}[section]
\newtheorem{prop}[lemma]{Proposition}
\newtheorem{theorem}[lemma]{Theorem}

\newtheorem{cor}[lemma]{Corollary}

\newtheorem{rem}[lemma]{Remark}

\theoremstyle{definition}
\newtheorem{algo}[lemma]{Protocol}

\numberwithin{equation}{section}

\newcommand{\NN}{\mathbb{N}}
\newcommand{\MM}{\mathbb{M}}
\newcommand{\RR}{\mathbb{R}}

\newcommand{\BB}{\mathbb{B}}
\newcommand{\E}{\mathcal{E}}
\renewcommand{\S}{\mathcal{S}}

\newcommand{\pl}{\hspace{.1cm}}

\newcommand{\M}{\mathcal{M}}
\newcommand{\N}{\mathcal{N}}

\DeclareMathOperator{\tr}{tr}

\newcommand{\A}{\mathcal{A}}

\newcommand{\CC}{\mathbb{C}}

\newcommand{\bE}{\mathbf{E}}

\newcommand{\swap}{\mathrm{SWAP}}
\DeclareMathOperator{\Id}{Id}
\newcommand{\one}{\mathds{1}}
\newcommand{\id}{\one}

\newcommand{\kS}{\mathfrak{S}}
\newcommand{\cU}{\mathcal{U}}
\newcommand{\cP}{\mathcal{P}}
\newcommand{\cR}{\mathcal{R}}
\newcommand{\cQ}{\mathcal{Q}}

\newcommand{\cQs}{\cQ_{\mathrm{sw}}}
\newcommand{\cQm}{\cQ_{\mathrm{mct}}}
\newcommand{\cH}{\mathcal{H}}
\newcommand{\cT}{\mathcal{T}}
\newcommand{\cV}{\mathcal{V}}
\newcommand{\cW}{\mathcal{W}}

\newcommand{\cL}{\mathcal{L}}
\newcommand{\bF}{\mathbb{F}}
\newcommand{\bG}{\mathbf{G}}
\newcommand{\bX}{\mathbf{X}}
\newcommand{\bL}{\mathbf{L}}
\DeclareMathOperator{\End}{End}
\DeclareMathOperator{\im}{im}

\newcommand{\nip}[2]{\left\langle #1,#2\right\rangle_{\mathrm{nrm}}}
\newcommand{\bx}{\mathbf{x}}
\newcommand{\bv}{\mathbf{v}}
\newcommand{\bw}{\mathbf{w}}
\newcommand{\by}{\mathbf{y}}
\newcommand{\bz}{\mathbf{z}}
\newcommand{\bM}{\mathbf{M}}
\newcommand{\bN}{\mathbf{N}}
\newcommand{\bD}{\mathbf{D}}
\newcommand{\bY}{\mathbf{Y}}
\newcommand{\bs}{\mathbf{s}}

\newcommand{\bt}{\mathbf{t}}
\newcommand{\br}{\mathbf{r}}
\newcommand{\vsigma}{{\vec{\sigma}}}
\newcommand{\vomega}{{\vec{\omega}}}
\DeclareMathOperator{\lmax}{\lambda_{\mathrm{max}}} 
\DeclareMathOperator{\lmin}{\lambda_{\mathrm{min}}} 

\newcommand{\eps}{\varepsilon}
\newcommand{\eqspace}{\qquad\qquad}
\DeclareMathOperator{\polylog}{polylog}

\DeclareUnicodeCharacter{2009}{\,} 

\usepackage[backend=bibtex8,sorting=nty,giveninits=true,doi=false,isbn=false,url=false,maxbibnames=10]{biblatex}
\renewbibmacro{in:}{}
\newbibmacro{string+doi}[1]{\iffieldundef{doi}{#1}{\href{https://dx.doi.org/\thefield{doi}}{#1}}}
\DeclareFieldFormat{title}{\usebibmacro{string+doi}{\mkbibemph{#1}}}
\DeclareFieldFormat[article]{title}{\usebibmacro{string+doi}{\mkbibquote{#1}}}
\DeclareFieldFormat[incollection]{title}{\usebibmacro{string+doi}{\mkbibquote{#1}}}                   
\DeclareFieldFormat[inproceedings]{title}{\usebibmacro{string+doi}{\mkbibquote{#1}}}

\title[Approximate Unitary $k$-Designs from Shallow, Low-Communication Circuits]{Approximate Unitary $k$-Designs \\ from Shallow, Low-Communication Circuits}
\author{Nicholas LaRacuente}
\address{Nicholas LaRacuente, Department of Computer Science, Indiana University Bloomington, 700 N Woodlawn Ave, Bloomington, IN, 47408, USA}
\email{nlaracu@iu.edu}
\author{Felix Leditzky}
\address{Felix Leditzky, Department of Mathematics and IQUIST, University of Illinois Urbana-Champaign, 1409 W Green St, Urbana, IL 61801, USA}
\email{leditzky@illinois.edu}

\begin{document}

\begin{abstract}
    Random unitaries are useful in quantum information and related fields, but hard to generate with limited resources. 
    An approximate unitary $k$-design is an ensemble of unitaries with an underlying measure over which the average is close to a Haar random ensemble up to the first $k$ moments. 
    A particularly strong notion of approximation bounds the distance from Haar randomness in relative error.
    Such relative-error approximate designs are secure against queries by an adaptive adversary trying to distinguish it from a Haar ensemble.
    We construct relative-error approximate unitary $k$-design ensembles for which communication between subsystems is $O(1)$ in the system size. 
    These constructions use the alternating projection method to analyze overlapping Haar averages, giving a bound on the convergence speed to the full averaging with respect to the $2$-norm. 
    Using von Neumann subalgebra indices to replace system dimension, the 2-norm distance converts to relative error without introducing any additional dimension dependence. 
    We use these constructions as the building blocks of a two-step protocol that achieves a relative-error design in $O \big ( (\log m + \log(1/\epsilon) + k \log k ) k\, \text{polylog}(k) \big )$ depth, where $m$ is the number of qudits in the complete system and $\epsilon$ the approximation error. 
    This sublinear depth construction answers a variant of \cite[Harrow and Mehraban 2023, Section 1.5, Open Question 1]{harrow_mehraban2023approximate} and \cite[Harrow and Mehraban 2023, Section 1.5, Open Question 7]{harrow_mehraban2023approximate}. 
    Moreover, entanglement generated by the sublinear depth scheme follows area laws on spatial lattices up to corrections logarithmic in the full system size.
\end{abstract}

\maketitle

\section{Introduction}

Uniformly random unitary and state ensembles are useful in quantum information \cite{dupuis_decoupling_2010, hayden_decoupling_2008, mele2024introduction} and fundamental physics \cite{aaronson_complexity_2016, piroli_random_2020}. However, an elementary counting argument shows that such unitaries are difficult to construct via polynomially sized quantum circuits \cite{nielsen_quantum_2010}. More accessible constructions approximate the properties of random unitaries \cite{brown_decoupling_2015, deshpande_tight_2022, nakata_efficient_2017}. The unitary $k$-design is one such notion: 
a measure over the unitary group that is indistinguishable from Haar random in the first $k$ statistical moments. Equivalently, any protocol that samples a unitary from a $k$-design $k$ or fewer times would be unable to distinguish it from the Haar ensemble by any test using $k$ or fewer copies with probability better than randomly guessing.

As exact designs are still challenging to construct, approximate designs often take their place. There are several notions of approximate designs. 
A particularly strong notion considered in \cite{brandao_local_2016} and \cite{chen2024incompressibility} is that of an approximate design in \emph{multiplicative} or \emph{relative error}.
Multiplicative error designs are secure against queries by an adaptive adversary trying to distinguish it from Haar ensembles \cite{chen_efficient_2024}. This relative notion implies a weaker, additive notion requiring that ensemble averages (or \emph{twirls}) over $k$-fold applications of unitaries from the design are close to the Haar ensemble average in diamond norm distance. Designs in diamond norm distance are secure against non-adaptive adversaries, and the parameter $k$ is a lower-bound on the circuit complexity required for typical unitaries in such a design \cite{brandao2021models}. Moreover, a weaker but related notion replaces the diamond norm distance by the two-norm distance in what is called a tensor product expander (TPE). We review these notions in more technical detail in Section \ref{sec:design-background} below.

Many previous works have studied how random circuits converge to $k$-designs in bounded depth \cite{harrow2009random, brandao_local_2016, hunter-jones_unitary_2019, harrow_mehraban2023approximate, haferkamp_improved_2021}. Recent results have yielded efficient constructions in required circuit depth on 1-dimensional connectivity \cite{haah2024efficient, metger_simple_2024, chen_efficient_2024}. Even more recently, it was shown that random 2-local circuits on one-dimensional connectivity yield $\epsilon$-approximate relative error $k$-designs on $m$-qudit systems in depth $O((m k + \log(1/\epsilon)) \text{polylog}(k))$, and random local circuits on all-to-all connectivity do so with gate count $O(m (m k + \log(1/\epsilon)) \text{polylog}(k))$ \cite{chen2024incompressibility}.

Nevertheless, many questions have remained open (see \cite[Section 1.5]{harrow_mehraban2023approximate} and \cite[Section 1.4]{chen2024incompressibility}). In particular, we highlight \cite[Section 1.5, Open Question 1]{harrow_mehraban2023approximate}, which asks whether one can construct a multiplicative error $k$-design on a system of $m$ qubits using random circuits with depth sublinear in $m$. Our results imply that this is possible, even in one-dimensional spatial lattice connectivity, if the circuits are allowed to follow a layer-dependent architecture as detailed in Section \ref{sec:lattices}.

We initially approach the problem by asking a different question: how much quantum communication do two parties need to jointly construct an approximate unitary $k$-design ensemble? Some inspiration comes from pseudoentangled state ensembles \cite{aaronson_quantum_2023}. Based on widely believed complexity-theoretic assumptions, these ensembles are hard to distinguish from uniformly random state ensembles for any polynomial-time quantum algorithm. Nonetheless, the entanglement across any cut of a pseudoentangled state ensemble need only scale slightly faster than logarithmic in the subsystem size. As noted in \cite[p.1]{aaronson_quantum_2023}, earlier results rule out an analog of pseudoentanglement for exact or highly precise unitary $k$-design ensembles, as statistically approximating the Haar measure to exponential precision requires near-maximal entanglement \cite{low_large_2009, cotler_fluctuations_2022}. Still, these barrier results left open whether smaller amounts of entanglement could suffice for an $\epsilon$-approximate $k$-design with constant $\epsilon$ in the system size.

\subsection{Summary of Main Results}
\label{sec:main-results}
We show that for $k$ significantly smaller than the logarithm of system dimension, it is indeed possible to construct unitary $k$-designs with relatively little quantum communication between subsystems, hence generating little entanglement. 
Our main tool is to consider protocols consisting of overlapping Haar twirls on smaller subsystems.
An application of von Neumann's alternating projection method \cite{vonNeumann1949rings} shows that such overlapping twirls form tensor product expanders, i.e., approximate designs with respect to the $2$-norm distance.
To obtain the much stronger notion of relative-error designs discussed in the introduction, we use a technique from operator algebra theory based on von Neumann subalgebra indices \cite{pimsner1986entropy} to convert these $2$-norm estimates to relative error without incurring any dimension factor penalties.
Our central result, constructing relative-error approximate designs in logarithmic circuit depth, is achieved by carefully building up a two-step protocol consisting of these overlapping local Haar twirls.
We now phrase our main results in more mathematical detail.

For a unitary measure $\mu$ on $U(d)$, let
\begin{align} \Phi_{\mu, k}(\rho) \coloneqq \int U^{\otimes k} \rho U^{\otimes k \dagger} d \mu(U) \end{align}
for every input state $\rho$ on a system of dimension $d$. An approximate (relative error) $k$-design is a measure $\mu$ on a finite-dimensional unitary group for which 
\begin{equation}
        (1-\epsilon) \Phi_{\text{Haar}, k} \prec \Phi_{\mu, k} \prec (1+\epsilon) \Phi_{\text{Haar}, k} \pl,
    \end{equation}
where by $\mathrm{Haar}$ we denote the Haar or uniform measure over the unitary group.
\begin{enumerate}
    \item\label{item:twirl-swap-twirl} \textbf{Twirl-Swap-Twirl:} Let $A$ and $B$ be systems of the same size. Consider a protocol that locally applies unitaries drawn from $k$-designs to $A$ and $B$ respectively, exchanges the same $\ell$ qudits between $A$ and $B$ respectively, then again applies local $k$-design unitaries. The unitary ensemble induced by such a protocol is an $\epsilon$-approximate relative $k$-design when
    $\ell = O(k \log k + \log( 1/\epsilon))$.
    In particular, this bound is \textit{independent} of the sizes of $A$ and $B$ as long as each contains at least $\omega(k \log k + \log(1/\epsilon))$ qudits. The full technical version of this result with explicit constants is stated in Theorem \ref{thm:twirl-swap-twirl-relative}. \smallskip
    
    \item\label{item:twirl-crosstwirl} \textbf{Twirl-Crosstwirl:} Let $A_1, \dots, A_P$ be subsystems of a multipartite system $A$. Consider the following protocol for $K$ copies of $A$: (1) locally apply a $K$-design unitary to each $A_p$ for $p = 1, \dots, P$; (2) apply a $K$-design on a joint system combining $\ell$ qudits from each $A_p$. 
    Figure \ref{fig:multipartite-twirl-crosstwirl-cartoon} illustrates this protocol for $K=3$ copies of $P=2$ parties.
    The Twirl-Crosstwirl protocol achieves an $\epsilon$-approximate relative $K$-design when
    \begin{align} \ell \geq 2 P K \log_q K + 2 \log_q K + \log_q P + \log_q(1/\epsilon) + \log_q 12 \pl, \end{align}
    $10K^2 \leq q^{\ell}$, and each $A_p$ contains at least $3 \ell$ qudits. This result is stated in Theorem \ref{thm:twirl-crosstwirl-relative}. Shown as Corollary \ref{cor:twirl-crosstwirl-tpe}, the protocol yields an $\epsilon$-approximate $K$-TPE (as defined in Section \ref{sec:background} when
    \begin{align} \ell \geq 2 \log_q K + \log_q P + \log_q (1/\epsilon) + \log_q (5/2) \pl. \end{align}
    Compared with the $K \log K$ dependence for designs, the $K$-TPE is achieved with logarithmic $K$ dependence. \smallskip
    
    \item\label{item:recursive-crosstwirl} \textbf{Two-step Crosstwirl:} Consider an $m$-qudit system with connectivity given by a lattice in spatial dimension $D$. We use a two-step protocol summarized in Protocol \ref{alg:2-layer} that applies overlapping unitaries to construct an $\epsilon$-approximate relative $k$-design using unitaries of circuit depth
    \begin{align} O \big ( k (k \log k + \log m + \log(1/\epsilon) ) \text{polylog}(k) \big ) \pl. \end{align}
    Moreover, for any spatially contiguous region $S$, the design requires only
    \begin{align} O \big ( \#(\partial S) \times \ell \big ) \end{align}
    qudits of quantum communication and generates only that much entanglement on any product state input.
    Here, $\#(\partial S)$ denotes the number of qudits on the boundary of $S$, and $\ell$ scales as $O(\log m)$. The technical version of this result is stated as Theorem \ref{thm:lattice-main} in Section \ref{sec:lattices}. This result addresses \cite[Section 1.5, Open Question 1]{harrow_mehraban2023approximate}, showing that with a specific, layered architecture, random circuits produce relative error $k$-designs in sublinear depth. It also address \cite[Section 1.5, Open Question 7]{harrow_mehraban2023approximate}, showing that circuits much shallower than the lattice diameter suffice.
\end{enumerate}

\subsection{Overview of the Proof Ideas}

Approximate unitary $k$-designs can be defined using different norms or approximations quantifying how well the design approximates the first $k$ moments of the full Haar measure.
The following choices are commonly used and defined in Section~\ref{sec:tpeconvert}: the $2\to 2$-norm, an extension to superoperators of the Schatten $2$-norm (leading to tensor product expanders or TPEs), the tensor-stabilized diamond norm (leading to additive designs), and the complete positivity (CP) partial order (leading to multiplicative or relative error designs). 
These three figures of merit are ordered by increasing strength: multiplicative error implies additive error, which in turn implies an error measured by the $2\to 2$-norm.

Our strategy is to prove that the protocols \hyperref[item:twirl-swap-twirl]{Twirl-Swap-Twirl} and \hyperref[item:twirl-crosstwirl]{Twirl-Crosstwirl} mentioned in Section \ref{sec:main-results} yield approximate designs with respect to the $2\to 2$-norm, which can be converted into the stronger notions of approximate design in terms of diamond norm and CP order.
In the generic case, such a conversion incurs dimensional factors that considerably loosen the obtained bounds.
However, the channels and states resulting from approximate design constructions only act on certain von Neumann subalgebras of the full Hilbert space.
Restricting the action of the channels to these subalgebras allows us to replace the full Hilbert space dimension by an effective dimension that excludes the multiplicities appearing in this algebra decomposition.
This bound conversion is the content of Section~\ref{sec:tpeconvert}.

In Section \ref{sec:alternating-projections} we prove the $2\to 2$-norm bounds for the \hyperref[item:twirl-swap-twirl]{Twirl-Swap-Twirl} and \hyperref[item:twirl-crosstwirl]{Twirl-Crosstwirl} protocols.
This analysis is closely related to the subspace angle determining the convergence speed in von Neumann's alternating projection method \cite{neumann_rings_1949}.
The main idea for proving the $2\to 2$-norm bounds is to exploit the commutant structure of local twirls implied by Schur-Weyl duality (see Section \ref{sec:schurweyl}), which can be used to expand the images of these twirls as linear combinations of tensor products of permutation operators.
The norm expressions can then be converted into inner products of the coefficient vectors in these expansions, acted upon by matrices whose entries are functions of normalized inner products of permutation operators.
We simplify these expressions using the notion of approximate orthogonality of permutation operators \cite{harrow_approximate_2023-1} together with various matrix norm and eigenvalue estimates.

The \hyperref[item:twirl-crosstwirl]{Twirl-Crosstwirl} protocol is the main building block for the lattice protocol described in Protocol \ref{alg:2-layer} in Section \ref{sec:lattices}. Here, we build up larger approximate designs from small designs by iterating the \hyperref[item:twirl-crosstwirl]{Twirl-Crosstwirl} protocol, while carefully keeping track of the error accumulated in the process.
We specifically consider spatial lattices in (spatial) dimension $D$, for which our protoocl yields approximate designs with the advertised properties, concluding the proof of our main result.

\subsection{Discussion}
The \hyperref[item:twirl-swap-twirl]{Twirl-Swap-Twirl} and \hyperref[item:twirl-crosstwirl]{Twirl-Crosstwirl} results are primarily about the entanglement or quantum communication needed to produce a design. One may replace the local $k$-designs in Twirl-Swap-Twirl by $\delta$-approximate relative $k$-designs to achieve error $(1+\epsilon)(1+\delta)^4 - 1$. Similarly, one may replace the local and crossing $k$-designs in \hyperref[item:twirl-crosstwirl]{Twirl-Crosstwirl} by $\delta$-approximate relative $k$-designs, achieving error $(1+\epsilon)(1+\delta)^{P+1} - 1$. Via the results of \cite{chen2024incompressibility}, if $A$ and $B$ are each $m$-qudit systems, this means these two schemes are approximately implementable in local circuit depth $O((m k + \log(1/\epsilon)) \text{polylog}(k))$ on a one-dimensional line of qudits. The main technical methods in both these results are (1) a new subspace angle approach using Schur-Weyl duality to show efficient TPE bounds, and (2) a more efficient conversion from TPE to relative error that replaces system dimension by von Neumann subalgebra indices.

\hyperref[item:recursive-crosstwirl]{Two-step Crosstwirl} is more directly analogous to the ``across any cut" notion of pseudoentanglement as in \cite{aaronson_quantum_2023}. In this case, up to a logarithmic dependence in the total system size $m$, the entanglement scales like an area law. Area-law versus volume-law entanglement is an important concept in condensed matter and high-energy physics \cite{eisert2010colloquium}. Moreover, \hyperref[item:recursive-crosstwirl]{Two-step Crosstwirl} achieves a new record for scaling of circuit depth in the system size. Though \hyperref[item:recursive-crosstwirl]{Two-step Crosstwirl} is not exactly equivalent to lattice brickwork, it can be implemented using 2-qudit Haar random unitaries. In Remark \ref{rem:almost-brickwork}, we discuss how \hyperref[item:recursive-crosstwirl]{Two-step Crosstwirl} can be implemented on an architecture that closely resembles brickwork in one dimension.

In general, these quantum communication and entanglement bounds are primarily intended for the regime in which the dimension of each system is larger than $k!$. A single iteration only works if subsystems support $\ell$-qudit exchanges, placing a lower bound on dimension. If $k!$ is too large for the subsystems, then it is still possible to obtain designs by repeating our constructions. For large $k$, however, the quantum communication required is no longer smaller than that of general random circuits, and the entanglement entropy resembles a volume law. Nonetheless, that \hyperref[item:recursive-crosstwirl]{Two-step Crosstwirl} achieves sublinear depth in $m$ may suggest further investigation of potential efficiencies of recursive constructions.

That $k$-designs form in logarithmic depth even on one-dimensional brickwork is highly surprising.
Indeed, it had previously been argued that depth $O(m k)$ would likely be optimal for $m$ qudits \cite{hunter-jones_unitary_2019, harrow_mehraban2023approximate} (though this depends on the specific notion of design considered).
Earlier results had also supported the intuition that the depth needed to form a design would be sensitive to the diameter of a connectivity graph and presence of bottlenecks \cite{harrow_mehraban2023approximate}. Efficient constructions of $2$-designs are particularly sought because of their many applications, e.g.~in decoupling \cite{dupuis_decoupling_2010, hayden_decoupling_2008,brown_decoupling_2015, mele2024introduction}. The fact that 2-designs form via logarithmic-depth circuits means that they could likely be constructed on noisy quantum computers, which makes them potentially useful in near-term applications of quantum computing in coding and communication tasks. In contrast, the ease of forming these designs suggests that barren plateaus or similar phenomena in quantum optimization may emerge more readily than expected \cite{larocca_review_2024}.

\subsection{Open Questions}
The $2$-norm convergence of the \hyperref[item:twirl-crosstwirl]{Twirl-Crosstwirl} protocol can be determined quite elegantly using the generalized Perron-Frobenius Theorem (see Section~\ref{sec:alternating-projections} for details), leading to a tight convergence rate. 
The \hyperref[item:twirl-swap-twirl]{Twirl-Swap-Twirl} protocol used a coarser analysis leading to a weaker convergence bound in this case.
Whether the \hyperref[item:twirl-swap-twirl]{Twirl-Swap-Twirl} analysis methods could be improved and unified with those of \hyperref[item:twirl-crosstwirl]{Twirl-Crosstwirl} is an interesting open problem, which we explain in more detail in Sec.~\ref{sec:technical-open-problem}.
Another remaining question from our work is whether standard brickwork random circuits in one dimension or any dimension produce relative error designs in sublinear depth.

\subsection{Structure of this Paper}
The remainder of this paper is structured as follows.
In Section~\ref{sec:background} we give some background on approximate $k$-designs, von Neumann algebras and index theory, and Schur-Weyl duality.
In Section~\ref{sec:tpeconvert} we discuss how tensor product expander (TPE) bounds in the $(2\to 2)$-norm can be converted into efficient relative error bounds using the algebraic structure of von Neumann subalgebras.
In Section~\ref{sec:alternating-projections} we prove TPE bounds for the \hyperref[item:twirl-swap-twirl]{Twirl-Swap-Twirl} and \hyperref[item:twirl-crosstwirl]{Twirl-Crosstwirl} protocols mentioned above using the alternating projection method.
In Section~\ref{sec:lattices} we discuss the \hyperref[item:recursive-crosstwirl]{Two-step Crosstwirl} protocol, a two-step protocol built up from the \hyperref[item:twirl-crosstwirl]{Twirl-Crosstwirl} protocol to construct approximate designs on spatial lattice geometries.

\subsection{Acknowledgments}
We acknowledge helpful email exchanges with Aram Harrow during the development of these results. FL acknowledges support by National Science Foundation grant no.~2137953. Part of this work was completed during the workshop “Beyond i.i.d.~in Information Theory 11” hosted by the University of T\"{u}bingen from July 31 to August 4, 2023.

\textit{Note added:} We note a concurrent and independent work on low-depth designs and pseudorandom unitaries \cite{schuster_random_2024} by Thomas Schuster, Jonas Haferkamp, and Hsin-Yuan Huang. Their work also obtains approximate design ensembles via $O(\log m)$ circuit depth and improves the communication scaling in $k$ to $O(\log k)$ via a different technical analysis. Our work uses a different extension to higher-dimensional lattices for a closer analogy to area law entanglement (see Remark \ref{rem:hamiltonian-vs-recursive}).

\section{Background}\label{sec:background}
A (finite-dimensional) quantum channel is a completely positive, trace-preserving map on a space of matrices. Following the physics convention, we apply quantum channels in the Schr\"odinger picture, such that a channel $\Phi$ is considered to act on states, and its adjoint on observables. We call a quantum channel trace-symmetric if it is its own adjoint under the trace, $\tr(\Phi(x) y) = \tr(x \Phi(y))$. 

The Schatten $p$-norm of a matrix $\rho$ is defined as $\|\rho\|_p = \tr(|\rho|^p)^{1/p}$ with $|\rho|\coloneqq \sqrt{\rho^\dagger\rho}$. We recall the $p \rightarrow q$ distance
\begin{align} \| \Phi - \Psi \|_{p \rightarrow q} = \sup_{\rho\neq 0} \frac{\| (\Phi - \Psi)(\rho) \|_q}{\|\rho\|_p} \end{align}
and its completely bounded strengthening,
\begin{align} \| \Phi - \Psi \|_{p \rightarrow p, \mathrm{cb}} = \sup_{d \in \NN} \| (\Phi - \Psi) \otimes \Id_{(d)} \|_{p \rightarrow p} \pl, \end{align}
where $\Id_{(d)}$ is the identity channel in dimension $d$. The diamond norm is defined as $\| \cdot \|_{\Diamond} \coloneqq \| \cdot \|_{1 \rightarrow 1, \mathrm{cb}}$. Furthermore, we recall the complete semidefinite ordering on a pair of channels $\Phi, \Psi$, in which $\Phi \prec \Psi$ iff $\Psi \otimes \Id_{(d)} - \Phi \otimes \Id_{(d)}$ is a positive map for every $d$. The notation $\succ$ is defined analogously.

\subsection{Designs} \label{sec:design-background}
We briefly recall notions of approximate $k$-designs and related concepts. For a unitary measure $\mu$ on $U(d)$, we define the $k$-fold weighted twirl channel
\begin{align} \Phi_{\mu, k}(\rho) \coloneqq \int U^{\otimes k} \rho U^{\otimes k \dagger} d \mu (U) \end{align}
for every input state $\rho$ on a system of dimension $d$. Let $\Phi_{\text{Haar}, k}$ denote such a construction with respect to the Haar measure on $U(d)$.
A measure $\mu$ on $U(d)$ is an $\epsilon$-approximate...
\begin{itemize}
    \item ...tensor product expander (TPE) if
    \begin{equation}
        \| \Phi_{\mu, k}  - \Phi_{\text{Haar}, k} \|_{2 \rightarrow 2} \leq \epsilon \pl.
    \end{equation}
    \item ...additive $k$-design if
    \begin{equation}
        \| \Phi_{\mu, k}  - \Phi_{\text{Haar}, k} \|_{\Diamond} \leq \epsilon \pl.
    \end{equation}
    \item ...multiplicative or relative error $k$-design if
    \begin{equation}
        (1-\epsilon) \Phi_{\text{Haar}, k} \prec \Phi_{\mu, k} \prec (1+\epsilon) \Phi_{\text{Haar}, k} \pl.
    \end{equation}
\end{itemize}
Multiplicative error is a stronger criterion than and implies additive error, which in turn applies the TPE condition. As a more general criterion that we may use in intermediate results, we say that a measure is a design with multiplicative error $(\epsilon, \delta)$ if
\begin{equation}
     (1-\epsilon) \Phi_{\text{Haar}, k} \prec \Phi_{\mu, k} \prec (1+\delta) \Phi_{\text{Haar}, k} \pl.
\end{equation}
For a $k$-design $\mu$, we refer to $\Phi_{\mu,k}$ as the corresponding $k$-design channel. For a pair of measures $\mu$ and $\nu$, we denote by $\mu * \nu$ the convolution, defined as the measure for which $\Phi_{
\mu * \nu} = \Phi_{\mu, k} \circ \Phi_{\nu, k}$ for all $k \in \NN$.

\subsection{Von Neumann Algebras and Entropy}
The proofs in Section \ref{sec:tpeconvert} involve finite-dimensional von Neumann algebra index theory, for which we now provide some background. By $\BB(\cH)$ we denote the bounded operators on Hilbert space $\cH$. Within finite dimension, this is equivalent to the space $\MM_d$ of matrices acting on $d$-dimensional Hilbert space. A von Neumann subalgebra $\N \subseteq \MM_d$ of the finite-dimensional matrix algebra $\MM_d$ consists of all matrices with the block diagonal form
\begin{equation}
    \N = \bigoplus_{\lambda} \A_\lambda \otimes \one_{B_\lambda} \pl,
\end{equation}
where each $\A_\lambda$ is itself a matrix subalgebra on subsystem $A_\lambda$, and each $B_\lambda$ corresponds to a subsystem on which an effective Haar unitary average has been performed \cite{neumann_rings_1949, gao_relative_2020}.
The notation ``$\one_{B_\lambda}$" distinguishes that this is the algebra of multiples of the identity matrix, also sometimes denoted $\CC$, as it is equivalent to the algebra of complex scalars. The dimension of such a mixed system is often referred to as a multiplicity, since the block $A_\lambda$ is repeated a number of times equal to the dimension of $\one_{B_\lambda}$, and then $B_\lambda$ is called a multiplicity space. By $\M_*$ we denote the predual of the von Neumann algebra $\M$. Although within finite dimension $\M_* \cong \M$, one conventionally thinks of states as being elements of the predual and observables being elements of the algebra. By $|\M|$ we denote the dimension of $\M$ as a vector space. Note that the usable, effective dimension of $\M$ is generally smaller, as multiplicities contribute to the space's dimension without contributing to the amount of information that could potentially be stored in a system represented by $\M$.

With each von Neumann algebra $\N$ comes a unique, unital, trace-preserving conditional expectation $\E_\N$. The action of $\E_\N$ on a matrix can be broken into two steps. First, remove off-diagonal blocks, taking $\MM_d \rightarrow \bigoplus_{\lambda} \A_\lambda \otimes B_\lambda$. Second, perform a complete mixing of each $B_\lambda$ system. A subalgebra $\N \subseteq \M$ is a subset of $\M$ that is also a von Neumann algebra. For $\M$ a von Neumann algebra and $\N$ a subalgebra, $\N \subseteq \M$, it follows that $\E_\M \E_\N = \E_\N \E_\M = \E_\N$.

A simple example of a conditional expectation is the completely depolarizing channel mapping to the subalgebra denoted $\one$ or $\CC$. We may for instance take a bipartite system $A \otimes B$, then write $\E_{\A \otimes \CC}$ as the conditional expectation that completely depolarizes the $B$ system. Another somewhat complementary example is the pinching to a basis: this removes all of the off-diagonal terms (with respect to that basis) of its input matrix. A more sophisticated example appears in the following Section \ref{sec:schurweyl} in the context of representation theory.

\subsection{Schur-Weyl Duality} \label{sec:schurweyl}
For a (finite-dimensional) Hilbert space $\cH$ we denote by $\cL(\cH)$ or $\End(\cH)$ the space of linear operators acting on $\cH$.
We recall Schur-Weyl duality in the general setting. Fixing a dimension $d\geq 2$ and a positive integer $k\in\mathbb{N}$, we consider the following unitary representations of the symmetric group $\kS_k$ and the unitary group $\cU_d$ on the state space $\cH = (\mathbb{C}^d)^{\otimes k}$:
\begin{align}
	\varphi\colon \kS_k &\longrightarrow \cU(\cH), \quad \pi \longmapsto \left(|\phi_1\rangle\otimes \dots\otimes |\phi_k\rangle \mapsto |\phi_{\pi^{-1}(1)}\rangle\otimes \dots\otimes |\phi_{\pi^{-1}(k)}\rangle \right)\label{eq:Sk-rep}\\
	\psi\colon \cU_d &\longrightarrow \cU(\cH), \quad U \longmapsto U^{\otimes k}\label{eq:Ud-rep}
\end{align}
It is straightforward to check that the two actions on $\cH$ commute, $[\psi(U),\varphi(\pi)] = 0$ for all $U\in\cU_d$ and $\pi\in\kS_k$.
This implies that the algebra $\langle \varphi(\pi) : \pi\in\kS_k\rangle \subseteq \End(\cH)$ spanned by the representation \eqref{eq:Sk-rep} is contained in the commutant of the algebra $\langle \psi(U) : U\in\cU_d\rangle\subseteq \End(\cH)$ spanned by the representation \eqref{eq:Ud-rep}, and vice versa.
Schur-Weyl duality (see, e.g., \cite{fulton2013representation}) states that, in fact, these two algebras are \emph{equal} to each other's commutant.
It follows that, if an operator $X \in\cL(\cH)$ commutes with all $\psi(U)$ for $U\in\cU_d$, it can be written as
\begin{align}
    X = \sum_{\pi\in\kS_k} x_\pi \varphi(\pi) \label{eq:unitary-invariant-X-permutations}
\end{align}
for some coefficients $x_{\pi}\in\mathbb{C}$.
Together with the complete reducibility of the representations \eqref{eq:Sk-rep} and \eqref{eq:Ud-rep}, Schur-Weyl duality furthermore yields the following useful decompositions of state space and representations:
\begin{align}
	\cH &\cong \bigoplus_{\lambda\vdash_d k} \cV_\lambda^d \otimes \cW_\lambda \label{eq:schur-weyl-decomposition}\\
	\varphi(\pi) &\cong \bigoplus_{\lambda\vdash_d k} \one_{\cV_\lambda^d} \otimes \varphi_\lambda(\pi) \\
    \psi(U) &\cong \bigoplus_{\lambda\vdash_d k} \psi_\lambda^d(U) \otimes \one_{\cW_\lambda}. 
\end{align}
In all three equations the direct sums run over partitions $\lambda = (\lambda_1,\dots,\lambda_d)\vdash_d k$ of $k$ into $d$ parts, i.e., $\lambda_1\geq \dots \geq \lambda_d \geq 0$ and $\sum_{i=1}^d \lambda_i = k$.\footnote{
	Partitions are often graphically represented as Young diagrams, a left-aligned arrangement of $k$ boxes with $\lambda_i$ boxes in the $i$-th row.
	For example, the Young diagram ${\tiny\Yvcentermath1\yng(3,1)}$ corresponds to the partition $\lambda = (3,1)$.}
These partitions label the irreducible representations (irreps) $(\cV_\lambda^d,\psi_\lambda)$ of $\cU_d$ and $(\cW_\lambda,\varphi_\lambda)$ of $\kS_k$.
A consequence of Schur-Weyl duality is that the $\kS_k$-irrep space $\cW_\lambda$ is the multiplicity space of the $\cU_d$-irrep $(\cV_\lambda^d,\psi_\lambda^d)$, and vice versa.
By Schur's Lemma, any operator $X\in\cL(\cH)$ commuting with all $\psi(U)$ is of the form
\begin{align}
	X = \bigoplus_{\lambda\vdash_d k} \one_{\cV_\lambda^d} \otimes X_\lambda \label{eq:unitary-commuting-block-structure}
\end{align}
for some $X_\lambda\in\cL(\cW_\lambda)$, and the analogous statement holds for operators commuting with all $\varphi(\pi)$.
The dimensions of the irreps appearing in \eqref{eq:schur-weyl-decomposition} are given by the following well-known formulae:
\begin{align}
	d_\lambda &\coloneqq \dim \cW_\lambda = \frac{k!}{\prod_{(i,j)\in\lambda} h(i,j)} \label{eq:d-lambda}\\
	m_\lambda^d &\coloneqq \dim \cV_\lambda^d = \prod_{1\leq i < j\leq d} \frac{\lambda_i - \lambda_j + j - i}{j-i}. \label{eq:m-lambda}
\end{align}
In Eq.~\eqref{eq:d-lambda} the hook length $h(i,j)$ of a box $(i,j)\in\lambda$ is defined as the sum of the number of boxes to the right and below of $(i,j)$ plus the box itself.

\section{TPE to Relative Error Designs} \label{sec:tpeconvert}
A common strategy for proving approximate design bounds is to first prove a tensor product expander (TPE) bound, that is, a $(2 \rightarrow 2)$-norm bound on the distance from the approximate to an exact design, which can then be converted to diamond norm, semidefinite order inequality, or other stronger criteria. Conventionally, this conversion involves a factor of the dimension of the space, usually $q^{2 m k}$ up to some constants in the exponential \cite[Lemma 3]{brandao_local_2016}. To counter this dimension factor, a standard approach applies additional circuit layers of depth $O(m k)$, reducing the ultimate $(2 \rightarrow 2)$-norm or related bound by a factor of $2^{-O(m k)}$. Previously, $O(m k)$ depth was conjectured to be optimal at least in one-dimensional connectivity \cite{hunter-jones_unitary_2019}. The conjectured $O(m)$ depth also matched intuition from light cones: at much smaller depths, information would fail to propagate from one end of a system to the other end.

The culminating results of this work show that the light cone intuition fails, and that $O(\log m)$ suffies with a polynomial $k$-dependence. To accomplish this, we show how in our circumstances, the conversion from $(2 \rightarrow 2)$-norm is much more efficient.

\subsection{Summary of This Section's Results}

When restricting states and channels to von Neumann subalgebras of the full Hilbert space, one may often replace the space's dimension by a lower effective dimension that excludes subspace multiplicity. One may commonly use the Pimsner-Popa index \cite{pimsner1986entropy} of a von Neumann algebra inclusion as an effective dimension in norm conversion, as noted in \cite{gao_relative_2020} and as \cite[Lemma A.1]{gao2022complete}. In this section, we present elementary proofs of this intuition in relevant cases. In particular:
\begin{itemize}
    \item Lemma \ref{lem:2-to-inf} shows that one may convert more efficiently from the $(2 \rightarrow 2)$-norm distance between two channels to $\infty$-norm distance of their Choi matrices when both channels restrict their inputs and outputs to a von Neumann algebra. The intuition for this result is that if a 1-normalized operator has large multiplicity (every eigenvalue is duplicated many times), then it must necessarily have small 2-norm and even smaller $\infty$-norm.
    \item Lemma \ref{lem:inf-to-rel} converts from the $\infty$-norm distance of Choi matrices to relative error. This Lemma uses Weyl's inequality to bound the ratios between eigenvalues in terms of the ratio between the $\infty$-norm distance of Choi matrices and each matrix's smallest eigenvalue.
    \item  Lemma \ref{lem:relative-convert} applies Lemmas \ref{lem:2-to-inf} and \ref{lem:inf-to-rel} to channels that start and end by applying $k$-fold twirls on subsystems, converting $(2 \rightarrow 2)$-norm to relative error without introducing any extra $m$-dependence. It uses that once local twirls have been applied, the input and output are each restricted to block diagonal subspaces in which each block has multiplicity that cancels the $m$-dependence of the full system dimension. The representation theory of the symmetric group factors heavily into the efficiency of this conversion, as the dimensions of unitary irreps only enter as multiplicities of symmetric irreps. The value of $k$ should be less than the single-copy dimension of each pre-twirled or post-twirled subsystem, lest the single-copy dimension constrain permutation irreps in ways that complicate the representation theory.\footnote{The isotypical decomposition \eqref{eq:schur-weyl-decomposition} of the tensor representation \eqref{eq:Sk-rep} of $\kS_k$ on $(\mathbb{C}^d)^{\otimes k}$ into irreps only includes those irreps $\cW_\lambda$ labeled by Young diagrams $\lambda$ with at most $d$ rows. Hence, if $d<k$ then some $\kS_k$-irreps will not appear in the isotypical decomposition. For example, the irrep of $\kS_3$ labeled by the partition $(1,1,1)$ of $3$ corresponding to the antisymmetric subspace does not appear in the isotypical decomposition of $(\mathbb{C}^2)^{\otimes 3}$ with respect to the representation \eqref{eq:Sk-rep}.}
\end{itemize}
The above allows us to efficiently convert from $(2-2)$-norm to relative error when a protocol applies a small number of Haar pre- and post-twirls to a few subsystems. Ultimately, we will want to replace these by twirls weighted via approximate $k$-designs, which do not perfectly project to block diagonal subspaces. However, we can effectively bypass that problem by arguments analogous to the triangle inequality. Specifically, if $(1-\epsilon) \Theta \prec \Phi \prec (1+\epsilon) \Theta$, and $(1-\epsilon) \Phi \prec \Phi' \prec (1+\epsilon) \Phi$, then $((1-\epsilon) / (1+\epsilon)) \Theta \prec \Phi'$, and $\Phi' \prec ((1+\epsilon)/(1-\epsilon)) \Theta$. Hence a general theme in applying Lemma \ref{lem:relative-convert} is to derive results assuming that exact local twirls are composed to an approximate global twirl, then compose errors from replacing local twirls by approximations.


\subsection{Norm Conversion: From Two-norm Bounds to Relative Error}

\begin{rem} \normalfont \label{rem:condexpbp}
    In the following we consider the strucure of states on a finite-dimensional von Neumann algebra. It has long been known \cite{vonNeumann1949rings} that a finite-dimensional von Neumann algebra has a decomposition of the form
    \begin{equation} \label{eq:basicblocks}
        \M \cong \bigoplus_{\lambda} ( \BB(\CC^{a_{\lambda}}) \otimes \S(\CC^{b_{\lambda}}) ) \pl,
    \end{equation}
    where $a_\lambda$ and $b_\lambda$ respectively denote the dimensions of matrix and scalar multiple blocks, $\lambda$ is the block index, $\BB(\CC^{a_{\lambda}})$ is the matrix algebra of bounded operators on a complex vector space of dimension $a_\lambda$, and $\S(\CC^{b_{\lambda}})$ is the algebra of scalar multiples of the identity on a space of dimension $b_\lambda$. Every von Neumann algebra admits a conditional expectation, an idempotent quantum channel that restricts input operators (in particular, density matrices) to that algebra. Let $\E_\M$ denote the conditional expectation to $\M$ and $Q$ the dimension of the input space. Recalling \cite[Equation (14)]{laracuente_quasi-factorization_2022}, for a state $\rho$ on $\MM_Q \otimes \MM_Q$,
    \begin{equation} \label{eq:blockform}
        (\E_\M \otimes \Id)(\rho) \cong \bigoplus_\lambda \Big ( \tr_{b_\lambda} ((P_\lambda \otimes \id_Q) \rho (P_\lambda \otimes \id_Q)) \otimes \frac{1}{b_\lambda} \id_{b_\lambda} \Big ) \pl,
    \end{equation}
    where $P_\lambda$ is the projection to the block indexed by $\lambda$, and $\tr_{b_\lambda}$ traces out the $b_\lambda$-dimensional multiplicity subsystem. Let $\E_{\text{bl}}$ and $\E_{\text{tr}}$ denote respectively the conditional expectations to the $\lambda$-labeled block-diagonal subspace, and that which replaces the appropriate subsystem within each block by scalar multiples of the identity. Then the Choi matrix for $\E_\M$ is given by
    \begin{align}
         (\E_\M \otimes \Id) \Big ( \frac{1}{Q} \sum_{i,j=1}^Q \ket{i}^{\otimes 2}\bra{j}^{\otimes 2} \Big )
         & = (\E_{\tr} \otimes \Id) \Big ( \frac{1}{Q}\sum_{\lambda}
             \sum_{i,j = 1}^{a_{\lambda} \times b_{\lambda}} \ket{\lambda, i}^{\otimes 2} \bra{\lambda, j}^{\otimes 2} \Big )
        \\ & = (\E_{\tr} \otimes \Id) \Big (  \sum_{\lambda} \frac{a_{\lambda} b_{\lambda}}{Q} \Big ( 
             \frac{1}{a_{\lambda} b_{\lambda}} \sum_{i,j = 1}^{a_{\lambda} \times b_{\lambda}} \ket{\lambda, i}^{\otimes 2} \bra{\lambda, j}^{\otimes 2} \Big ) \Big )
        \\ & \cong \frac{1}{Q} \sum_{\lambda} \Big ( 
             a_{\lambda} b_{\lambda} \Big ( \frac{1}{a_\lambda} \sum_{i,j = 1}^{a_{\lambda} } \ket{\lambda, i}^{\otimes 2} \bra{\lambda, j}^{\otimes 2} \Big ) \otimes \Big ( \frac{1}{b_\lambda^2} \id_{b_\lambda^2} \Big ) \Big )
        \\ & = \frac{1}{Q} \bigoplus_{\lambda} \Big ( \mathrm{BP}_{a_\lambda} \otimes \frac{a_\lambda}{b_\lambda} \id_{b_\lambda^2} \Big) \pl,
        \label{eq:E-id-BP}
    \end{align}
    where we use the notation $\mathrm{BP}_{a_\lambda}\coloneqq \frac{1}{a_\lambda} \sum_{i,j = 1}^{a_{\lambda} } \ket{\lambda, i}^{\otimes 2} \bra{\lambda, j}^{\otimes 2}$ for a (generalized) Bell pair with local dimension $a_\lambda$. By ``$\cong$" we denote equality up to system swaps that do not change the eigenvalues.
\end{rem}
\begin{lemma} \label{lem:2-to-inf}
    Consider a conditional expectation $\E_0$ to a finite-dimensional von Neumann algebra $\M$ with block decomposition as in Equation \eqref{eq:basicblocks}. Let $\Phi$ and $\Psi$ be quantum channels for which $\E_0 \Phi = \Phi \E_0 = \Phi$, and $\E_0 \Psi = \Psi \E_0 = \Psi$. Let $a_\lambda, b_\lambda$, and $Q$ be defined as in Remark \ref{rem:condexpbp} for $\E_0$.
    Then
    \begin{equation}
        \| (\Phi \otimes \Id - \Psi \otimes \Id)(\mathrm{BP}) \|_\infty \leq 
            \|\Phi - \Psi \|_{2 \rightarrow 2} \max_{\lambda} \frac{1}{b_{\lambda}^2} 
    \end{equation}
\end{lemma}
\begin{proof}
    For every input $\rho$,
    \begin{equation}
        \|\Phi - \Psi \|_{2 \rightarrow 2} \geq \frac{\|(\Phi - \Psi) \E_0(\rho) \|_2}{\|\E_0(\rho)\|_2} \pl.
    \end{equation}
    Using the form from Remark \ref{rem:condexpbp}, and noting that $Q = \sum_\lambda a_\lambda b_\lambda$,
    \begin{align}
        \| (\E_0 \otimes \Id)(\mathrm{BP}) \|_2 = \frac{1}{Q} \sqrt{ \sum_\lambda \frac{a_\lambda^2}{b_\lambda^2} \sum_{i,j=1}^{b_\lambda} 1 }
            = \frac{1}{Q} \sqrt{\sum_\lambda a_\lambda^2}
            = \sqrt{ \frac{\sum_{\lambda} a_\lambda^2}{(\sum_{\lambda'} a_{\lambda'} b_{\lambda'})^2}}
            \leq \max_\lambda \frac{1}{b_\lambda} \pl.
    \end{align}
    Therefore, recalling that the $2 \rightarrow 2$ norm is stable under tensor products \cite{watrous_notes_2004},
    \begin{equation} \label{eq:2norm1}
        \| (\Phi \otimes \Id - \Psi \otimes \Id)(\mathrm{BP}) \|_2 \leq \|\Phi - \Psi \|_{2 \rightarrow 2} \max_\lambda \frac{1}{b_\lambda} \pl.
    \end{equation}
    
    Equation \eqref{eq:2norm1} bounds the output 2-norm. The next step is a bound on the ratio between the output 2-norm and output $\infty$-norm. Since $\Phi = \Phi \E_0$, again recalling Remark \ref{rem:condexpbp},
    \begin{align}
        (\Phi \otimes \Id)(\mathrm{BP})
            & \cong (\Phi \otimes \Id) \Big ( \frac{1}{Q} \bigoplus_{\lambda} \Big ( \mathrm{BP}_{a_\lambda} \otimes \frac{a_\lambda}{b_\lambda} \id_{b_\lambda^2} \Big ) \Big )
            \\ & \cong \frac{1}{Q} \sum_\lambda b_\lambda \sum_{i,j = 1}^{a_\lambda} \Phi \Big ( \ket{\lambda, i} \bra{\lambda, j} \otimes \frac{1}{b_\lambda} \id_{b_\lambda} \Big ) \otimes \Big( \ket{\lambda, i}\bra{\lambda, j} \otimes \frac{1}{b_\lambda} \id_{b_\lambda} \Big).
    \end{align}
    Using Equation \eqref{eq:blockform},
    \begin{align}
        (\Phi \otimes \Id)(\mathrm{BP})
             & \cong  \frac{1}{Q} \sum_\lambda a_\lambda b_\lambda \bigg ( \frac{1}{a_\lambda} \sum_{i,j = 1}^{a_\lambda} (\Phi \otimes \Id_{a_\lambda}) \Big ( \ket{\lambda, i} \bra{\lambda, j} \otimes \frac{1}{b_\lambda} \id_{b_\lambda} \otimes \ket{\lambda, i}\bra{\lambda, j} \Big ) \bigg ) \otimes \frac{1}{b_\lambda} \id_{b_\lambda} .
    \end{align}
    A key observation is that the block is labeled by $\lambda$ on a subspace not directly acted upon by $\Phi$. Therefore, the subspaces for each $\lambda$ remain orthogonal. Again invoking Equation \eqref{eq:blockform} at the output of $\Phi$,
    \begin{align}
            (\Phi \otimes \Id)(\mathrm{BP}) & \cong 
            \frac{1}{Q} \bigoplus_{\lambda} a_\lambda b_\lambda \Big ( \bigoplus_{\lambda'} \omega_{\lambda, \lambda'} \otimes \frac{1}{b_{\lambda'}} \id_{b_{\lambda'}} \Big ) \otimes \frac{1}{b_\lambda} \id_{b_\lambda}
            \\ & = \frac{1}{Q} \bigoplus_{\lambda, \lambda'} a_\lambda b_\lambda \omega_{\lambda, \lambda'} \otimes \frac{1}{b_\lambda b_{\lambda'}} \id_{b_\lambda b_{\lambda'}}
    \end{align}
    for some positive semidefinite matrices $(\omega_{\lambda, \lambda'})$. Via an analogous calculation,
    \begin{equation}
    (\Psi \otimes \Id)(\mathrm{BP}) = \frac{1}{Q} \bigoplus_{\lambda, \lambda'}  a_\lambda b_\lambda \eta_{\lambda, \lambda'} \otimes \id_{b_{\lambda} b_{\lambda'}} / b_{\lambda} b_{\lambda'}
    \label{eq:E-no-structure}
    \end{equation}
    for some positive semidefinite matrices $(\eta_{\lambda, \lambda'})$. 
    Note that for any operator $X$ and dimension factor $c$ we have $\| X \otimes \id_c / c \|_2 = \|X\|_2 / \sqrt{c}$, while $\| X \otimes \id_c / c\|_\infty = \|X\|_\infty / c$.
    \begin{align}
        & \Big \| \frac{1}{Q} \bigoplus_{\lambda, \lambda'} a_\lambda b_\lambda (\omega_{\lambda, \lambda'} - \eta_{\lambda, \lambda'}) \otimes \frac{\id}{b_\lambda b_{\lambda'}} \id_{b_\lambda b_{\lambda'}} \Big \|_2
            \geq \max_{\lambda, \lambda'} \frac{a_\lambda b_\lambda}{Q} \frac{1}{\sqrt{b_\lambda b_{\lambda'}}} \Big \| \omega_{\lambda, \lambda'} - \eta_{\lambda, \lambda'} \Big \|_2 \pl.
    \end{align}
    Meanwhile,
    \begin{align}
        \Big \| \frac{1}{Q} \bigoplus_{\lambda, \lambda'} a_\lambda b_\lambda  (\omega_{\lambda, \lambda'} - \eta_{\lambda, \lambda'}) \otimes \frac{\id}{b_\lambda b_{\lambda'}} \id_{b_\lambda b_{\lambda'}} \Big \|_\infty
        = \max_{\lambda, \lambda'} \frac{a_\lambda b_\lambda}{Q}
            \frac{1}{b_\lambda b_{\lambda'}}
                \Big \| \omega_{\lambda, \lambda'} - \eta_{\lambda, \lambda'} \Big \|_\infty \pl.
    \end{align}
    Together, these infinity and 2-norm bounds imply that
    \begin{equation} \label{eq:2norm2}
        \| (\Phi \otimes \Id - \Psi \otimes \Id)(\mathrm{BP}) \|_\infty \leq \max_{\lambda} \frac{1}{b_\lambda} \| (\Phi \otimes \Id - \Psi \otimes \Id)(\mathrm{BP}) \|_2 \pl.
    \end{equation}
    Applying Equation \eqref{eq:2norm2} to the 2-norm bound of \eqref{eq:2norm1} completes the proof.
\end{proof}
Straightforwardly, one might upper-bounded the $\infty$-norm distance of two channels' outputs from a pure input state by the 2-norm, which in turn is upper-bounded by the $(2 \rightarrow 2)$-norm distance of the channels. Lemma \ref{lem:2-to-inf} tightens this bound when both channels restrict their inputs and outputs to a von Neumann algebra. It first exploits that (1) a density matrix with a large multiplicity must have a 2-norm that is upper-bounded by the inverse square root of that multiplicity (2) the ratio of the infinity norm to the 2-norm is also upper-bounded by the square root of a state's multiplicity.

The following Lemma \ref{lem:inf-to-rel} converts from the operator norm on the Choi matrix to relative error. We note that the concurrent work \cite{schuster_random_2024} had derived a more specific analog of this Lemma \ref{lem:inf-to-rel}. We also note \cite[Proposition II.16]{laracuente_quasi-factorization_2022} as a one-side analog of the bound, and \cite[Lemma 3]{brandao_local_2016} as yielding a similar result up to the polynomial order of $k!$.
\begin{lemma} \label{lem:inf-to-rel}
    Let $\Phi$ and $\E$ be quantum channels for which $\Phi \E = \E \Phi = \E$. Assume that $\E$ is a conditional expectation (idempotent and self-adjoint with respect to the Hilbert-Schmidt inner product). Let $a_\lambda, b_\lambda$, and $Q$ be defined as in Remark \ref{rem:condexpbp} for $\E$. If
    \begin{equation}
         \|(\Phi \otimes \Id)(\mathrm{BP}) - (\E \otimes \Id)(\mathrm{BP}) \|_\infty \times Q \max_\lambda \frac{b_\lambda}{a_\lambda} \leq \epsilon < 1 \pl,
    \end{equation}
    then $(1-\epsilon) \E \prec \Phi \prec (1+\epsilon) \E$.
\end{lemma}
\begin{proof}
    Relative error bounds relating quantum channels are equivalent to relative error bounds on the semidefinite order of the Choi matrices, that is, $\mathcal{N}\prec \mathcal{M}$ for two quantum channels $\mathcal{N},\mathcal{M}$ iff $(\mathcal{M}\otimes \Id)(\mathrm{BP}) - (\mathcal{N}\otimes \Id)(\mathrm{BP})\geq 0$.
    Let $X$ and $Y$ be Hermitian matrices on a space of dimension $Q$. For each $j \in \{1, \dots, Q\}$ and a Hermitian matrix $Z$, let $\beta_j(Z)$ denote the $j$th eigenvalue of $Z$ in descending order. Weyl's inequality \cite{horn_topics_1991} states that
    \begin{equation}
        \beta_Q(Y) \leq \beta_j(X+Y) - \beta_j(X) \leq \beta_1(Y)
    \end{equation}
    for each $j \in \{1, \dots, Q\}$. Also note that $\|Y\|_\infty = \max\{|\beta_1(Y)|, |\beta_Q(Y)|\}$. Therefore,
    \begin{equation} \label{eq:weylcorr}
        |\beta_j(X+Y) - \beta_j(X)| \leq \|Y\|_\infty \pl.
    \end{equation}
    We now identify a condition under which $(1+\epsilon) (\E \otimes \Id)(\mathrm{BP}) - (\Phi \otimes \Id)(\mathrm{BP}) \geq 0$, or equivalently $(1+\epsilon)\E \succ \Phi$. Recalling \eqref{eq:E-id-BP}, we have
    \begin{equation}
        (\E \otimes \Id)(\mathrm{BP}) = \frac{1}{Q} \bigoplus_{\lambda} \mathrm{BP}_{a_\lambda} \otimes \frac{a_\lambda}{b_\lambda} \id_{b_\lambda^2} \pl,
        \label{eq:E-structure}
    \end{equation}
    where $\mathrm{BP}_{a_\lambda}$ denotes a generalized Bell pair on a subspace of dimension $a_\lambda$. Since $\mathrm{BP}_{a_\lambda}$ is a rank 1 vector, each eigenvalue of $(\E \otimes \Id)(\mathrm{BP})$ is either equal to zero or to $a_\lambda / Q b_\lambda$ for some index $\lambda$. 
    
    Setting $X=\epsilon (\E \otimes \Id)(\mathrm{BP})$ and $Y=(\E \otimes \Id)(\mathrm{BP}) - (\Phi \otimes \Id)(\mathrm{BP})$,
    we have $X+Y = (1+\epsilon) (\E \otimes \Id)(\mathrm{BP}) - (\Phi \otimes \Id)(\mathrm{BP})$, and using eq.~\eqref{eq:weylcorr} we get
    \begin{multline}
         \big | \beta_j \big ( \epsilon (\E \otimes \Id)(\mathrm{BP})
            + ((\E \otimes \Id)(\mathrm{BP}) - (\Phi \otimes \Id)(\mathrm{BP})) \big )
            - \beta_j \big ( \epsilon (\E \otimes \Id)(\mathrm{BP}) \big ) \big |
        \\ \leq \| (\E \otimes \Id)(\mathrm{BP}) - (\Phi \otimes \Id)(\mathrm{BP}) \|_\infty \pl,
    \end{multline}
    which we rearrange as
    \begin{align}
        \beta_j \big ( (1+\epsilon) (\E \otimes \Id)(\mathrm{BP}) - (\Phi \otimes \Id)(\mathrm{BP})\big ) \geq \beta_j \big ( \epsilon (\E \otimes \Id)(\mathrm{BP}) \big ) - \| (\E \otimes \Id)(\mathrm{BP}) - (\Phi \otimes \Id)(\mathrm{BP}) \|_\infty.
    \end{align}
    Therefore, $
        (1+\epsilon) (\E \otimes \Id)(\mathrm{BP}) - (\Phi \otimes \Id)(\mathrm{BP})
    $
    has positive eigenvalues as long as
    \begin{equation} \label{eq:diff-versus-mineig}
        \| (\E \otimes \Id)(\mathrm{BP}) - (\Phi \otimes \Id)(\mathrm{BP}) \|_\infty
            \leq \frac{\epsilon}{Q} \min_\lambda \frac{a_\lambda}{b_\lambda} \pl.
    \end{equation}
    Thus, when the condition of Equation \eqref{eq:diff-versus-mineig} is satisfied, $(1+\epsilon) \E \succ \Phi$. Similarly, we aim to show that $\Phi \succ (1-\epsilon) \E$, or equivalently that $\epsilon \E \succ \E - \Phi$. This follows from the same condition on $\epsilon$ as above.
\end{proof}

\begin{lemma} \label{lem:relative-convert}
    Let $A$ be a quantum system with subsystem decomposition $A = A_1 \otimes \dots \otimes A_r$. Denote by $\cT_C$ the $k$-fold twirl, averaging over the unitary group on $C^{\otimes k}$ for a (sub)system $C$. Let the
    von Neumann algebra $\N$ be a subalgebra of the image of the projection $\bigotimes_{n=1}^r \cT_{A_n}$, 
    and let $\Phi$ be any trace-symmetric quantum channel on $A^{\otimes k}$ for which $\N$ is a fixed-point subalgebra. If
    \begin{align} 
    1 > \epsilon \geq \left\| \left(\bigotimes\nolimits_{n=1}^r \cT_{A_n}\right) (\Phi - \E_\N) \left(\bigotimes\nolimits_{n=1}^r \cT_{A_n}\right) \right\|_{2 \rightarrow 2} k!^{2 r-1} e^{k^2 / |A|} \prod_{n=1}^r \Big ( 1 - \frac{k^2}{|A_n|} \Big )^{-2} \pl,
    \end{align}
    and $k^2/|A_n| < 1$ for each $n$, then $(1-\epsilon) \E_\N \prec \Phi \prec (1+\epsilon) \E_\N $.
\end{lemma}
\begin{proof}
    Let $\M_{(n),k}$ denote the von Neumann subalgebra corresponding to the invariant subspace of the $k$-fold Haar unitary twirl $\cT_{A_n}$ on a  system $A_n^{\otimes k}$. Recalling Schur-Weyl duality, $\CC_d^{\otimes k} \cong \bigoplus_{\lambda\vdash_d k} \cV_\lambda^d \otimes \cW_\lambda$, and the block structure \eqref{eq:unitary-commuting-block-structure} of operators commuting with any unitary (which a Haar unitary twirl enforces), the invariant subalgebra of $\bigotimes_{n=1}^r \cT_{A_n}$ has the form
    \begin{align} \bigotimes_{n=1}^r \M_{(n),k} \coloneqq
        \bigotimes_{n=1}^r \left( \bigoplus\nolimits_{\lambda}  \mathcal{S}(\CC^{|\cV_\lambda^{|A_n|}|})  \otimes \BB(\cW_\lambda) \right) \pl,
        \label{eq:product-M-decomposition}
    \end{align}
    where $\mathcal{S}(\mathbb{C}^k) \cong \mathbb{C}$ denotes the algebra of scalar multiples of the identity operator on $\mathbb{C}^k$.
    By the block decomposition product, each block of $\bigotimes_{n=1}^r \M_{(n),k}$ has the product of multiplicities of its constituents. Let $\vec{\lambda} = (\lambda_1, \dots, \lambda_r)$ index each block of $\bigotimes_{n=1}^r \M_{(n),k} $. By \cite[Theorem 1.16]{christandl_structure_2006} or its later restatement as \cite[Lemma 2.11]{metger_simple_2024},
    \begin{align} |\cV_\lambda^d| = \frac{|\cW_{\lambda}|}{k!} \prod_{(i,j) \in \lambda} (d + j - i) \pl, \end{align}
    where $(i, j) \in \lambda$ denotes the cell in row $i$ and column $j$ of the Young diagram $\lambda$. Since $\lambda$ has $k$ cells in total, we have $i, j \leq k$ and thus
    \begin{align} \prod_{(i,j) \in \lambda} (d + j - i) \geq (d - k)^k = d^k (1 - k/d)^k \geq d^k (1 - k^2 / d) \pl, \end{align}
    where the final inequality follows from Bernoulli's inequality provided that $k\leq d$. 
    Conversely,
    \begin{equation}
        \prod_{(i,j) \in \lambda} (d + j - i) = d^k \prod_{(i,j) \in \lambda} \Big ( 1 + \frac{j-i}{d} \Big )
            \leq d^k e^{k^2 / d} \pl.
    \end{equation}
    Therefore, each multiplicity $b_{\vec{\lambda}}$ in $\bigotimes_{n=1}^r \M_{(n),k} $ obeys
    \begin{align} \label{eq:mult1}
    \frac{1}{k!^r} \prod_{n=1}^r |A_n|^{k} \big |\cW_{\lambda_n} \big | e^{k^2 / |A_n|} \geq b_{\vec{\lambda}} \geq \max \bigg \{ \frac{1}{k!^r} \prod_{n=1}^r \bigg ( |A_n|^{k} \big |\cW_{\lambda_n} \big | \Big ( 1 - \frac{k^2}{|A_n|} \Big ) \bigg ) , 1 \bigg \} \pl. \end{align}
    The dimension of the factors corresponding to the $\cW_{\lambda_n}$ indexed by $\vec{\lambda}$ in \eqref{eq:product-M-decomposition} is
    \begin{equation} \label{eq:untraced1}
    a_{\vec{\lambda}} \coloneqq |\cW_{\lambda_1}| \dots |\cW_{\lambda_r}| \pl.
    \end{equation}
    Recall also that $|A|^k = |A_1|^k ... |A_r|^k$.

    We apply Lemma \ref{lem:2-to-inf} with $\E_0 = \bigotimes_{n=1}^r \cT_{A_n}$. In this case, the index $\lambda$ therein is replaced by the vector index $\vec{\lambda}$. Via Equations \eqref{eq:mult1} and \eqref{eq:untraced1}, and with the assumption that $k^2 < |A_n|$ for each $n$, we have
    \begin{equation}
        \max_{\vec{\lambda}} \frac{1}{b_{\vec{\lambda}}^2} \leq \frac{k!^{2 r}}{|A|^{ 2 k} } \,\prod_{n=1}^r \Big ( 1 - \frac{k^2}{|A_n|} \Big )^{-2} \max_{\vec{\lambda}}\frac{1}{a_{\vec{\lambda}}^{2}} \leq \frac{k!^{2 r}}{|A|^{ 2 k} } \prod_{n=1}^r \Big ( 1 - \frac{k^2}{|A_n|} \Big )^{-2} \pl,
        \label{eq:b-lambda-bound}
    \end{equation}
    where we used $a_{\vec{\lambda}}\geq 1$ for all $\vec{\lambda}$ in the second inequality.
    Therefore,
    \begin{multline} \label{eq:infnorm1}
        \left\| \Big ( \left(\bigotimes\nolimits_{n=1}^r \cT_{A_n}\right) (\Phi - \E_\N) \left(\bigotimes\nolimits_{n=1}^r \cT_{A_n}\right) \otimes \Id \Big ) (\mathrm{BP}) \right\|_{\infty}
        \\ \leq \frac{k!^{2 r}}{|A|^{2 k}} \prod_{n=1}^r \Big ( 1 - \frac{k^2}{|A_n|} \Big )^{-2} \left\| \left(\bigotimes\nolimits_{n=1}^r \cT_{A_n}\right) (\Phi - \E_\N) \left(\bigotimes\nolimits_{n=1}^r \cT_{A_n}\right) \right\|_{2 \rightarrow 2} \pl.
    \end{multline}
    We will next apply Lemma \ref{lem:inf-to-rel}. Note that ``$a_\lambda$'' and ``$b_\lambda$'' therein are defined with respect to $\E$, the global twirl, so these are not obviously the same as those defined with respect to the individual pre- and post-twirls. We therefore denote these $\tilde{a}_\lambda$ and $\tilde{b}_\lambda$, noting that instead of $r$ systems undergoing local twirls, we now regard $A$ as one system. By Equations \ref{eq:mult1} and \ref{eq:untraced1} with the assumption that $k^2 / |A| < 1$,
    \begin{equation}
        \frac{|A|^{k}}{k!} \Big ( 1 - \frac{k^2}{|A|} \Big ) \leq
        \frac{\tilde{b}_\lambda}{ \tilde{a}_\lambda}
        \leq \frac{|A|^{k}}{k!} e^{k^2 / |A|} \pl.
    \end{equation}
    Therefore, with $Q=|A|^k$,
    \begin{equation}
    \begin{split}
        & Q \max_\lambda \frac{\tilde{b}_\lambda}{\tilde{a}_\lambda}
            \leq \frac{|A|^{2 k}}{k!} e^{k^2 / |A|} \pl.
    \end{split}
    \end{equation}
    Thus, by Lemma \ref{lem:inf-to-rel}, the desired cp-order inequality holds whenever
    \begin{equation}
        \epsilon \geq k!^{2 r-1} e^{k^2 / |A|} \prod_{n=1}^r \Big ( 1 - \frac{k^2}{|A_n|} \Big )^{-2} \left\| \left(\bigotimes\nolimits_{n=1}^r \cT_{A_n}\right) (\Phi - \E_\N) \left(\bigotimes\nolimits_{n=1}^r \cT_{A_n}\right) \right\|_{2 \rightarrow 2} \pl,
    \end{equation}
	which concludes the proof.
\end{proof}

\section{Analyzing approximate designs using the alternating projection method} \label{sec:alternating-projections}

In this section we prove the $(2\to 2)$-norm bounds for the \hyperref[item:twirl-swap-twirl]{Twirl-Swap-Twirl} and \hyperref[item:twirl-crosstwirl]{Twirl-Crosstwirl} protocols mentioned in Section \ref{sec:main-results}.
These norm bounds are closely related to the subspace angle determining the convergence speed in von Neumann's alternating projection method \cite{neumann_rings_1949}.
The main idea for proving the $(2\to 2)$-norm bounds is to exploit the commutant structure of local twirls implied by Schur-Weyl duality (see Section \ref{sec:schurweyl}), which can be used to expand the images of these twirls as linear combinations of tensor products of permutation operators.
The norm expressions can then be converted into inner products of the coefficient vectors in these expansions, together with matrices whose entries are functions of normalized inner products of permutation operators.
We simplify these expressions using the notion of approximate orthogonality of permutation operators \cite{helsen2023thrifty,harrow_approximate_2023-1} together with various matrix norm and eigenvalue estimates.

\subsection{Preliminaries}
We will use the following notation in this section.
We consider $k$ copies of systems $A$ and $B$ each consisting of $m$ qudits of local dimension $q$.
Each copy of $A$ (or $B$) thus has dimension $q^m$, and the full system $A^kB^k$ has dimension $q^{2mk}$.
We number individual systems $A_1,\dots,A_k$ and $B_1,\dots,B_k$, and sometimes write the full system either as $A_1\dots A_k B_1\dots B_k$ or $A_1B_1\dots A_kB_k$ to highlight the different locality cuts. 
We denote by $\kS_k$ the symmetric group of degree $k$, and we consider its action \eqref{eq:Sk-rep} on a tensor space $X_1\dots X_k$ by permuting tensor factors.
The corresponding permutation operator for $\pi\in\kS_k$ will be denoted by $\pi_X \in \cL(X_1\dots X_k)$.

For $X\in\lbrace A,B,AB\rbrace$ we consider twirling with respect to the representation \eqref{eq:Ud-rep} of the unitary group $\cU(X)$ acting on $X$,
\begin{align}
	\cT_X(\cdot) = \int_{\cU(X)} d U_X\,U_X^{\otimes k}\cdot (U_X^\dagger)^{\otimes k},\label{eq:twirling}
\end{align}
with $d U_X$ denoting the Haar measure on $\cU(X)$.
Such twirls give rise to orthogonal projections on the Hilbert space $\cL(X^{\otimes k})$ of operators acting on $X^{\otimes k}$, equipped with the Hilbert-Schmidt (or Frobenius) inner product $\langle M,N\rangle \coloneqq \tr(M^\dagger N)$.
In the proofs of this section we will repeatedly use the following fact:
\begin{lemma}\label{lem:projections-2-norm-inner-product}
    Let $X$ be a Hilbert space and denote by $\cL(X)$ its Hilbert space of operators equipped with the Hilbert-Schmidt inner product.
    Let $\cP,\cQ,\cR$ be orthogonal projections acting on $\cL(X)$, and assume that
    \begin{align}
        \cP\cR = \cR\cP = \cR = \cR\cQ = \cQ\cR. \label{eq:projections-assumption}
    \end{align}
    Then for any $X\in\cL(X)$ we have
    \begin{align}
        \|\cQ \cP(X) - \cR(X)\|_{2}^2 = \langle X,\cP\cQ\cP(X)-\cR(X)\rangle.
    \end{align}
\end{lemma}
\begin{proof}
    By definition, $\|X\|_2^2 = \langle X,X\rangle$, which we use to expand the $2$-norm as follows:
    \begin{align}
        \|\cQ \cP(X) - \cR(X)\|_{2}^2 &= \langle \cQ \cP(X) - \cR(X), \cQ \cP(X) - \cR(X)\rangle\\
        &= \langle \cQ\cP(X),\cQ\cP(X)\rangle - \langle \cQ\cP(X),\cR(X)\rangle \notag\\
        &\eqspace {} - \langle \cR(X),\cQ\cP(X)\rangle + \langle \cR(X),\cR(X)\rangle. \label{eq:expanded}
    \end{align}
    Because of the self-adjointness of $\cP$, $\cQ$ and $\cR$ with respect to $\langle\cdot,\cdot\rangle$, we have 
    \begin{align}
        \langle \cQ\cP(X),\cQ\cP(X)\rangle &= \langle X,\cP\cQ^2\cP(X)\rangle = \langle X,\cP\cQ\cP(X)\rangle \\
        \langle \cR(X),\cR(X)\rangle &= \langle X,\cR^2(X)\rangle = \langle X,\cR(X)\rangle.
    \end{align}
    Moreover, because of the assumptions in \eqref{eq:projections-assumption}, 
    \begin{align}
        \langle \cQ\cP(X),\cR(X)\rangle &= \langle X,\cP\cQ\cR(X)\rangle = \langle X,\cR(X)\rangle,
    \end{align}
    and similarly $\langle \cR(X),\cQ\cP(X)\rangle = \langle X,\cR(X)\rangle$.
    Using these identities in \eqref{eq:expanded} finishes the proof.
\end{proof}

In the following, we analyze two different protocols.
The first one alternates between individual twirls on $A^k$ and $B^k$ and swapping the first $\ell$ qudits between each of the $A_i$ and $B_i$ blocks.
The bipartite version of the second protocol alternates between individual twirls on $A^k$ and $B^k$, and a twirl across the $A:B$ cut, involving the first $\ell$ qudits of each $A_i$ and $B_i$.
In fact, we prove a more general multipartite version of this ``Twirl-Crosstwirl'' protocol.
For both protocols we show convergence to the full twirl $\cT_{AB}$ and derive bounds on the corresponding convergence rates.

\subsection{Twirl-Swap-Twirl}\label{sec:twirl-swap-twirl}

Here we show that a protocol alternating between a) individual twirls on $A^k$ and $B^k$ and b) swapping the first $\ell$ qudits between each of the $A_i$ and $B_i$ blocks converges to the full twirl on $\cT_{AB}$, thus implementing an approximate $k$-design on $A^kB^k$ that converges to an exact design in the limit.
We will give a bound on the convergence rate of this protocol using the alternating projection method \cite{neumann_rings_1949}.

To this end, we denote by $A_i^\ell$ the first $\ell$ qudits in system $A_i$, and similarly for $B_i^\ell$.
The swap operator $\bF_{A_i^\ell B_i^\ell}$ swaps the first $\ell$ qudits in $A_i$ with the first $\ell$ qudits in $B_i$.
Numbering those qudits $A_{i,1},\dots,A_{i,\ell}$ and $B_{i,1},\dots,B_{i,\ell}$, we have the relation 
\begin{align*}
    \bF_{A_i^\ell B_i^\ell} \equiv \bF_{A_{i,1}B_{i,1}} \dots \bF_{A_{i,\ell}B_{i,\ell}}.
\end{align*}
We now define the following map that swaps the first $\ell$ qudits between $A$ and $B$ in each of the $k$ blocks:
\begin{align}
	\swap_\ell (\cdot) \coloneqq \bigotimes_{i=1}^k \bF_{A_i^\ell B_i^\ell} (\cdot) \bF_{A_i^\ell B_i^\ell}.\label{eq:swapping}
\end{align}
Note that $\swap_\ell^2 = \Id$. 

Consider now the following maps defined in terms of $\swap_\ell$ and the twirls $\cT_X(\cdot)$ for $X\in\lbrace A,B,AB\rbrace$ from \eqref{eq:twirling}:
\begin{align}
	\cP &= \cT_A\otimes \cT_B\label{eq:P-projection}\\
	\cQs &= \swap_\ell\circ \cP \circ\swap_\ell\label{eq:Q-projection}\\
	\cR &= \cT_{AB}.\label{eq:R-projection}
\end{align}
Our protocol consists of a repeated application of $\cP$ followed by $\cQs$.

The maps $\cP,\cQs,\cR$ defined in \eqref{eq:P-projection}--\eqref{eq:R-projection} are orthogonal projections in the space of operators acting on $A^kB^k$.
For example, $\cP$ satisfies $\cP^\dagger = \cP$ with respect to the Frobenius inner product $\langle X,Y\rangle = \tr(X^\dagger Y)$ and $\cP^2=\cP$, and similarly for $\cQs$ and $\cR$.
By left- and right-invariance of the Haar measure, we also have the `dominance relations'
\begin{align}
	\cP \cR = \cR \cP = \cR = \cR \cQs = \cQs \cR,\label{eq:dominance}
\end{align}
so that we will be able to use Lemma \ref{lem:projections-2-norm-inner-product} in the following discusison.

We claim that for $1\leq \ell<m$ the images of the projections in \eqref{eq:P-projection}--\eqref{eq:R-projection} satisfy
\begin{align}
	\im \cP \cap \im \cQs = \im \cR,\label{eq:images}
\end{align}
which implies via von Neumann's alternating projection theorem \cite{neumann_rings_1949} that
\begin{align}
	\lim_{n\to \infty} (\cQs\cP)^n = \cR.\label{eq:von-Neumann}
\end{align}
To see that \eqref{eq:images} is true, note that any operator $\cP(X)$ for $X\in\cL(A^k B^k)$ is invariant under $U_A^{\otimes k}\otimes V_B^{\otimes k}$ for any $U_A\in\cU(A)$ and $V_B\in\cU(B)$.
Schur-Weyl duality in the form \eqref{eq:unitary-invariant-X-permutations} then says that we have $\cP(X) = \sum_{\sigma,\tau\in\kS_k} x_{\sigma,\tau} \sigma_A \otimes \tau_B$ for some coefficients $x_{\sigma,\tau}\in\mathbb{C}$, where $\sigma_A=\varphi(\sigma)_A\in\cL(A^k)$ denotes the permutation operator corresponding to the permutation $\sigma\in\kS_k$ acting on the $k$ copies of the $A$-system, and similarly for $\tau_B$.
On the other hand, the swap operator $\swap_\ell$ acts on a pair of permutation operators as 
\begin{align}
	\swap_\ell(\sigma_A\otimes \tau_B) = \tau_{A^\ell}\otimes \sigma_{A^{m-\ell}}\otimes \sigma_{B^\ell}\otimes \tau_{B^{m-\ell}}. 
\end{align}
The intersection of $\im\cP$ and $\im\cQ$ thus consists of operators that can be written as a linear combination of permutation operators $\pi_{AB} = \pi_A\otimes \pi_B$ for $\pi\in\kS_k$, and this space is equal to $\im\cR$ by another application of \eqref{eq:unitary-invariant-X-permutations}.

Our goal is to obtain a bound on the convergence rate in \eqref{eq:von-Neumann}, for which we follow the exposition in~\cite{deutsch_best_2001}.
The cosine $c(\im\cP,\im\cQs)$ of the \emph{subspace angle} can be expressed as \cite[Lemma 9.5]{deutsch_best_2001}
\begin{align}
	c(\im\cP,\im\cQs) \coloneqq \|\cQs\cP-\cR\|_{2\to 2} = \sup_{X\neq 0}\frac{\left\|\cQs\cP(X)-\cR(X)\right\|_2}{\|X\|_2}.\label{eq:subspace-angle}
\end{align}
Henceforth, we refer to the quantity $c(\cdot,\cdot)$ as \emph{subspace angle cosine}.
The following bound in terms of the subspace angle cosine $c(\im\cP,\im\cQs)$ then holds for all $X\in\cL(A^kB^k)$ and $n\in\mathbb{N}$ \cite[Theorem 9.8]{deutsch_best_2001}:
\begin{align}
	\left\| (\cQs\cP)^n(X) - \cR(X)\right\|_2 \leq c(\im\cP,\im\cQs)^{2n-1} \|X\|_2. \label{eq:alternating-projection-convergence-speed}
\end{align}

Our primary technical result in this section is a bound on the subspace angle cosine of the images of the projections $\cP$ and $\cQ$:
\begin{prop}\label{prop:subspace-angle-swap} 
	Let $A$ and $B$ each consist of $k$ blocks of $m$ qudits of local dimension $q$, and let $\ell$ be the number of qudits swapped between $A$ and $B$ as defined in \eqref{eq:swapping}, where $\ell\leq m/2$ without loss of generality.\footnote{The case $m/2 < \ell < m$ follows from this by a simple relabeling of $A$ and $B$.}
	We assume that $k^2\leq q^m$ and set $a = (1-\eps)^{-2}$ with $\eps = \frac{k^2}{2q^{m}}$.
	Then the subspace angle cosine $c(\im \cP,\im \cQs)$ corresponding to the projections $\cP$ and $\cQs$ defined in \eqref{eq:P-projection} and \eqref{eq:Q-projection}, respectively, satisfies
	\begin{align}
		c(\im \cP,\im \cQs) \leq \sqrt{4 a k^2 q^{-m} + 4a^2k!^3 q^{-2\ell} + 2a^2k!^5 q^{-4\ell} } \pl.
	\end{align}
\end{prop}

The proof of this proposition is given in Section \ref{sec:subspace-angle-swap}.
Here we show how it implies the following main result on approximate unitary designs stated informally in \hyperref[item:twirl-swap-twirl]{Twirl-Swap-Twirl} in Section \ref{sec:main-results}:

\begin{theorem} \label{thm:twirl-swap-twirl-relative}
    Let $A$ and $B$ each consist of $k$ blocks of $m$ qudits of local dimension $q$. Let $\cT_A$ and $\cT_B$ be exact $k$-design channels respectively on $A^{\otimes k}$ and $B^{\otimes k}$. Let $\swap_{\ell}$ denote a unitary that exchanges an arbitrary $\ell$-qubit subsystem of $A$ with one of $B$ in each of the $k$ blocks for $\ell \leq m/2$. Furthermore, assume that $k^2 \leq q^m$. Then $(\cT_A \otimes \cT_B) \swap_{\ell} (\cT_A \otimes \cT_B)$ is an $\epsilon$-approximate multiplicative error $k$-design channel when
    \begin{align}
        \ell \geq \frac{11}{2}\log_q k! + \log_q\left(\frac{1}{\epsilon}\right)-6
        \log_q\left(1-\frac{k^2}{q^m}\right)+\log_q (4 e).
    \end{align}
\end{theorem}

\begin{proof}
By Proposition \ref{prop:subspace-angle-swap},
\begin{align}
\| \cP \cQs - \cT_{AB} \|_{2 \rightarrow 2} &\leq\sqrt{4 a k^2 q^{-m} + 4a^2k!^3 q^{-2\ell} + 2a^2k!^5 q^{-4\ell} } \\
&\leq \sqrt{3\cdot 4 a^2 k!^5 q^{-2\ell}}\\
&\leq \sqrt{12} a k!^{5/2}q^{-\ell},
\end{align}
where we used $a\geq 1$, $k^2\leq k!^5$ for $k\geq 1$, and $2\ell\leq m$.
Using Lemma \ref{lem:relative-convert} with $r=2$, we can therefore choose
\begin{align}
    \epsilon & = \sqrt{12} e^{k^2 / q^{m}} a q^{-\ell}  k!^{11/2} \left(1 - \frac{k^2}{q^m} \right)^{-4} \leq 4 e q^{-\ell}k!^{11/2}  \left(1 - \frac{k^2}{q^m} \right)^{-6},
\end{align}
which gives the bound stated in the theorem after solving for $\ell$.
\end{proof}

\subsection{Twirl-Crosstwirl}\label{sec:multipartite}

\begin{figure}
	\scriptsize
	\centering
	\begin{subfigure}[t]{0.48\textwidth}
		\centering
        \includegraphics[width=\textwidth]{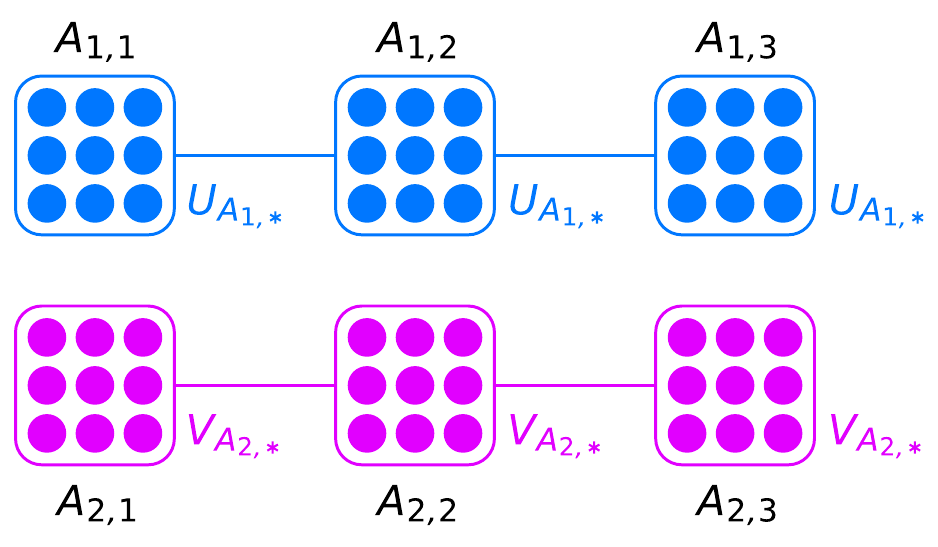}
		\caption{Twirl $\cP$ within rows, consisting of independent twirls $\cT_{A_1}=\int_{\cU(A_{1,*})} dU \, U_{A_{1,*}}^{\otimes 3}(\cdot) (U_{A_{1,*}}^\dagger)^{\otimes 3}$ along the first row (in blue), and $\cT_{A_2}=\int_{\cU(A_{2,*})} dV\, V_{A_{2,*}}^{\otimes 3}(\cdot) (V_{A_{2,*}}^\dagger)^{\otimes 3}$ along the second row (in magenta), as defined in \eqref{eq:multipartite-P-row} and \eqref{eq:multipartite-P}. Both unitaries $U_{A_{1,*}}$ and $V_{A_{2,*}}$ each act on 9 qudits of local dimension $q$ making up the systems $A_{p,k}$ for $p=1,2$ and $k=1,2,3$. }
        \label{fig:row-twirl}
	\end{subfigure}\hfill
	\begin{subfigure}[t]{0.48\textwidth}
		\centering
		\includegraphics[width=\textwidth]{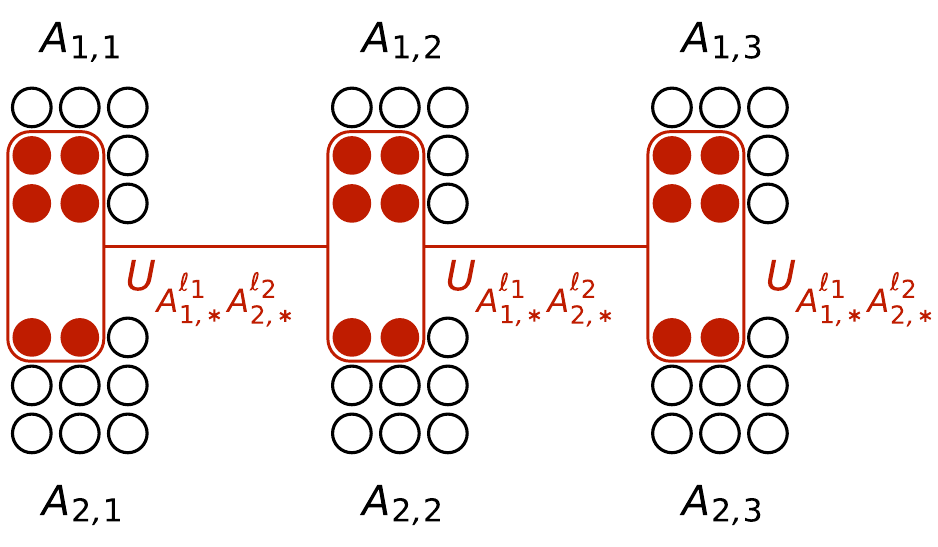}
		\caption{Crosstwirl $\cQm$ across the rows, defined as $\cQm = \int_{\cU\left(A_{1,*}^{\ell_1}A_{2,*}^{\ell_2}\right)}dU\, U_{A_{1,*}^{\ell_1}A_{2,*}^{\ell_2}}^{\otimes 3} (\cdot) (U_{A_{1,*}^{\ell_1}A_{2,*}^{\ell_2}}^\dagger)^{\otimes 3}$. The unitary $U_{A_{1,*}^{\ell_1} A_{2,*}^{\ell_2}}$ acts on the 6 qudits marked in red, comprised of $\ell_1=4$ qudits in each of the first-row blocks, and $\ell_2=2$ qudits in each of the second-row blocks. Note that this unitary acts across the locality cut between rows $A_{1}$ and $A_{2}$, in contrast to the twirls within rows in panel (\textsc{\scriptsize A}).}
		\label{fig:crosstwirl-Q}
	\end{subfigure}
	\caption{Concrete example of the Twirl-Crosstwirl protocol for $P=2$ parties each consisting of $K=3$ copies of 9-qudit blocks $A_{p,k}$. 
    Each circle represents a qudit, and unitaries acting on these qudits are indicated by boxes around them. The first $\ell_1=4$ qudits in the first-row blocks $A_{1,*}$ and the first $\ell_2=2$ qudits in the second-row blocks $A_{2,*}$ participate in the Crosstwirl $\cQm$ in panel (\textsc{\scriptsize B}).}
	\label{fig:multipartite-twirl-crosstwirl-cartoon}
\end{figure}

In this section we analyze the Twirl-Crosstwirl protocol depicted in Fig.~\ref{fig:multipartite-twirl-crosstwirl-cartoon}.
This is a variation of the Twirl-Swap-Twirl protocol from the previous Section \ref{sec:twirl-swap-twirl}, in which the swapping of qudits across the bipartition is replaced by a twirl.
We further generalize the setup beyond the bipartite setting as follows:
For given $P,K\in\mathbb{N}$ we consider a total of $PK$ systems arranged in a $(P\times K)$-grid and numbered $A_{1,1},\dots,A_{1,K}$ in the first row, $A_{2,1},\dots,A_{2,K}$ in the second row, and so on until the last row $A_{P,1},\dots,A_{P,K}$; see Fig.~\ref{fig:multipartite-twirl-crosstwirl-cartoon} for a visualization.\footnote{In Sec.~\ref{sec:twirl-swap-twirl} we only considered bipartite systems ($P=2$) and used the naming scheme $A\equiv A_1$ and $B\equiv A_2$.}
In the $p$-th row each $A_{p,k}$ system for $k=1,\dots,K$ consists of $m_p$ qudits of local dimension $q$, so that $|A_{p,k}|=q^{m_p}$ for $k=1,\dots,K$.
In the grid $(A_{p,k})_{p,k}$ we refer to $p$ as the `row'-index and to $k$ as the `column'-index.
We stress that the rows should be thought of as `local', and any operations along columns (between the different parties $A_1,\dots,A_P$) are expensive.

We now consider a protocol that alternates between independent twirls within the `rows', and twirls of the first $\ell_p \leq m_p$ qudits in each column across the rows.
For the following definitions it will be helpful to consult Fig.~\ref{fig:multipartite-twirl-crosstwirl-cartoon}.
Let us denote one of the identical blocks of $m_p$ qudits (of local dimension $q$) in the $p$-th row of the block array by $A_{p,*}$.
That is, the $p$-th row consists of $K$ identical copies of blocks $A_{p,*}$.
We also denote the first $\ell_p$ qudits of such a block by $A_{p,*}^{\ell_p}$.

Now, the first operation in our alternating protocol twirls the systems within the $P$ rows of blocks.
To this end, we first consider a `local' twirl $\cT_{A_p}$ within the $p$-th row, as depicted in panel ({\scriptsize A}) of Fig.~\ref{fig:multipartite-twirl-crosstwirl-cartoon}:
\begin{align} 
    \cT_{A_p}(\cdot) = \int_{\cU(A_{p,*})} dU\,U_{A_{p,*}}^{\otimes K} (\cdot) (U_{A_{p,*}}^\dagger)^{\otimes K},
    \label{eq:multipartite-P-row}
\end{align}
where $U_{A_{p,*}}\in\cU(A_{p,*})=\cU((\mathbb{C}^q)^{\otimes m_p})$ denotes a unitary acting on the block $A_{p,*}$ of dimension $q^{m_p}$.
Performing this twirl along each row gives a map
\begin{align}
    \cP &= \bigotimes_{p=1}^P \cT_{A_{p}}\label{eq:multipartite-P}.
\end{align}
The second operation twirls the first $\ell_p$ qudits in each $A_{p,*}$ across a column:
\begin{align}
    \cQm &= \cT_{A_{1}^{\ell_1}\dots A_{P}^{\ell_P}} = \int_{\cU\left(A_{1,*}^{\ell_1}\dots A_{P,*}^{\ell_P}\right)} dU\, \left(U_{A_{1,*}^{\ell_1}\dots A_{P,*}^{\ell_P}}\right)^{\otimes K} (\cdot) \left(U_{A_{1,*}^{\ell_1}\dots A_{P,*}^{\ell_P}}^\dagger\right)^{\otimes K}, \label{eq:multipartite-Q}
\end{align}
where now the unitary $U_{A_{1,*}^{\ell_1}\dots A_{P,*}^{\ell_P}}$ acts on the system $A_{1,*}^{\ell_1}\dots A_{P,*}^{\ell_P}$ (the first $\ell_p$ qudits within each $A_{p,*}$) of dimension $q^{\ell_1+\dots+\ell_P}$.
The subscript $\mathrm{mct}$ stands for ``\underline{M}ultipartite \underline{C}ross\underline{T}wirl''.

Finally, we will show in this section that alternating the maps $\cP$ in \eqref{eq:multipartite-P} and $\cQm$ in \eqref{eq:multipartite-Q} converges to the ``full twirl''
\begin{align}
    \cR &= \cT_{A_1\dots A_P} = \int_{\cU\left(A_{1,*}\dots A_{P,*}\right)} \left(U_{A_{1,*}\dots A_{P,*}}\right)^{\otimes K} (\cdot) \left(U_{A_{1,*}\dots A_{P,*}}^\dagger\right)^{\otimes K},
\end{align}
where the unitary $U_{A_{1,*}\dots A_{P,*}}$ is a ``global'' unitary across one of the columns  of blocks (again we refer to Fig.~\ref{fig:multipartite-twirl-crosstwirl-cartoon} for a visualization).

As in the previous sections, we have $\im\cP \cap \im \cQm = \im \cR$, as well as the relations
\begin{align}
	\cP \cR &= \cR \cP = \cR = \cR \cQm = \cQm \cR. \label{eq:dominance-multipartite}
\end{align}
A repeated alternating application of $\cP$ and $\cQm$ thus converges to the full twirl $\cR$,
\begin{align}
	\lim_{n\to\infty} (\cQm \cP)^n = \cR, \label{eq:multipartite-protocol-convergence}
\end{align}
and we prove the following bound on the subspace angle cosine $c(\im\cP,\im\cQm)$ controlling the convergence speed in \eqref{eq:multipartite-protocol-convergence}:

\begin{prop} \label{prop:subspace-angle-multi-crosstwirl}
	Let $K,q,P,\ell_1,\dots,\ell_P,m_1,\dots,m_P$ be as above, and set $L=\ell_1+\dots+\ell_P$ and $M=m_1+\dots+m_P$.
    We assume that $1\leq \ell_p \leq m_p/2$ and $K^2< 2q^{\ell_p}$ for all $p=1,\dots,P$.
	Then the subspace angle cosine $c(\im\cP,\im\cQm)$ of the two projections $\cP$ and $\cQm$ defined in \eqref{eq:multipartite-P} and \eqref{eq:multipartite-Q}, respectively, satisfies
	\begin{align}
		c(\im\cP,\im\cQm)^2 &\leq b_1b_3b_4 q^{-2LK} K!^P \prod_{p=1}^P \binom{q^{2\ell_p}+K-1}{K} - 2b_1b_2 q^{-2MK} K!^P \prod_{p=1}^P \binom{q^{2m_p}+K-1}{K} \notag\\
		& \eqspace{}+ b_1b_2 q^{-2MK} K!^{2P} \prod_{p=1}^P \binom{q^{m_p}+K-1}{K}^2,
	\end{align}
	with constants
	\begin{align}
		\begin{aligned} 
			b_1 &= \prod_{p=1}^P \left(1-\frac{K^2}{2q^{m_p}}\right)^{-1} & \qquad b_2&= \exp\left(-\frac{K^2}{2q^M}\right) \\ b_3 &= \left(1-\frac{K^2}{2q^L}\right)^{-1}
			& b_4&= \prod_{p=1}^P \exp\left(\frac{K^2}{2q^{m_p-\ell_p}}\right).
		\end{aligned}
        \label{eq:constants}
	\end{align}
\end{prop}

\begin{rem}\normalfont
    The assumptions $1\leq \ell_p\leq m_p/2$ and $K^2< 2q^{\ell_p}$ for all $p=1,\dots,P$ ensure that also $K^2 < 2q^{\ell_p}<2q^{m_p}$ as well as $K^2 < 2q^{\ell_p} < 2q^{\ell_1+\dots+\ell_P} = 2q^L \leq 2q^M$.
    As a result, the following quantities appearing in \eqref{eq:constants} are each strictly smaller than $1$ for all $p$,
    \begin{align}
        \max\left\lbrace \frac{K^2}{2q^{m_p}} , \frac{K^2}{2q^M}, \frac{K^2}{2q^L}, \frac{K^2}{2q^{m_p-\ell_p}} \right\rbrace <1,
    \end{align}
    and hence the constants $b_i$ in \eqref{eq:constants} are well-defined and non-trivial.
\end{rem}

The proof of Proposition \ref{prop:subspace-angle-multi-crosstwirl} is given in Section \ref{sec:proof-multipartite}. 
For later use, we rephrase Proposition~\ref{prop:subspace-angle-multi-crosstwirl} as a simplified bound in terms of TPEs:
\begin{cor} \label{cor:twirl-crosstwirl-tpe}
	Let $\cP, \cQm, K,q,P,\ell_1, \dots, \ell_P, m_1, \dots, m_P$ be as in Proposition \ref{prop:subspace-angle-multi-crosstwirl}, and furthermore assume that $2 K^2 \sum_{p=1}^P q^{-m_p} \leq 1$ and that $m_p \geq 3 \ell_p$ for each $p$. 
	Then,
	\begin{align}
		\| \cP \cQm - \cR \|_{2 \rightarrow 2} \leq 5 K \sqrt{\sum_p \frac{1}{q^{2 \ell_p}}} \pl.
	\end{align}
\end{cor}
As noted in the introduction, a direct implication of Corollary \ref{cor:twirl-crosstwirl-tpe} is that, if $\ell_p = \ell$ for each $p$, then $\ell = \log_q K + \log_q(5) + (1/2) \log_q(P) + \log_q(1/\epsilon)$ suffices to obtain TPE error $\epsilon$. We leave the proof of this Corollary to the end of the section, as it primarily involves algebraic manipulation, and instead illustrate its main consequences.
\begin{theorem} \label{thm:twirl-crosstwirl-relative}
	Let $A = \bigotimes_{p=1}^P A_p$. Let $\nu$ be an exact $K$-design on a system defined as the tensor product over $p = 1, \dots, P$ of any $\ell_p$ qudits of local dimension $q$ from each $A_p$. Assume that $10 K^2 \times \sum_p q^{-\ell_p} \leq 1$ and that $m_p \geq 3 \ell_p$ for each $p$. Let each $\mu_p$ be an exact $K$-design on the system $A_p$. Then $\nu * \bigotimes_p \mu_p$ is an $\epsilon$-approximate relative $K$-design provided that
	\begin{align} 
		\epsilon \coloneqq 10 K!^{2 P - 1} K
		\sqrt{ \sum_p \frac{1}{q^{2 \ell_p}} } \pl < 1.
	\end{align}
	Moreover, each unitary in $\nu * \bigotimes_p \mu_p$ requires at most $2 \ell_p$ qudits of quantum communication between $A_p$ and the rest of $A$ to implement, generating at most $2 \ell_p \log_2 q$ ebits of entanglement entropy from a product state input.
\end{theorem}
\begin{proof}
	For each $p = 1,\dots,P$, let $\tilde{A}_p$ denote the subsystem of $A_p$ consisting of the $\ell_p$ qudits affected by $\cQm$, and let $\tilde{A}_p'$ be its complement within $A_p$. The fixed point subalgebra of $\bigotimes_{p=1}^P \cT_{A_p}$ is a subalgebra of that of $(\bigotimes_{p=1}^P \cT_{\tilde{A}_p} \otimes \cT_{\tilde{A}_p'})$. Observe that
	\begin{align}
		\cQm \circ \bigotimes\nolimits_{p=1}^P \cT_{A_p} = \left(\bigotimes\nolimits_{p=1}^P \cT_{\tilde{A}_p} \otimes \cT_{\tilde{A}_p'}\right) \circ \cQm \circ \bigotimes\nolimits_{p=1}^P \cT_{A_p} \circ \left(\bigotimes\nolimits_{p=1}^P \cT_{\tilde{A}_p} \otimes \cT_{\tilde{A}_p'}\right) \pl,
	\end{align}
	since each $\cT_{\tilde{A}_p'}$ commutes with $\cQm$ and is absorbed by $\cQm$, and since each $\cT_{\tilde{A}_p}$ commutes with $\cT_{A_p}$ and is also absorbed by $\cQm$. In the formulation of Lemma \ref{lem:relative-convert}, the channel is effectively pre-processed by local twirls on $2 P$ subsystems, then post-processed by local twirls on $2 P$ subsystems. Lemma \ref{lem:relative-convert} then implies that the ratio of relative error to $2 \rightarrow 2$-norm error is at most
	\begin{align}
		K!^{2 P - 1} e^{K^2 / |A|} \prod_{p=1}^{P} \Big ( 1 - \frac{K^2}{q^{\ell_{p}}} \Big )^{-2} \prod_{p=1}^{P} \Big ( 1 - \frac{K^2}{q^{m_p - \ell_{p}}} \Big )^{-2} \pl.
	\end{align}
    By the assumption that each $m_p \geq 3 \ell_p$, each subsystem has dimension at least $q^{\ell_p}$, so the factor simplifies to
    \begin{align}
		K!^{2 P - 1} e^{K^2 / |A|} \prod_{p=1}^{P} \Big ( 1 - \frac{K^2}{q^{\ell_{p}}} \Big )^{-4} \pl.
	\end{align}
	By the assumptions of the Theorem and Remark \ref{rem:generalbernoulli},
	\begin{align}
        e^{K^2 / |A|} \prod_{p=1}^{P} \Big ( 1 - \frac{K^2}{q^{\ell_{p}}} \Big )^{-4} 
        \leq \frac{e^{K^2 / |A|}}{1 - 4 \sum_p K^2 / q^{\ell_p}}
		\leq \Big ( 1 + \frac{K^2}{|A|} + \frac{K^4}{|A|^2} \Big )(1 + 8 K^2 \sum_p q^{- \ell_p})  \leq 2 \pl.
	\end{align}
	Combining with Corollary \ref{cor:twirl-crosstwirl-tpe} completes the proof.
\end{proof}
We now proceed to prove Corollary \ref{cor:twirl-crosstwirl-tpe}, for which we need the following remark and lemma. 
\begin{rem} \label{rem:generalbernoulli} \normalfont
	For any $n \in \NN$ and $\epsilon_1, \dots, \epsilon_n \in \RR$ such that each $\epsilon_i>-1$, and all $\epsilon_i$ having the same sign, the generalized Bernoulli's inequality \cite[eq.~(7.1)]{Mitrinovic1993} states that
	\[ \prod_{j=1}^n (1 + \epsilon_j) \geq 1 + \sum_{j=1}^n \epsilon_j \]
	Moreoever,
	\[ \prod_{j=1}^n (1 + \epsilon_j) \leq \prod_{j=1}^n e^{\epsilon_j} \leq \exp \Big ( \sum_j \epsilon_j \Big ) \pl. \]
	Furthermore, $e^x \leq 1 + x + x^2$ for every $0 \leq x \leq 1$, and $e^{-x} \leq 1 - x + x^2$ for every $-1 \leq x \leq 1$. Therefore,
	\[ \prod_{j=1}^n (1 + \epsilon_j) \leq 1 + \sum_{j=1}^n \epsilon_j + \Big ( \sum_{j=1}^n \epsilon_j \Big )^2
	\leq 1 + 2 \sum_{j=1}^n \epsilon_j\]
	if $|\sum_{j=1}^n \epsilon_j| \leq 1$. Also note that for any $\delta \in [0, 1/2]$,
	\begin{align} \label{eq:frac-flip}
		\frac{1}{1 - \delta} = 1 + \delta + \frac{\delta^2}{1 - \delta} \leq 1 + 2 \delta \pl.
	\end{align}
\end{rem}
\begin{lemma} \label{lem:cancel}
	Let $q, m, k \in \NN$. Then
	\begin{equation}
		\begin{split}
			1 + \frac{k(k-1)}{2 q^m} \leq q^{- m k} k! \binom{q^m + k - 1}{k} \leq \exp \Big ( \frac{k(k-1)}{2 q^m}  \Big ) \pl,
		\end{split}
	\end{equation}
	and if $k(k-1)/2 q^m \leq 1.7$, then
	\begin{align*} 
        q^{- m k} k! \binom{q^m + k - 1}{k} \leq 1 + \frac{k(k-1)}{2 q^m} + \frac{k^2(k-1)^2}{4 q^{2 m}} \pl. 
    \end{align*}
\end{lemma}
\begin{proof}
	Expanding the binomial,
	\begin{equation}
		\begin{split}
			q^{- m k} k! \binom{q^m + k - 1}{k} = q^{- m k} \frac{(q^m + k - 1)!}{(q^m - 1)!} \pl,
		\end{split}
	\end{equation}
	as the $k!$ factors immediately cancel. Expanding factorials,
	\begin{equation}
		q^{- m k} \frac{(q^m + k - 1)!}{(q^m - 1)!} = q^{- m k} \prod_{j=1}^k (q^m + j - 1)
		= \prod_{j=1}^k \Big ( 1 + \frac{j-1}{q^m} \Big ) \pl.
	\end{equation}
	Remark \ref{rem:generalbernoulli} completes this Lemma along with the fact that $(1+x)^r \leq e^{rx}$ for $x \in \RR, r \geq 0$.
\end{proof}
\begin{proof}[Proof of Corollary \ref{cor:twirl-crosstwirl-tpe}]
	Most of the proof follows from successively applying Remark \ref{rem:generalbernoulli}. Applying Lemma \ref{lem:cancel} to Proposition \ref{prop:subspace-angle-multi-crosstwirl},
	\begin{align}
		\| \cP \cQm - \cR \|_{2 \rightarrow 2}^2
		& \leq b_1 b_3 b_4 \exp \Big ( \sum_p \frac{K^2}{2 q^{2 \ell_p}} \Big )
		- 2 b_1 b_2 \Big ( 1 + \sum_p \frac{K(K-1)}{2q^{2 m_p}} \Big )
		+ b_1 b_2 \exp \Big ( \sum_p \frac{K^2}{q^{m_p}} \Big )
		\\ & \leq b_1 b_3 b_4 \Big (1 + \sum_p \frac{K^2}{q^{2 \ell_p}} \Big )
		+ b_1 b_2 \Big ( 1 + \sum_p \frac{2 K^2}{q^{m_p}} \Big )
		- 2 b_1 b_2 \pl.
	\end{align}
	Using that $1 \geq b_2 \geq 1 - K^2 / 2 q^M$,
	\begin{align}
		\| \cP \cQm - \cR \|_{2 \rightarrow 2}^2 & \leq b_1 \bigg ( b_3 b_4 \Big (1 + \sum_p \frac{K^2}{q^{2\ell_p}} \Big )
		+  \sum_p \frac{2 K^2}{q^{m_p}}
		- 1 + \frac{K^2}{q^M} \bigg ) \pl.
	\end{align}
	Expanding the first term using Remark \ref{rem:generalbernoulli},
	\begin{align}
		 b_3 b_4 \Big (1 + \sum_p \frac{K^2}{q^{2 \ell_p}} \Big )
		& \leq \Big ( 1 + \frac{K^2}{q^L} \Big )
		\Big ( 1 + \sum_p \frac{K^2}{q^{m_p - \ell_p}} \Big )
		\Big ( 1 + \sum_p \frac{K^2}{q^{2 \ell_p}} \Big )
		\\ & \leq 1 + 2 K^2 \Big ( \sum_p \frac{1}{q^{m_p - \ell_p}} + \sum_p \frac{1}{q^{2 \ell_p}} + \frac{1}{q^L} \Big )
		\pl.
	\end{align}
	Therefore,
	\begin{align}
		\| \cP \cQm - \cR \|_{2 \rightarrow 2}^2 & \leq b_1 K^2 \bigg ( 
		2 \sum_p \Big ( \frac{1}{q^{m_p}} + \frac{1}{q^{2 \ell_p}} + \frac{1}{q^{m_p - \ell_p}} \Big ) + \frac{2}{q^L} + \frac{1}{q^M}
		\bigg ) \pl.\label{eq:corollary-intermediate-step}
	\end{align}
	The assumption that $m_p \geq 3 \ell_p$ for each $p$ ensures that each summand in the sum over $p$ in \eqref{eq:corollary-intermediate-step} is bounded from above by $q^{-2\ell_p}$.
    To bound the terms $2q^{-L}$ and $q^{-M}$, we use the following generalized version of Young's inequality: If $a_i>0$ and $0\leq p_i\leq 1$ with $\sum_i p_i = 1$, then $\prod_{i=1}^n a_i^{p_i} \leq \sum_{i=1}^n p_i a_i$.
    Setting $a_p = q^{-\ell_p P}$ and using $P\geq 2$, we then have
    \begin{align}
        q^{-L} = \prod_{p=1}^P q^{-\ell_p} = \prod_{p=1}^P a_p^{1/P} \leq \frac{1}{P} \sum_{p=1}^P a_p = \sum_{p=1}^P q^{-\ell_p P} \leq \sum_{p=1}^P q^{-2\ell_p}.
    \end{align}
    A similar argument shows that also $q^{-M} \leq \sum_{p=1}^P q^{-2m_p} \leq \sum_{p=1}^P q^{-2\ell_p}$.
    Using these observations, \eqref{eq:corollary-intermediate-step} can be bounded as
	\begin{align}
		\| \cP \cQm - \cR \|_{2 \rightarrow 2}^2 & \leq 9 b_1 K^2 \sum_p \frac{1}{q^{2 \ell_p}} \pl. \label{eq:corollary-intermediate-step2}
	\end{align}
	By the Corollary's assumptions, $\sum_{p=1}^P K^2 q^{-m_p}\leq 1/2$, which means in particular that $K^2q^{-m_p}\leq 1/2$ for each $p=1,\dots,P$.
    In that regime, we have $(1-K^2q^{-m_p}/2)^{-1} \leq 1+K^2q^{-m_p}$, and hence
    \begin{align}
        b_1 = \prod_{p=1}^P \left(1-\frac{K^2}{2q^{m_p}}\right)^{-1} \leq \prod_{p=1}^P \left(1+\frac{K^2}{q^{m_p}}\right) \leq 1+2\sum_{p=1}^P \frac{K^2}{q^{m_p}} \leq 2,
    \end{align}
    where we used Remark \ref{rem:generalbernoulli} in the second inequality.
    Using this bound in \eqref{eq:corollary-intermediate-step2} and taking the square root, we finally arrive at
	\begin{align}
		\| \cP \cQm - \cR \|_{2 \rightarrow 2} & \leq 5 K \sqrt{\sum_p \frac{1}{q^{2 \ell_p}}} \pl,
	\end{align}
    which concludes the proof.
\end{proof}

\subsection{Technical proofs}

The rest of this section is dedicated to proving the technical results in Propositions \ref{prop:subspace-angle-swap} and \ref{prop:subspace-angle-multi-crosstwirl}.

\subsubsection{Auxiliary lemmas for rewriting Frobenius inner products}\label{sec:auxiliary-lemma}

We first prove a lemma that is used repeatedly in the proofs of Propositions \ref{prop:subspace-angle-swap} and \ref{prop:subspace-angle-multi-crosstwirl}.
It expresses the inner products of certain operators acting on $A^kB^k$ (appearing in the subspace angle formulas \eqref{eq:norm-as-inner-product} in Section \ref{sec:subspace-angle-swap} and \eqref{eq:norm-to-ip-multipartite} in Section \ref{sec:proof-multipartite} below) in terms of certain matrices acting on vectors in $\mathbb{C}^{k!^2}$.
Part~(\ref{item:X}) of the lemma generalizes \cite[Eq.~(18)]{harrow_approximate_2023-1}.

In the statement and proof of the lemma, we use the following notation.
We denote by 
\begin{align} 
\nip{X}{Y}= (q^{2m})^{-k} \tr(X^\dagger Y)
\end{align}
the normalized Frobenius product on $\cL(A^kB^k)$ (with the subscript `$\mathrm{nrm}$' indicating the normalization), and by $\langle \mathbf{x},\mathbf{y}\rangle$ the standard inner product on $\mathbb{C}^{k!^2}$.
For a permutation $\pi\in\kS_k$, we denote by $\pi_X$ the corresponding permutation operator acting on a tensor product space $X^k$ by permuting systems, as defined in \eqref{eq:Sk-rep}.
We define the $(k!\times k!)$-Gram matrix $\bG_X$ of permutation operators $\lbrace \pi_X : \pi\in\kS_k\rbrace$ with coefficients
\begin{align}
    \left(\bG_X\right)_{\sigma,\pi} = \frac{1}{|X|}\tr(\sigma_X^\dagger \pi_X).
\end{align}
In our setting, the space $X$ typically consists of $m$ of qudits of fixed local dimension $q$.
For a permutation $\pi\in\kS_k$ and an integer $\ell<m$, we denote by $\pi_{X^\ell}$ the permutation operator that permutes the first $\ell$ qudits across the $k$ copies of $X$.

\begin{lemma}\label{lem:inner-product-rewriting}
	Let $k,q,m,\ell$ be as in Proposition \ref{prop:subspace-angle-swap}, and consider the following operators $X,Z\in\cL(A^kB^k)$:
 \begin{align} 
 X &= \sum_{\sigma,\tau\in \kS_k} x_{\sigma,\tau}\, \sigma_A \otimes \tau_B \quad\text{for some $x_{\sigma,\tau}\in\mathbb{C}$,}\\
 Z &= \sum_{\sigma,\tau\in \kS_k} z_{\sigma,\tau}\, \sigma_A \otimes \tau_B \quad\text{for some $z_{\sigma,\tau}\in\mathbb{C}$.}
 \end{align}
	Denote by $\bx,\bz\in\mathbb{C}^{k!^2}$ the vectors with coefficients $(\bx)_{\sigma,\tau} = x_{\sigma,\tau}$ and $(\bz)_{\sigma,\tau} = z_{\sigma,\tau}$ for $\sigma,\tau\in\kS_k$, respectively.
	
	\begin{enumerate}[{\normalfont (i)}]
		\item\label{item:X} $\nip{X}{Z} = \left\langle \bx,(\bG_A\otimes \bG_B)\bz \right\rangle.$
		
		\item\label{item:M} For $k\leq q^m$ we have
		\begin{align}
			\nip{X}{\cP\cQs\cP(X)} = \left\langle \bx, \bM\left(\mathbf{G}_A^{-1}\otimes \mathbf{G}_B^{-1}\right)\bM\bx\right\rangle,
		\end{align}
		where $\bM$ is a Hermitian $(k!^2\times k!^2)$-matrix with coefficients
		\begin{align}
			(\bM)_{(\pi,\rho),(\omega,\chi)} = \nip{\pi_A}{\chi_{A^\ell}\otimes \omega_{A^{m-\ell}}} \nip{\rho_B}{\omega_{B^\ell}\otimes \chi_{B^{m-\ell}}}.\label{eq:M-matrix}
		\end{align}
		
		
		\item\label{item:N} For $k\leq q^{2m}$ we have
		\begin{align}
			\nip{X}{\cR(X)} = \left\langle \bx, \bN^T \bG_{AB}^{-1} \bN \bx \right\rangle,
		\end{align}
		where $\bN$ is a $(k!\times k!^2)$-matrix with coefficients
		\begin{align}
			(\bN)_{\pi,(\sigma,\tau)} = \nip{\pi_{AB}}{\sigma_A\otimes \tau_B}.\label{eq:N-matrix}
		\end{align}
	\end{enumerate}
\end{lemma}

\begin{proof}[Proof of Lemma \ref{lem:inner-product-rewriting}(\ref{item:X})]
	Using the assumptions on the operators $X$ and $Z$, we calculate:
	\begin{align}
		\nip{X}{Z} &= \frac{1}{q^{2mk}} \tr(X^\dagger Z)\\
		&= \frac{1}{q^{2mk}} \sum_{\sigma,\tau,\pi,\rho\in\kS_k} x_{\sigma,\tau}^*\, z_{\pi,\rho} \tr\left((\sigma_A\otimes\tau_B)^\dagger (\pi_A\otimes \rho_B)\right)\\
		&= \frac{1}{q^{2mk}} \sum_{\sigma,\tau,\pi,\rho\in\kS_k} x_{\sigma,\tau}^*\, z_{\pi,\rho} \tr\left( \sigma_A^\dagger \pi_A\right) \tr\left(\tau_B^\dagger \rho_B\right)\\
		&= \sum_{\sigma,\tau\in\kS_k} x_{\sigma,\tau}^*\, \sum_{\pi,\rho\in\kS_k} \left(\bG_A\otimes\bG_B\right)_{(\sigma,\tau),(\pi,\rho)} z_{\pi,\rho}\\
		&= \langle \bx , \left(\bG_A\otimes\bG_B\right)\bz\rangle,
	\end{align}
	where we used the definition $(\bG_A)_{\sigma,\pi} = \frac{1}{q^{mk}} \tr(\sigma_A^\dagger\pi_A)$ for the elements of the Gram matrix $\bG_A$ following \cite{harrow_approximate_2023-1}, and similary for $\bG_B$.
\end{proof}
	
\begin{proof}[Proof of Lemma \ref{lem:inner-product-rewriting}(\ref{item:M})]
To compute the operator $\cP\cQs\cP(X)$, recall that we have $X=\cP(X)$ by assumption, and $\cQs = \swap_\ell\circ\cP\circ\swap_\ell$ with $\swap_\ell$ defined in \eqref{eq:swapping}.
	The action of $\swap_\ell$ on $X$ is equal to
	\begin{align}
		\swap_{\ell}(X) &= \sum_{\sigma,\tau\in\kS_k} x_{\sigma,\tau}\,\swap_{\ell}\left(\sigma_A\otimes \tau_B\right)\\
		&= \sum_{\sigma,\tau\in\kS_k} x_{\sigma,\tau}\, \tau_{A^\ell} \otimes \sigma_{A^{m-\ell}} \otimes \sigma_{B^\ell} \otimes \tau_{B^{m-\ell}},\label{eq:S_ell-action}
	\end{align}
	where we recall that $\tau_{A^\ell}$ denotes the permutation operator corresponding to $\tau\in\kS_k$ acting on the first $\ell\leq m$ qudits within each of the $k$ blocks (of size $m$ each) in $A$.
	Similary, $\sigma_{A^{m-\ell}}$ acts on the remaining $m-\ell$ qudits within each block, and the  same conventions hold for the blocks in $B$.
	
	Another application of the projection $\cP = \cT_A\otimes \cT_B$ to the operator in \eqref{eq:S_ell-action} gives
	\begin{align}
		\cP\circ\swap_{\ell}\circ\cP(X) &= \sum_{\sigma,\tau\in\kS_k} x_{\sigma,\tau}\, \cT_A\left(\tau_{A^\ell} \otimes \sigma_{A^{m-\ell}}\right) \otimes \cT_B\left(\sigma_{B^\ell} \otimes \tau_{B^{m-\ell}}\right)\\
		&= \sum_{\sigma,\tau\in\kS_k} x_{\sigma,\tau}\,\sum_{\omega\in\kS_k} a_\omega^{\sigma,\tau}\,\omega_A \otimes \sum_{\chi\in\kS_k} a_{\chi}^{\tau,\sigma} \,\chi_B,\\
		&= \sum_{\omega,\chi\in\kS_k} \bigg( \sum_{\sigma,\tau\in\kS_k} x_{\sigma,\tau}\, a_\omega^{\sigma,\tau} \,a_{\chi}^{\tau,\sigma} \bigg) \omega_A\otimes \chi_B\\
		&\equiv \sum_{\omega,\chi\in\kS_k} y_{\omega,\chi}\,\omega_A\otimes \chi_B,\label{eq:P-S_ell-action}
	\end{align}
	where the second line uses the fact that the twirl $\cT_A$ projects an operator onto the span of permutation operators $\lbrace\omega_A\colon \omega\in\kS_k\rbrace$, and hence $\cT_A\left(\tau_{A^\ell} \otimes \sigma_{A^{m-\ell}}\right)$ can be expanded as a linear combination of these operators with coefficients $a_{\omega}^{\sigma,\tau}$.
	The same argument applies to $\cT_B\left(\sigma_{B^\ell} \otimes \tau_{B^{m-\ell}}\right)$ which defines the coefficients $a_{\chi}^{\tau,\sigma}$.
	In the last line we defined $y_{\omega,\chi}\coloneqq \sum_{\sigma,\tau\in\kS_k} x_{\sigma,\tau}\, a_\omega^{\sigma,\tau} \, a_{\chi}^{\tau,\sigma}$.
	
	Since $\cP\cQs\cP = \cP\circ\swap_\ell\circ\cP\circ\swap_\ell\circ\cP$, we repeat the above argument by applying $\cP\circ\swap_{\ell}$ once again to the operator in \eqref{eq:P-S_ell-action}, giving
	\begin{align}
		\cP\cQs\cP(X) &= \cP\circ\swap_\ell\left(\cP\circ\swap_\ell\circ\cP(X)\right)\\
		&=\sum_{\omega,\chi\in\kS_k} y_{\omega,\chi}\,\cP\circ\swap_{\ell}\left(\omega_A\otimes \chi_B\right)\\
		&= \sum_{\pi,\rho\in\kS_k} \sum_{\omega,\chi\in\kS_k} y_{\omega,\chi}\, a_{\pi}^{\omega,\chi}\, a_\rho^{\chi,\omega}\, \pi_A\otimes \rho_B\\
		&= \sum_{\pi,\rho\in\kS_k} z_{\pi,\rho}\, \pi_A\otimes \rho_B,
	\end{align}
	where we defined $z_{\pi,\rho} = \sum_{\omega,\chi\in\kS_k} y_{\omega,\chi}\, a_{\pi}^{\omega,\chi}\, a_\rho^{\chi,\omega}$.
    Defining a vector $\bz$ with components $(\bz)_{\pi,\rho} = z_{\pi,\rho}$ for $\pi,\rho\in\kS_k$, part (\ref{item:X}) of the lemma then gives
    \begin{align}
        \nip{X}{\cP\cQs\cP(X)}&= \left\langle \bx, \left(\bG_A\otimes \bG_B\right)\bz \right\rangle. \label{eq:inner-product-translation}
    \end{align}
	
	We now express the vector $\bz$ with coefficients $(\bz)_{\pi,\rho}= z_{\pi,\rho} = \sum_{\omega,\chi\in\kS_k} y_{\omega,\chi}\, a_{\pi}^{\omega,\chi}\, a_{\rho}^{\chi,\omega}$ in terms of $\bx$.
	Using $y_{\omega,\chi}\coloneqq \sum_{\sigma,\tau\in\kS_k} x_{\sigma,\tau}\, a_\omega^{\sigma,\tau} \, a_{\chi}^{\tau,\sigma}$ from \eqref{eq:P-S_ell-action}, the coefficients $z_{\pi,\rho}$ are given by
	\begin{align}
		z_{\pi,\rho} &= \sum_{\omega,\chi\in\kS_k} y_{\omega,\chi}\, a_{\pi}^{\omega,\chi}\, a_{\rho}^{\chi,\omega}\\
		&= \sum_{\omega,\chi,\sigma,\tau\in\kS_k} x_{\sigma,\tau}\, a_\omega^{\sigma,\tau} \, a_{\chi}^{\tau,\sigma} \, a_{\pi}^{\omega,\chi}\, a_{\rho}^{\chi,\omega}\\
		&= \sum_{\sigma,\tau\in\kS_k} \bigg( \sum_{\omega,\chi\in \kS_k} a_\omega^{\sigma,\tau} \, a_{\chi}^{\tau,\sigma} \, a_{\pi}^{\omega,\chi}\, a_{\rho}^{\chi,\omega} \bigg) x_{\sigma,\tau}\label{eq:z-in-terms-of-a-and-x}
	\end{align}
	Our goal is to interpret the sum in parentheses as a $(k!^2\times k!^2)$-matrix with coefficients indexed by pairs of permutations $((\pi,\rho),(\sigma,\tau))$.
	Those coefficients can in turn be expressed in terms of the matrix $\bM$ in the statement of the lemma and the Gram matrices $\bG$.
	
	To this end, recall from the calculation leading to \eqref{eq:P-S_ell-action} that the $a_{\pi}^{\omega,\chi}$ are defined via
	\begin{align}
		\cT_A(\chi_{A^\ell}\otimes \omega_{A^{m-\ell}}) = \sum_{\pi\in\kS_k} a_{\pi}^{\omega,\chi} \pi_A.
	\end{align}
	We take the inner product with $\pi'_A$ on both sides.
	For the left-hand side, we have
	\begin{align}
		\nip{\pi'_A}{\cT_A(\chi_{A^\ell}\otimes \omega_{A^{m-\ell}})} = \nip{\cT_A(\pi'_A)}{\chi_{A^\ell}\otimes \omega_{A^{m-\ell}}} = \nip{\pi'_A}{\chi_{A^\ell}\otimes \omega_{A^{m-\ell}}} \eqqcolon v_{\pi'}^{\omega,\chi},
	\end{align}
	since $\pi'_A$ is invariant under $\cT_A$.
	Hence, 
	\begin{align}
		v_{\pi'}^{\omega,\chi} = \nip{\pi'_A}{\cT_A(\chi_{A^\ell}\otimes \omega_{A^{m-\ell}})} = \sum_{\pi\in\kS_k} a_\pi^{\omega,\chi} \nip{\pi'_A}{\pi_A} = \sum_{\pi\in\kS_k} (\mathbf{G}_A)_{\pi',\pi} a_\pi^{\omega,\chi},
	\end{align}
	or equivalently $\mathbf{v}^{\omega,\chi} = \mathbf{G}_A \mathbf{a}^{\omega,\chi}$ where $(\mathbf{v}^{\omega,\chi})_\pi = v_{\pi}^{\omega,\chi}$ and $(\mathbf{a}^{\omega,\chi})_\pi = a_\pi^{\omega,\chi}$, and the same relation holds with $\bG_B$.\footnote{Since $|A|=|B|$ it is not strictly necessary to distinguish between $\bG_A$ and $\bG_B$, but it does help in parsing the somewhat cumbersome expressions appearing in the sums.}
	
	The Gram matrices $\bG_A$ and $\bG_B$ are invertible if and only if $k\leq q^m$ \cite[Lemma 1]{harrow_approximate_2023-1}, in which case
	$\mathbf{a}^{\omega,\chi} = \bG_A^{-1} \mathbf{v}^{\omega,\chi}$.
	Recall that we try to interpret the sum in parentheses in \eqref{eq:z-in-terms-of-a-and-x} as a matrix with coefficients indexed by $((\pi,\rho),(\sigma,\tau))$.
	Using the formula for $\mathbf{a}^{\omega,\chi}$, we compute:
	\begin{align}
		&\sum_{\omega,\chi\in \kS_k} a_\omega^{\sigma,\tau} \, a_{\chi}^{\tau,\sigma} \, a_{\pi}^{\omega,\chi}\, a_{\rho}^{\chi,\omega}\\
		&=
		\sum_{\omega,\chi} (\bG_A^{-1} \mathbf{v}^{\sigma,\tau})_\omega\, (\bG_B^{-1} \mathbf{v}^{\tau,\sigma})_\chi \, (\bG_A^{-1} \mathbf{v}^{\omega,\chi})_\pi \, (\bG_B^{-1} \mathbf{v}^{\chi,\omega})_\rho\\
		&= \sum_{\omega,\chi} \sum_{\omega',\chi'\pi',\rho'} (\bG_A^{-1})_{\omega,\omega'} v^{\sigma,\tau}_{\omega'} (\bG_B^{-1})_{\chi,\chi'} v^{\tau,\sigma}_{\chi'} (\bG_A^{-1})_{\pi,\pi'} v^{\omega,\chi}_{\pi'} (\bG_B^{-1})_{\rho,\rho'} v^{\chi,\omega}_{\rho'}\\
		&= \sum_{\pi',\rho'} (\bG_A^{-1})_{\pi,\pi'} (\bG_B^{-1})_{\rho,\rho'} \sum_{\omega,\chi} v^{\omega,\chi}_{\pi'} v^{\chi,\omega}_{\rho'} \sum_{\omega',\chi'} (\bG_A^{-1})_{\omega,\omega'}(\bG_B^{-1})_{\chi,\chi'} v^{\sigma,\tau}_{\omega'}  v^{\tau,\sigma}_{\chi'}\\
		&= \sum_{\pi',\rho'} (\bG_A^{-1}\otimes \bG_B^{-1})_{(\pi,\rho),(\pi',\rho')}  \sum_{\omega,\chi} v^{\omega,\chi}_{\pi'} v^{\chi,\omega}_{\rho'} \sum_{\omega',\chi'} (\bG_A^{-1}\otimes \bG_B^{-1})_{(\omega,\chi),(\omega',\chi')} v^{\sigma,\tau}_{\omega'}  v^{\tau,\sigma}_{\chi'}.
	\end{align}
	Defining a matrix $\bM$ with coefficients $(\bM)_{(\pi,\rho),(\omega,\chi)} = v_\pi^{\omega,\chi} v_\rho^{\chi,\omega}$, we thus see that
	\begin{align}
		\sum_{\omega,\chi\in \kS_k} a_\omega^{\sigma,\tau} \, a_{\chi}^{\tau,\sigma} \, a_{\pi}^{\omega,\chi}\, a_{\rho}^{\chi,\omega} = \left(
		(\bG_A^{-1}\otimes \bG_B^{-1}) \bM (\bG_A^{-1}\otimes \bG_B^{-1}) \bM \right)_{(\pi,\rho),(\sigma,\tau)},
	\end{align}
	and together with \eqref{eq:z-in-terms-of-a-and-x} we obtain the vector equation
	\begin{align}
		\mathbf{z} = (\bG_A^{-1}\otimes \bG_B^{-1}) \bM (\bG_A^{-1}\otimes \bG_B^{-1}) \bM \mathbf{x}.
	\end{align}
	Using this in \eqref{eq:inner-product-translation} finally gives 
	\begin{align}
		\nip{X}{\cP\cQs\cP(X)} = \left\langle \bx, \left(\bG_A\otimes \bG_B\right)\bz \right\rangle = \left\langle \bx ,\bM (\bG_A^{-1}\otimes \bG_B^{-1}) \bM \mathbf{x} \right\rangle,
	\end{align} 
	which concludes the proof of Lemma \ref{lem:inner-product-rewriting}(\ref{item:M}).
\end{proof}

\begin{proof}[Proof of Lemma \ref{lem:inner-product-rewriting}(\ref{item:N})]
	Recall once more that we may assume $X=\cP(X) = \sum_{\sigma,\tau\in\kS_k} x_{\sigma,\tau} \sigma_A\otimes \tau_B$ for some $x_{\sigma,\tau}\in\mathbb{C}$.
	For the full twirl $\cR = \cT_{AB}$, we then have
	\begin{align}
		\cR(X) &= \sum_{\sigma,\tau} x_{\sigma,\tau} \cR(\sigma_A\otimes \tau_B) = \sum_{\sigma,\tau} x_{\sigma,\tau} \sum_{\pi} r_\pi^{\sigma,\tau} \pi_{AB}\label{eq:twirl-r-coefficients}
	\end{align}
	for some coefficients $r_\pi^{\sigma,\tau}\in\mathbb{C}$, and 
	\begin{align}
		\nip{X}{\cR(X)} &= \sum_{\sigma,\tau,\sigma',\tau',\pi} x^*_{\sigma',\tau'} x_{\sigma,\tau} r_\pi^{\sigma,\tau} \nip{\sigma'_A\otimes \tau'_B}{\pi_{AB}}\\
		&= \sum_{\sigma,\tau,\sigma',\tau',\pi} x^*_{\sigma',\tau'} x_{\sigma,\tau} r_\pi^{\sigma,\tau} t^{\sigma',\tau'}_\pi,\label{eq:XRX}
	\end{align}
	where we defined $t_\pi^{\sigma,\tau} \coloneqq \nip{ \pi_{AB}}{\sigma_A\otimes\tau_B }$.
	
	Note that $\cT_{AB}(\pi_{AB}) = \pi_{AB}$ for any $\pi\in\kS_k$, and hence
	\begin{align}
		t_\pi^{\sigma,\tau} = \nip{\pi_{AB}}{\sigma_A\otimes\tau_B} = \nip{ \cT(\pi_{AB})}{\sigma_A\otimes \tau_B} = \nip{ \pi_{AB}}{\cT_{AB}(\sigma_A\otimes \tau_B)}.
	\end{align}
	Thus, we may express $r_\pi^{\sigma,\tau}$ in terms of $t_\pi^{\sigma,\tau}$ as
	\begin{align}
		t_\pi^{\sigma,\tau} &= \langle\pi_{AB}, \cT_{AB}(\sigma_A\otimes \tau_B)\rangle\\
		&= \sum_{\pi'} r_{\pi'}^{\sigma,\tau} \langle \pi_{AB},\pi'_{AB}\rangle\\
		&= \sum_{\pi'} (\bG_{AB})_{\pi,\pi'} \, r_{\pi'}^{\sigma,\tau}. \label{eq:t-in-terms-of-r}
	\end{align}
	Defining vectors $\mathbf{t}^{\sigma,\tau},\mathbf{r}^{\sigma,\tau}\in\mathbb{C}^{k!}$ with $(\mathbf{t}^{\sigma,\tau})_{\pi} = t_\pi^{\sigma,\tau}$ and $(\mathbf{r}^{\sigma,\tau})_\pi = r_{\pi'}^{\sigma,\tau}$, the identity \eqref{eq:t-in-terms-of-r} can be rewritten as the vector equation $\mathbf{t}^{\sigma,\tau} = \bG_{AB} \mathbf{r}^{\sigma,\tau}$, which is equivalent to $\mathbf{r}^{\sigma,\tau} = \bG_{AB}^{-1} \mathbf{t}^{\sigma,\tau}$ in the regime $k\leq q^{2m}$ such that $\bG_{AB}$ is invertible \cite[Lemma 1]{harrow_approximate_2023-1}.
	
	Using $\mathbf{r}^{\sigma,\tau} = \bG_{AB}^{-1} \mathbf{t}^{\sigma,\tau}$ in \eqref{eq:XRX}, we obtain
	\begin{align}
		\nip{X}{\cR(X)} &= \sum_{\sigma,\tau,\sigma',\tau',\pi} x^*_{\sigma',\tau'} x_{\sigma,\tau} r_\pi^{\sigma,\tau} t^{\sigma',\tau'}_\pi\\
		&= \sum_{\sigma,\tau,\sigma',\tau',\pi} x^*_{\sigma',\tau'} x_{\sigma,\tau} (\bG_{AB}^{-1} \mathbf{t}^{\sigma,\tau})_{\pi} \,  t^{\sigma',\tau'}_\pi\\
		&= \sum_{\sigma,\tau,\sigma',\tau',\pi,\pi'} x^*_{\sigma',\tau'} x_{\sigma,\tau} (\bG_{AB}^{-1})_{\pi,\pi'} t^{\sigma,\tau}_{\pi'} \,  t^{\sigma',\tau'}_\pi\\
		&= \sum_{\sigma',\tau'} x^*_{\sigma',\tau'} \sum_\pi t_{\pi}^{\sigma',\tau'} \sum_{\pi'} (\bG_{AB}^{-1})_{\pi,\pi'} \sum_{\sigma,\tau} t^{\sigma,\tau}_{\pi'} x_{\sigma,\tau}.
	\end{align}
	Defining the $(k!\times k!^2)$-matrix $\bN$ with coefficients $(\bN)_{\pi,(\sigma,\tau)} = t_\pi^{\sigma,\tau} = \nip{\pi_{AB}}{\sigma_A\otimes\tau_B}$, the last equation can be written as 
	\begin{align}
		\nip{X}{\cR(X)} = \langle \bx, \bN^T \bG_{AB}^{-1} \bN \bx\rangle,
	\end{align}
	which concludes the proof of Lemma \ref{lem:inner-product-rewriting}(\ref{item:N}).
\end{proof}

The following auxiliary lemma, a multipartite version of the previous Lemma~\ref{lem:inner-product-rewriting}, is used in the proof of Proposition \ref{prop:subspace-angle-multi-crosstwirl}.
It is again concerned with rewriting normalized Frobenius inner products as standard inner products on $\mathbb{C}^{K!^P}$ using Gram matrices of permutation operators.
We use similar notation as in Lemma~\ref{lem:inner-product-rewriting}.

\begin{lemma}\label{lem:frobenius-as-ip-multipartite}~
	Let $K,q,P,\ell_1,\dots,\ell_P,m_1,\dots,m_P$ be as in Proposition \ref{prop:subspace-angle-multi-crosstwirl}, and consider operators $X,Z\in\cL(A_1^K\dots A_P^K)$ of the form
	\begin{align} 
		X &= \sum_{\sigma^1,\dots,\sigma^P\in\kS_k} x_{\sigma^1,\dots,\sigma^P}\,\sigma^1_{A_1}\otimes \dots\otimes \sigma^P_{A_P} \quad \text{for some $x_{\sigma^1,\dots,\sigma^P}\in\mathbb{C}$,}\\
		Z &= \sum_{\sigma^1,\dots,\sigma^P\in\kS_k} z_{\sigma^1,\dots,\sigma^P}\,\sigma^1_{A_1}\otimes \dots\otimes \sigma^P_{A_P} \quad \text{for some $z_{\sigma^1,\dots,\sigma^P}\in\mathbb{C}$.}
	\end{align}
	Denote by $\bx,\bz\in \mathbb{C}^{K!^P}$ the vectors with coefficients $(\bx)_{\sigma^1,\dots,\sigma^P} = x_{\sigma^1,\dots,\sigma^P}$ and $(\bz)_{\sigma^1,\dots,\sigma^P} = z_{\sigma^1,\dots,\sigma^P}$, respectively.
	
	\begin{enumerate}[{\normalfont (i)}]
		\item\label{item:X-multipartite} $\nip{X}{Z} = \left\langle \bx, \left(\bG_{A_1}\otimes \dots\otimes \bG_{A_P}\right) \bz \right\rangle$
		
		\item\label{item:N_ell-multipartite} If $K\leq |A_{p,k}|=q^{m_p}$ for all $p=1,\dots,P$ and $K\leq |A_{1}^{\ell_1}\dots A_{p}^{\ell_p}| = q^{\ell_1 + \dots + \ell_P}$, 
		\begin{align}
			\nip {X}{\cP\cQm\cP(X)} = \left\langle \bx, \left(\bG_{A_1^{m_1-\ell_1}}\otimes \dots\otimes \bG_{A_P^{m_P-\ell_P}}\right) \odot \left(\bN_{\ell_1,\dots,\ell_P}^T \bG_{\ell_1,\dots,\ell_P}^{-1} \bN_{\ell_1,\dots,\ell_P} \right) \bx \right\rangle,
		\end{align}
		where $\bG_{\ell_1,\dots,\ell_P}\equiv \bG_{A_{1}^{\ell_1}\dots A_{P}^{\ell_P}}$ is the Gram matrix of permutation operators acting on the systems $A_{1}^{\ell_1}\dots A_{P}^{\ell_P}$, and $\bN_{\ell_1,\dots,\ell_P}$ is a $(K!\times K!^P)$-matrix with coefficients
		\begin{align}
			(\bN_{\ell_1,\dots,\ell_P})_{\tau,(\sigma^1,\dots,\sigma^P)} = \nip{\tau_{A_1^{\ell_1}\dots A_P^{\ell_P}}}{\sigma^1_{A_1^{\ell_1}}\otimes \dots\otimes \sigma^P_{A_P^{\ell_P}}}.
		\end{align}
		
		\item\label{item:N-multipartite} 
		If $K\leq |A_1\dots A_P| = q^{m_1+\dots + m_P}$, then
		\begin{align}
			\nip{X}{\cR(X)} = \left\langle \bx, \bN_{A_1\dots A_P}^T \bG_{A_1\dots A_P}^{-1} \bN_{A_1\dots A_P} \bx\right\rangle,
		\end{align}
		where $\bN_{A_1\dots A_P}$ is a $(K!\times K!^P)$-matrix with coefficients
		\begin{align}
			(\bN_{A_1\dots A_P})_{\tau,(\sigma^1,\dots,\sigma^P)} = \nip{\tau_{A_1\dots A_P}} {\sigma^1_{A_1}\otimes \dots\otimes \sigma^P_{A_P}}.
		\end{align}
	\end{enumerate}
\end{lemma}

For the proof of Lemma \ref{lem:frobenius-as-ip-multipartite} we introduce the following shorthand notation to streamline the presentation.
For $\sigma\in\kS_K$ and $p=1,\dots,P$, we write
\begin{align}
	\sigma_p &\equiv \sigma_{A_p} &	\sigma_{\ell_p} &\equiv \sigma_{A_p^{\ell_p}} & \sigma_{m_p-\ell_p} &\equiv \sigma_{A_p^{m_p-\ell_p}}.
\end{align}
In the proofs below we will index different permutations as $\sigma^p$ for $p=1,\dots,P$.
For the corresponding permutation operators $\sigma^p_{A_p}$ or $\sigma^p_{A_p^{\ell_p}}$ we then write $\sigma^p_p$ and $\sigma^p_{\ell_p}$, respectively.
Note that in $\sigma^p_p$ the superscript $p$ indexes the permutation, while the subscript $p$ indicates the $p$-th collection of $A$-systems (or `row', as visualized in Fig.~\ref{fig:multipartite-twirl-crosstwirl-cartoon}) on which $\sigma^p_p$ acts. 
For multipartite permutation operators we often skip system labels too, e.g.,
\begin{align}
	\sigma_{A_1\dots A_P} &\equiv \sigma_{1\dots P} & \sigma_{A_1^{\ell_1}\dots A_P^{\ell_P}} &\equiv \sigma_{\ell_1,\dots,\ell_P} & \sigma_{A_1^{m_1-\ell_1}\dots A_P^{m_P-\ell_P}} &\equiv \sigma_{m_1-\ell_1,\dots,m_P-\ell_P}
\end{align}
We also write $\vec{\sigma}=(\sigma^1,\dots,\sigma^P)$ for a tuple of permutations $\sigma^p\in\kS_K$ with $p=1,\dots,P$.

\begin{proof}[Proof of Lemma \ref{lem:frobenius-as-ip-multipartite}(\ref{item:X-multipartite})]
	This is a straightforward generalization of the proof of Lem.~\ref{lem:inner-product-rewriting}(\ref{item:X}).
	With the assumptions on $X$ and $Z$, we calculate:
	\begin{align}
		\nip{X}{Z} &= \sum_{\vomega,\vsigma} x^*_{\vomega}\, z_{\vsigma}\, \nip{\omega^1_1 \otimes \dots \otimes \omega^P_P} {\sigma^1_1 \otimes \dots \otimes \sigma^P_P }\\
		&= \sum_{\vomega,\vsigma} x^*_{\vomega}\, z_{\vsigma}\, \nip{\omega^1_1}{\sigma^1_1} \dots \nip{\omega^P_P}{\sigma^P_P}\\
		&= \sum_{\vomega,\vsigma} x^*_{\vomega}\, z_{\vsigma}\, \left(\bG_{A_1}\right)_{\omega^1,\sigma^1} \dots \left(\bG_{A_P}\right)_{\omega^P,\sigma^P}\\
		&= \sum_{\vomega} x^*_{\vomega}\, \sum_{\vsigma} \left(\bG_{A_1} \otimes \dots \otimes \bG_{A_P} \right)_{\vomega,\vsigma} z_{\vsigma}\\
		&= \langle \bx,\left(\bG_{A_1} \otimes \dots \otimes \bG_{A_P} \right) \bz\rangle,
	\end{align}
	which concludes the proof.
\end{proof}

\begin{proof}[Proof of Lemma \ref{lem:frobenius-as-ip-multipartite}(\ref{item:N_ell-multipartite})]
	Recall that we have
	\begin{align}
		X = \cP(X) = \sum_{\vec{\sigma}} x_{\vec{\sigma}}\, \sigma^1_{1}\otimes \dots\otimes \sigma^P_{P}\label{eq:P(X)-proof}
	\end{align}
	for some coefficients $x_{\vec{\sigma}}\in\mathbb{C}$.
	Applying the crosstwirl $\cQm$ to an operator $\sigma^1_{{\ell_1}}\otimes \dots\otimes \sigma^P_{{\ell_P}}$ gives
	\begin{align}
		\cQm\left(\sigma^1_{{\ell_1}}\otimes \dots\otimes \sigma^P_{{\ell_P}}\right) &= \sum_\tau y_\tau^{\vec{\sigma}}\, \tau_{{\ell_1},\dots,{\ell_P}} =  \sum_\tau y_\tau^{\vec{\sigma}}\, \tau_{{\ell_1}} \otimes \dots \otimes \tau_{{\ell_P}} \label{eq:y-multipartite-def}
	\end{align}
	for some coefficients $y_\tau^{\vec{\sigma}}$, and using this relation in \eqref{eq:P(X)-proof} gives
	\begin{align}
		\cQm\cP(X) &= \sum_{\vec{\sigma}} x_{\vec{\sigma}}\, \cQm\left(\sigma^1_{\ell_1}\otimes \dots\otimes \sigma^P_{\ell_P}\right) \otimes \sigma^1_{m_1-\ell_1} \otimes \dots \sigma^P_{m_P-\ell_P}\\
		&=\sum_{\vec{\sigma}} x_{\vec{\sigma}} \sum_\tau y_\tau^{\vec{\sigma}}\, \tau_{\ell_1} \otimes \sigma^1_{m_1-\ell_1} \otimes \dots \otimes \tau_{\ell_P} \otimes \sigma^P_{m_P-\ell_P}.
	\end{align}
	Another application of $\cP$ to this operator yields
	\begin{align}
		\cP\cQm\cP(X) &= \sum_{\vec{\sigma}} x_{\vec{\sigma}} \sum_\tau y_\tau^{\vec{\sigma}}\, \cP\left(\tau_{\ell_1} \otimes \sigma^1_{m_1-\ell_1} \otimes \dots \otimes \tau_{\ell_P} \otimes \sigma^P_{m_P-\ell_P}\right)\\
		&= \sum_{\vec{\sigma}} x_{\vec{\sigma}} \sum_\tau y_\tau^{\vec{\sigma}}\, \cT_{A_1}\left(\tau_{\ell_1} \otimes \sigma^1_{m_1-\ell_1}\right) \otimes \dots \otimes \cT_{A_P}\left(\tau_{\ell_P} \otimes \sigma^P_{m_P-\ell_P}\right)\\
		&= \sum_{\vec{\sigma}} x_{\vec{\sigma}} \sum_\tau y_\tau^{\vec{\sigma}} \sum_{\vec{\omega}} z_{\omega^1}^{\tau,\sigma^1}\dots z_{\omega^P}^{\tau,\sigma^P} \omega^1_{1} \otimes \dots \otimes \omega^P_{P}\\
		&= \sum_{\vec{\omega}} \left(\sum_{\vec{\sigma}} x_{\vec{\sigma}} \sum_\tau y_\tau^{\vec{\sigma}}\,  z_{\omega^1}^{\tau,\sigma^1}\dots z_{\omega^P}^{\tau,\sigma^P}\right) \omega^1_{1} \otimes \dots \otimes \omega^P_{P},
	\end{align}
	with coefficients $z_*^{*,*}$ defined for $p=1,\dots, P$ via
	\begin{align}
		\cT_{A_p}\left(\tau_{\ell_p} \otimes \sigma_{m_p-\ell_p}^p \right) = \sum_{\omega^p} z_{\omega^p}^{\tau,\sigma^p} \omega^p_{p}. \label{eq:z}
	\end{align}
	
	Defining vectors $\bs,\bx\in\mathbb{C}^{K!^p}$ via
	\begin{align}
		(\bx)_{\vec{\sigma}} &= x_{\vec{\sigma}} &
		(\bs)_{\vec{\omega}} &= \sum_{\vec{\sigma}} x_{\vec{\sigma}} \sum_\tau y_\tau^{\vec{\sigma}}  z_{\omega^1}^{\tau,\sigma^1}\dots z_{\omega^P}^{\tau,\sigma^P},
	\end{align}
	part (\ref{item:X-multipartite}) of the lemma gives
	\begin{align}
		\langle X,\cP\cQm\cP(X)\rangle &= \langle \bx, (\bG_{A_1}\otimes \dots \otimes \bG_{A_P}) \bs\rangle = \langle \bx, (\bG_{A_1}\otimes \dots \otimes \bG_{A_P}) \bL \bx\rangle, \label{eq:frobenius-as-inner-product-L}
	\end{align}
	where the inner product on the right is the standard one on $\mathbb{C}^{K!^p}$, and we defined a $(K!^p\times K!^p)$-matrix $\bL$ with coefficients
	\begin{align}
		(\bL)_{\vec{\omega}, \vec{\sigma}} = \sum_\tau y_\tau^{\vec{\sigma}}\,  z_{\omega^1}^{\tau,\sigma^1}\dots z_{\omega^P}^{\tau,\sigma^P}. \label{eq:L-multipartite}
	\end{align}
    Our goal in the following is to further manipulate the expression in \eqref{eq:frobenius-as-inner-product-L}.
	
	The Gram matrices $\bG_{A_p}$ are invertible since $K\leq |A_p|=q^{m_p}$ for all $p$ by assumption~\cite[Lem.~1]{harrow_approximate_2023-1}.
	Defining vectors $\bz^{\tau,\sigma^p},\bv^{\tau,\sigma^p}\in \mathbb{C}^{K!}$ via $(\bz^{\tau,\sigma^p})_{\omega^p} = z_{\omega^p}^{\tau,\sigma^p}$ and $(\bv^{\tau,\sigma^p})_{\omega^p} = v^{\tau,\sigma^p}_{\omega^p}$, with
	\begin{align}
		v^{\tau,\sigma^p}_{\omega^p} = \nip{\omega^p_{p}}{\tau_{\ell_p}\otimes \sigma^p_{m_p-\ell_p}},
	\end{align} 
    we take the inner product with $\widetilde{\omega}^p_{p}$ on both sides of \eqref{eq:z}:
    \begin{align}
        \nip{\widetilde{\omega}^p_{p}}{\cT_{A_p}\left(\tau_{\ell_p}\otimes \sigma^p_{m_p-\ell_p}\right)} = \sum_{\omega^p} z_{\omega^p}^{\tau,\sigma^p} \nip{\widetilde{\omega}^p_{p}}{\omega^p_{p}} = \sum_{\omega^p} z_{\omega^p}^{\tau,\sigma^p} (\bG_{A_p})_{\widetilde{\omega},\omega}. \label{eq:Tau-z}
    \end{align}
    Since the twirl $\cT_{A_p}$ is self-adjoint with respect to $\nip{\cdot}{\cdot}$ and $\widetilde{\omega}_p^p$ is invariant under  it, the left-hand side of \eqref{eq:Tau-z} is equal to 
    \begin{align} 
    \nip{\widetilde{\omega}^p_{p}}{\cT_{A_p}\left(\tau_{\ell_p}\otimes \sigma^p_{m_p-\ell_p}\right)} = \nip{\cT_{A_p}\left(\widetilde{\omega}^p_{p}\right)}{\tau_{\ell_p}\otimes \sigma^p_{m_p-\ell_p}} = \nip{\widetilde{\omega}^p_{p}}{\tau_{\ell_p}\otimes \sigma^p_{m_p-\ell_p}} = v^{\tau,\sigma^p}_{\widetilde{\omega}^p}.
    \end{align}
    Hence, \eqref{eq:Tau-z} can be rewritten as a vector equation $\bv^{\tau,\sigma^p} = \bG_{A_p}\bz^{\tau,\sigma^p}$, and multiplying with $\bG_{A_p}^{-1}$ gives the following vector equations for all $p=1,\dots,P$:
	\begin{align}
		\bz^{\tau,\sigma^p} &= \bG_{A_p}^{-1} \bv^{\tau,\sigma^p} \label{eq:z-multipartite}
	\end{align}
	We also rewrite the coefficients $y_\tau^{\vec{\sigma}}$ defined via \eqref{eq:y-multipartite-def}:
	\begin{align}
		w^{\vec{\sigma}}_{\chi} &\coloneqq \nip{\chi_{\ell_1 \dots \ell_P} }{\sigma^1_{\ell_1}\otimes \dots\otimes \sigma^P_{\ell_P}}\\
        &= \nip{\cQm\left(\chi_{\ell_1 \dots \ell_P}\right) }{\sigma^1_{\ell_1}\otimes \dots\otimes \sigma^P_{\ell_P}}\\
		&= \nip{\chi_{\ell_1 \dots \ell_P} }{\cQm\left(\sigma^1_{\ell_1}\otimes \dots\otimes \sigma^P_{\ell_P}\right)}\\
		&= \sum_\tau y_\tau^{\vec{\sigma}} \nip{\chi_{\ell_1}\otimes \dots \otimes \chi_{\ell_P}}{\tau_{\ell_1} \otimes \dots \otimes \tau_{\ell_P}}\\
		&= \sum_\tau (\bG_{\ell_1,\dots,\ell_P})_{\chi,\tau} y_\tau^{\vec{\sigma}},
	\end{align}
	where the second equality follows since $\chi_{\ell_1 \dots \ell_P}$ is invariant under $\cQm$, the third equality follows from the self-adjointness of $\cQm$, and in the last equality $\bG_{\ell_1,\dots,\ell_P}\equiv \bG_{A_{1}^{\ell_1}\dots A_{P}^{\ell_P}}$ is the Gram matrix of permutation operators acting on $A_{1}^{\ell_1}\dots A_{P}^{\ell_P}$.
	If $K\leq |A_{1}^{\ell_1}\dots A_{P}^{\ell_P}| = q^{\ell_1 + \dots + \ell_P}$, then this matrix is also invertible, and similar to \eqref{eq:z-multipartite} above we have the identity
	\begin{align}
		\by^{\vec{\sigma}} = \bG_{\ell_1,\dots,\ell_P}^{-1} \bw^{\vec{\sigma}} \label{eq:y-multipartite}
	\end{align}
	for the vectors $(\by^{\vec{\sigma}})_\tau = y_\tau^{\vec{\sigma}}$ and $(\bw^{\vec{\sigma}})_\tau = w^{\vec{\sigma}}_{\tau}$.
	
	We now substitute the expressions \eqref{eq:z-multipartite} and \eqref{eq:y-multipartite} in the definition \eqref{eq:L-multipartite} of $\bL$:
    \begin{align} 
    (\bL)_{\vec{\omega}, \vec{\sigma}} &= \sum_\tau y_\tau^{\vec{\sigma}}\,  z_{\omega^1}^{\tau,\sigma^1}\dots z_{\omega^P}^{\tau,\sigma^P}\\
    &=\sum_{\tau} \left(\bG_{\ell_1,\dots,\ell_P}^{-1} \bw^{\vec{\sigma}}\right)_{\tau} \left( \bG_{A_1}^{-1} \bv^{\tau,\sigma^1} \right)_{\omega^1} \dots \left( \bG_{A_P}^{-1} \bv^{\tau,\sigma^P} \right)_{\omega^P}\\
    &= \sum_{\tau,\chi,\eta^1,\dots,\eta^P}  \left(\bG_{\ell_1,\dots,\ell_P}^{-1}\right)_{\tau,\chi} w^{\vec{\sigma}}_{\chi} \left(\bG^{-1}_{A_1}\right)_{\omega^1,\eta^1} v_{\eta^1}^{\tau,\sigma^1} \dots \left(\bG^{-1}_{A_P}\right)_{\omega^P,\eta^P} v_{\eta^P}^{\tau,\sigma^P}\\
    &= \sum_{\vec{\eta}} \left(\bG^{-1}_{A_1}\right)_{\omega^1,\eta^1} \dots \left(\bG^{-1}_{A_P}\right)_{\omega^P,\eta^P} \sum_{\tau,\chi} \left(\bG_{\ell_1,\dots,\ell_P}^{-1}\right)_{\tau,\chi}v_{\eta^1}^{\tau,\sigma^1} \dots v_{\eta^P}^{\tau,\sigma^P} w^{\vec{\sigma}}_{\chi}\\
    &= \sum_{\vec{\eta}} \left(\bG^{-1}_{A_1}\otimes \dots\otimes \bG^{-1}_{A_P}\right)_{\vec{\omega},\vec{\eta}}\sum_{\tau,\chi} \left(\bG_{\ell_1,\dots,\ell_P}^{-1}\right)_{\tau,\chi}v_{\eta^1}^{\tau,\sigma^1} \dots v_{\eta^P}^{\tau,\sigma^P} w^{\vec{\sigma}}_{\chi}
    \end{align}
    Once again, this expression implies a matrix equation 
	\begin{align}
		\bL = \left(\bG_{A_1}^{-1} \otimes \dots \otimes \bG_{A_P}^{-1}\right) \overline{\bL}, \label{eq:L-as-product-multipartite}
	\end{align}
	where $\overline{\bL}$ is a $(K!^P\times K!^P)$-matrix with coefficients
	\begin{align}
		&(\overline{\bL})_{\vec{\eta}, \vec{\sigma}}\notag\\
		&= \sum_{\tau,\chi} (\bG_{\ell_1,\dots,\ell_P}^{-1})_{\tau,\chi} v_{\eta^1}^{\tau,\sigma^1}\dots v_{\eta^P}^{\tau,\sigma^P} w_\chi^{\vec{\sigma}}\\
		&= \sum_{\tau,\chi} (\bG_{\ell_1,\dots,\ell_P}^{-1})_{\tau,\chi} \left(\prod\nolimits_{p=1}^P \nip{\eta^p_p}{\tau_{\ell_p} \otimes \sigma^p_{m_p-\ell_p} }\right) \nip{\chi_{\ell_1\dots\ell_P}} {\sigma^1_{\ell_1}\otimes \dots \otimes \sigma^P_{\ell_P}} \\
		&= \sum_{\tau,\chi}(\bG_{\ell_1,\dots,\ell_P}^{-1})_{\tau,\chi} \left(\prod\nolimits_{p=1}^P \nip{\eta^p_{\ell_p}}{\tau_{\ell_p}} \nip{\eta^p_{m_p-\ell_p}}{\sigma^p_{m_p-\ell_p}}\right) \nip{\chi_{\ell_1\dots\ell_P}} {\sigma^1_{\ell_1}\otimes \dots \otimes \sigma^P_{\ell_P}}\\
		&= \prod\nolimits_{p=1}^P \nip{\eta^p_{m_p-\ell_p}}{\sigma^p_{m_p-\ell_p}} \sum_{\tau} \left(\prod\nolimits_{p=1}^P \nip{\eta^p_{\ell_p}}{\tau_{\ell_p}} \right) \sum_\chi (\bG_{\ell_1,\dots,\ell_P}^{-1})_{\tau,\chi} \left(\prod\nolimits_{p=1}^P \nip{\chi_{\ell_p}}{\sigma^p_{\ell_p}} \right).
	\end{align}
	This shows that $\overline{\bL}$ is equal to the following Hadamard product:
	\begin{align}
		\overline{\bL} =  \left(\bG_{A_1^{m_1-\ell_1}}\otimes \dots\otimes \bG_{A_P^{m_P-\ell_P}}\right) \odot \left(\bN_{\ell_1,\dots,\ell_P}^T \bG_{\ell_1,\dots,\ell_P}^{-1} \bN_{\ell_1,\dots,\ell_P} \right), \label{eq:Lbar-hadamard-multipartite}
	\end{align}
	where we defined the $(K!\times K!^P)$-matrix $\bN_{\ell_1,\dots,\ell_P}$ with coefficients
	\begin{align}
		(\bN_{\ell_1,\dots,\ell_P})_{\tau,\vec{\sigma}} = \nip{\tau_{\ell_1,\dots,\ell_P}}{\sigma^1_{\ell_1}\otimes\dots\otimes \sigma^P_{\ell_P}}.
	\end{align}
	Using \eqref{eq:L-as-product-multipartite} with \eqref{eq:Lbar-hadamard-multipartite} in \eqref{eq:frobenius-as-inner-product-L}, we obtain
	\begin{align}
		\langle X,\cP\cQm\cP(X)\rangle = \left\langle \bx,\left(\bG_{A_1^{m_1-\ell_1}}\otimes \dots\otimes \bG_{A_P^{m_P-\ell_P}}\right) \odot \left(\bN_{\ell_1,\dots,\ell_P}^T \bG_{\ell_1,\dots,\ell_P}^{-1} \bN_{\ell_1,\dots,\ell_P} \right) \bx \right\rangle,
	\end{align}
	which proves the claim.
\end{proof}

\begin{proof}[Proof of Lemma \ref{lem:frobenius-as-ip-multipartite}(\ref{item:N-multipartite})]
	Once again, let $X=\cP(X)=\sum_{\vsigma} x_\vsigma\, \sigma_1^1\otimes \dots\otimes \sigma^P_P$ for some coefficients $x_\vsigma\in\mathbb{C}$ by assumption of the lemma.
	We define coefficients $r^\vsigma_\pi$ via
	\begin{align}
		\cR\left(\sigma^1_1\otimes \dots \otimes \sigma^P_P\right) = \sum_{\pi} r^\vsigma_\pi\, \pi_{1\dots P}, \label{eq:r}
	\end{align}
	so that we can write
	\begin{align}
		\cR(X) = \sum_\vsigma x_\vsigma \sum_\pi r_\pi^\vsigma\, \pi_{1\dots P}.
	\end{align}
	
	Taking the normalized inner product with $X$ gives
	\begin{align}
		\nip{X}{\cR(X)} &= \sum_{\vomega,\vsigma,\pi} x_\vomega^* \, x_\vsigma\,r_\pi^\vsigma\, \nip{\omega^1_1\otimes \dots\otimes \omega^P_P}{\pi_{1\dots P}}\\
		&= \sum_{\vomega,\vsigma,\pi} x_\vomega^* \, x_\vsigma\,r_\pi^\vsigma\,t_{\pi}^\vomega,\label{eq:X-RX}
	\end{align}
	where we defined the coefficients
	\begin{align}
		t_{\pi}^\vomega &= \nip{\pi_{1\dots P}}{\omega^1_1\otimes \dots\otimes \omega^P_P}\\
		&= \nip{\cR(\pi_{1\dots P})}{\omega^1_1\otimes \dots\otimes \omega^P_P}\\
		&= \nip{\pi_{1\dots P}}{\cR\left(\omega^1_1\otimes \dots\otimes \omega^P_P\right)}\\
		&= \sum_\rho r_\rho^\vomega\, \nip{\pi_{1\dots P}}{\rho_{1\dots P}}\\
		&= \sum_{\rho}(\bG_{A_1\dots A_P})_{\pi,\rho}\, r_\rho^\vomega.
	\end{align}
	Here, we used the invariance of $\pi_{1\dots P}$ under $\cR$ in the second equality, and the definition \eqref{eq:r} of the $r_*^*$ coefficients in fourth equality.
	
	Defining vectors $\br^\vsigma$ with $(\br^\vsigma)_\pi = r_\pi^\vsigma$ and $\bt^\vsigma$ with $(\bt^\vsigma)_\pi = t_\pi^\vsigma$, we then have $\bt^\vsigma = \bG_{A_1\dots A_P} \br^\vsigma$.
	If $K \leq |A_{1}\dots A_P| = q^{m_1+\dots m_P}$, then $\bG_{A_1\dots A_P}$ is invertible \cite{harrow_approximate_2023-1}, and hence $\br^\vsigma = \bG_{A_1\dots A_P}^{-1} \bt^\vsigma$.
	Using this in \eqref{eq:X-RX}, we have
	\begin{align}
		\nip{X}{\cR(X)} &= \sum_{\vomega,\vsigma,\pi} x_\vomega^* \, x_\vsigma\,r_\pi^\vsigma\,t_{\pi}^\vomega\\
		&= \sum_{\vomega,\vsigma,\pi,\rho} x_\vomega^* \, x_\vsigma\,(\bG_{A_1\dots A_P}^{-1})_{\pi,\rho}\, t_\rho^\vsigma \, t_{\pi}^\vomega\\
		&= \sum_\vomega x_\vomega^* \, \sum_\pi t_{\pi}^\vomega\, \sum_{\rho} (\bG_{A_1\dots A_P}^{-1})_{\pi,\rho}\, \sum_{\vsigma} t_\rho^\vsigma \,x_\vsigma,
	\end{align}
	which can be rewritten as
	\begin{align}
		\nip{X}{\cR(X)} = \left\langle \bx, \bN_{A_1 \dots A_P}^T \bG_{A_1\dots A_P}^{-1} \bN_{A_1 \dots A_P} \bx \right\rangle
	\end{align}
	in terms of the $(K!\times K!^P)$-matrix $\bN_{A_1 \dots A_P}$ with coefficients
	\begin{align}
		\left(\bN_{A_1 \dots A_P}\right)_{\pi,\vsigma} = t_\pi^\vsigma = \nip{\pi_{1\dots P}}{\sigma^1_1 \otimes \dots \otimes \sigma^P_P}.
	\end{align}
	This concludes the proof.
\end{proof}

\subsubsection{Proof of subspace angle bound for Twirl-Swap-Twirl}
\label{sec:subspace-angle-swap}

We now give the proof of Proposition \ref{prop:subspace-angle-swap}, which bounds the subspace angle cosine for the Twirl-Swap-Twirl protocol.
To this end, we first rewrite the squared numerator on the right-hand side of \eqref{eq:subspace-angle} by expanding the $2$-norm, and using Lemma \ref{lem:projections-2-norm-inner-product} via the `dominance relations' \eqref{eq:dominance}:
\begin{align}
	\sup_{X\neq 0}\frac{\left\|\cQs\cP(X)-\cR(X)\right\|_2^2}{\|X\|_2^2} &= \sup_{X\neq 0}\frac{\langle X, \cP\cQs\cP(X) - \cR(X)\rangle}{\langle X,X\rangle} .\label{eq:norm-as-inner-product}
\end{align}

The relations \eqref{eq:dominance} furthermore show that the supremum in \eqref{eq:norm-as-inner-product} is achieved on some $X=\cP(X)$, which according to \eqref{eq:unitary-invariant-X-permutations} can be written as
\begin{align}
	X = \sum_{\sigma,\tau\in \kS_k} x_{\sigma,\tau}\, \sigma_A \otimes \tau_B\label{eq:Px}
\end{align}
for some coefficients $x_{\sigma,\tau}\in\mathbb{C}$.
Note that $\sigma_A \otimes \tau_B$ is invariant under the projections $\cP,\cQs,\cR$ whenever $\sigma=\tau$.
Following \cite{harrow_approximate_2023-1}, we denote by $\mathbf{G}_S$ the $(k!\times k!)$-Gram matrix of permutation operators acting on $k$ copies of a system $S$ with elements $(\mathbf{G}_S)_{\sigma,\tau} = \nip{\sigma_S}{\tau_S}$ for $\sigma,\tau\in\kS_k$.
Here, $
	\nip{X}{Y} = |S|^{-k} \tr(X^\dagger Y)$
denotes the normalized Frobenius product and $|S|=\dim\cH_S$.

Because of the quotient of inner products in \eqref{eq:norm-as-inner-product}, we can replace $\langle \cdot,\cdot\rangle$ by the normalized inner product $\nip{\cdot}{\cdot}$.
Using the assumptions $X=\cP(X)$ and $k^2\leq q^m$ together with Lemma \ref{lem:inner-product-rewriting}(\ref{item:X}), (\ref{item:M}) and (\ref{item:N}) from Section \ref{sec:auxiliary-lemma}, the expression \eqref{eq:norm-as-inner-product} is then equal to
\begin{align}
	\frac{\langle X, \cP\cQs\cP(X) - \cR(X)\rangle}{\langle X,X\rangle}  &= \frac{\left\langle \bx,\left[\bM\left(\mathbf{G}_A^{-1}\otimes \mathbf{G}_B^{-1}\right)\bM - \bN^T \bG_{AB}^{-1} \bN\right] \bx \right\rangle}{\left\langle \bx,(\bG_A\otimes \bG_B)\bx \right\rangle},\label{eq:inner-product-bound}
\end{align}
where the inner product $\langle\cdot,\cdot\rangle$ on the right-hand side (and in the following) is now the usual standard inner product on $\mathbb{C}^{k!^2}$.
The matrices $\bM$ of size $(k!^2\times k!^2)$ and $\bN$ of size $(k!\times k!^2)$ have the following coefficients for $\pi,\rho,\omega,\chi\in\kS_k$ (cf.~Lemma \ref{lem:inner-product-rewriting}):
\begin{align}
	(\bM)_{(\pi,\rho),(\omega,\chi)} &= \nip{\pi_A}{\chi_{A^\ell}\otimes \omega_{A^{m-\ell}}} \nip{\rho_B}{\omega_{B^\ell}\otimes \chi_{B^{m-\ell}}} \label{eq:M-text}\\
	(\bN)_{\pi,(\omega,\chi)} &= \nip{\pi_{AB}}{\omega_A\otimes \chi_B}. \label{eq:N-text1}
\end{align}
Here, for $\chi,\omega\in\kS_k$ we write $\chi_{A^\ell}$ for the permutation operator acting on the first $\ell\leq m$ qudits within each of the $k$ blocks (of size $m$ each) in $A$, and $\omega_{A^{m-\ell}}$ for the permutation operator acting on the remaining $m-\ell$ qudits within each block.
The same conventions hold for permutations acting on the blocks in $B$.

The rest of the proof is concerned with bounding the expression on the right-hand side of \eqref{eq:inner-product-bound} for all $\cL(A^k B^k)\ni X\neq 0$ (or equivalently $\mathbb{C}^{k!^2}\ni \bx\neq 0$).
This in turn will give a bound on the supremum over $\bx\neq 0$ and thus the $2$-norm in \eqref{eq:subspace-angle} that we seek to estimate.
To this end, we employ the following inequalities proved in \cite[Lemma 1]{harrow_approximate_2023-1}, with $\lmin(\mathbf{H})$ and $\lmax(\mathbf{H})$ denoting the smallest and largest eigenvalue of a Hermitian operator $\mathbf{H}$, respectively:
\begin{align}
	\lmin(\bG_A) &\geq 1-\frac{k^2}{2q^m} \equiv 1-\eps\label{eq:Ga-min-eval}\\
	\lmax(\bG_{AB}) &\leq \exp\left(\frac{k^2}{2q^{2m}}\right) \equiv \exp(\eps'), \label{eq:Gab-max-eval}
\end{align}
where we set $\eps \coloneqq \frac{k^2}{2q^m}$ and $\eps' \coloneqq \frac{k^2}{2q^{2m}}$.
From \eqref{eq:Ga-min-eval} we get $\lmax\left(\bG_A^{-1}\right) \leq (1-\eps)^{-1}$ and hence, using $|A_i|=|B_i|=q^m$, that
\begin{align}
	\bG_A^{-1} \otimes \bG_B^{-1} \leq (1-\eps)^{-2} \, \one_{k!^2}.
\end{align}
We also have $\lmax\left(-\bG_{AB}^{-1}\right)\leq -\exp(-\eps')$ from \eqref{eq:Gab-max-eval}, which leads to
\begin{align}
	-\bG_{AB}^{-1} \leq -\exp(-\eps') \one_{k!}. \label{eq:inverse-Gab-bound}
\end{align}
Defining constants
\begin{align}
    a &\coloneqq (1-\eps)^{-2} = \left(1-\frac{k^2}{2q^m}\right)^{-2} & a_1 &\coloneqq\exp(-\eps') = \exp\left(-\frac{k^2}{2q^{2m}}\right), \label{eq:constants-twirl-swap-twirl}
\end{align} 
these operator inequalities yield 
\begin{align}
	\left\langle \bx, \left[\bM\left(\mathbf{G}_A^{-1}\otimes \mathbf{G}_B^{-1}\right)\bM - \bN^T \bG_{AB}^{-1} \bN \right] \bx \right\rangle &\leq 
	\left\langle \bx, \left(a \bM^2 - a_1 \bN^T \bN\right)\bx\right\rangle. \label{eq:inner-product-without-G}
\end{align}
The bound \eqref{eq:Ga-min-eval} further gives $\bG_A\otimes\bG_B \geq a^{-1} \one_{k!^2} = (1-\eps)^{2}\one_{k!^2}$ and thus
\begin{align}
	\left\langle \bx,(\bG_A\otimes \bG_B)\bx \right\rangle^{-1} \leq a \langle \bx,\bx\rangle^{-1}. \label{eq:x-G-x-bound}
\end{align}
Using the inequalities \eqref{eq:inner-product-without-G} and \eqref{eq:x-G-x-bound} in \eqref{eq:inner-product-bound} gives
\begin{align}
	\sup_{X\neq 0} \frac{\langle X, \cP\cQs\cP(X) - \cR(X)\rangle}{\langle X,X\rangle} \leq a \sup_{\mathbf{x}\neq 0} \frac{\left\langle \bx, \left(a \bM^2 - a_1 \bN^T \bN\right)\bx\right\rangle}{\langle \bx,\bx\rangle} = a \left\|a \bM^2 - a_1 \bN^T \bN \right\|, \label{eq:inner-prod-as-op-norm}
\end{align}
where the equality follows from the fact that $\bM$ is Hermitian by Lemma \ref{lem:inner-product-rewriting}, and thus $a \bM^2 - a_1 \bN^T \bN$ is Hermitian as well.

We now focus on bounding the operator norm of $a \bM^2 - a_1 \bN^T \bN$. 
The main idea is that both $\bM^2$ and $\bN^T\bN$ can be written as a sum of two matrices: a diagonal matrix $\bM_0$ with operator norm $1$ that is identical for both $\bM^2$ and $\bN^T\bN$, and a matrix with small coefficients provided $q^{\ell}$ is much larger than certain powers of $k!$.
Using simple norm estimates, $a \bM^2 - a_1 \bN^T \bN$ can then be written as a sum of matrices with small operator norm.

Inspecting the definition of $\bM$ in \eqref{eq:M-text}, we see that $(\bM)_{(\pi,\pi),(\pi,\pi)}$ is equal to $1$ for any $\pi\in\kS_k$.
On the other hand, for a quantum system $S=(\mathbb{C}^d)^{\otimes k}$ with dimension $|S|=d^k$ we have $\nip{\pi_S}{\rho_S} = |S|^{-|\pi^{-1}\rho|}$, where $|\sigma|$ denotes the minimum number of terms needed to write $\sigma\in\kS_k$ as a product of transpositions \cite{harrow_approximate_2023-1}.
Thus, $\nip{\pi_S}{\rho_S} \leq |S|^{-1}$ if $\pi\neq \rho$, which implies that $(\bM)_{(\pi,\rho),(\omega,\chi)} \leq q^{-2\ell}$
if at least one of $\pi,\rho,\omega,\chi$ is distinct from the rest.
For example, for $\chi\neq \pi$ we have 
\begin{align} 
	(\bM)_{(\pi,\pi),(\pi,\chi)} &= \nip{\pi_A}{\chi_{A^\ell}\otimes \pi_{A^{m-\ell}}} \nip{\pi_B}{\pi_{B^\ell}\otimes \chi_{B^{m-\ell}}}\\
	&= \nip{\pi_{A^\ell}}{\chi_{A^\ell}} \nip{\pi_{B^{m-\ell}}}{\chi_{B^{m-\ell}}}\\
	&\leq q^{-\ell} q^{-(m-\ell)}\\
	&\leq q^{-2\ell}.
\end{align}
We can thus write
\begin{align}
	\bM = \bM_0 + q^{-2\ell} \bM_1, \label{eq:M-decomposition}
\end{align}
where $\bM_1$ is a $(k!^2\times k!^2)$-matrix with each coefficient bounded from above by $1$, and $\bM_0$ is a $(k!^2\times k!^2)$-matrix with coefficients $(\bM_0)_{(\pi,\rho),(\omega,\chi)} = \delta_{\pi,\rho}\delta_{\pi,\omega}\delta_{\pi,\chi}$.
In other words, $\bM_0$ is a diagonal matrix with $1$'s in the positions $((\pi,\pi),(\pi,\pi))$ for $\pi\in\kS_k$, and $0$'s elsewhere.
Squaring both sides of \eqref{eq:M-decomposition} gives
\begin{align}
	\bM^2 &= \left(\bM_0 + q^{-2\ell} \bM_1\right)^2\\
	&= \bM_0^2 + q^{-2\ell}(\bM_0\bM_1 + \bM_1\bM_0) + q^{-4\ell}\bM_1^2\\
	&= \bM_0 + 2q^{-2\ell} \bM_2 + k!^2 q^{-4\ell} \bM_3, \label{eq:M2-bound}
\end{align}
where $\bM_2,\bM_3$ are again matrices with coefficients bounded by $1$.
In the third equality, we used $\bM_0^2=\bM_0$, and the fact that the coefficients of $\bM_1^2$ are sums of $k!^2$ products of coefficients of $\bM_1$, each at most equal to $1$.
Hence, the coefficients of $\bM_1^2$ are bounded by $k!^2$, which we use to define $\bM_3$.

We now apply the same argument to $\bN^T\bN$:
Inspecting the definition of $\bN$ in \eqref{eq:N-text1} reveals that $(\bN)_{\pi,(\sigma,\tau)}$ is $1$ if $(\sigma,\tau) = (\pi,\pi)$, and bounded from above by $q^{-m}$ otherwise.
Hence, we write
\begin{align}
	\bN = \bN_0 + q^{-m} \bN_1, \label{eq:N-decomposition}
\end{align}
where $\bN_0$ is a $(k!\times k!^2)$-matrix with coefficients $(\bN_0)_{\pi,(\sigma,\tau)} = \delta_{\pi,\sigma}\delta_{\pi,\tau}$, and $\bN_1$ is a matrix whose coefficients are bounded from above by $1$.
Note that we have $\bN_0^T\bN_0 = \bM_0$, and hence
\begin{align}
	\bN^T\bN &= (\bN_0+q^{-m}\bN_1)^T(\bN_0+q^{-m}\bN_1)\\
	&= \bM_0 + q^{-m} \left(\bN_1^T\bN_0 + \bN_0^T\bN_1 \right) + q^{-2m} \bN_1^T\bN_1\\
	&= \bM_0 + 2q^{-m} \bM_4 + k!q^{-2m}\bM_5, \label{eq:NTN-bound}
\end{align}
with matrices $\bM_4,\bM_5$ whose coefficients are bounded by $1$ as before.

Using \eqref{eq:M2-bound} and \eqref{eq:NTN-bound}, we get the following bound on the operator norm of $a \bM^2 - a_1 \bN^T \bN$:
\begin{align}
	&\left\|a\bM^2-a_1\bN^T\bN \right\| \notag\\
	&\quad {} = \left\| a \bM_0^2 + 2aq^{-2\ell} \bM_2 + ak!^2 q^{-4\ell} \bM_3 - a_1 \bM_0 - 2b q^{-m} \bM_4 - a_1k! q^{-2m}\bM_5 \right\|\\
	&\quad {} \leq |a-a_1| \|\bM_0\| + 2aq^{-2\ell} \|\bM_2\| + ak!^2 q^{-4\ell} \|\bM_3\| + 2a_1q^{-m} \| \bM_4\| + a_1k!q^{-2m} \| \bM_5\|, \label{eq:sum-of-op-norms}
\end{align}
where we used the triangle inequality in the last inequality.
Since $\bM_0$ is a diagonal matrix with $1$'s and $0$'s on the diagonal, we have $\|\bM_0\|=1$.
Furthermore, $\bM_i$ for $i=2,\dots,5$ are $(k!^2\times k!^2)$-matrices with coefficients bounded above by 1.
Using the norm inequality 
\begin{align}
	\|X\| \leq \sqrt{d} \max_{1\leq i\leq d}\sum_{j=1}^d |x_{ij}|
\end{align} 
valid for any $(d\times d)$-matrix $X$, we can bound the operator norms of $\bM_i$ for $i=2,\dots,5$ by $k!\cdot k!^2 = k!^3$, giving
\begin{align}
	\| a\bM^2-a_1\bN^T\bN \| \leq a-a_1 + 4ak!^3 q^{-2\ell} + 2ak!^5 q^{-4\ell}, \label{eq:L-op-norm}
\end{align}
where we also used $a = (1-\eps)^{-2}\geq \exp(-\eps') = a_1$ and the assumption $2\ell\leq m$.
Note that both $a,a_1\to 1$ for $k^2\ll q^{2m}$, and hence also $a-a_1\to 0$ in this limit.

Using the bound on $\|a \bM^2 - a_1 \bN^T \bN\|$ from \eqref{eq:L-op-norm} in \eqref{eq:inner-product-without-G}, we arrive at
\begin{align}
	\frac{\langle X, \cP\cQs\cP(X) - \cR(X)\rangle}{\langle X,X\rangle} &\leq \left(a-a_1 + 4ak!^3 q^{-2\ell} + 2ak!^5 q^{-4\ell}\right) \frac{\langle \bx,\bx\rangle}{\left\langle \bx,(\bG_A\otimes \bG_B)\bx \right\rangle}\\
	&\leq a(a-a_1) + 4a^2k!^3 q^{-2\ell} + 2a^2k!^5 q^{-4\ell}.\label{eq:last-ip-bound}
\end{align}
Since the left-hand side of \eqref{eq:last-ip-bound} bounds the subspace angle cosine $c(\im\cP,\im\cQs)^2$ from above via \eqref{eq:subspace-angle} and \eqref{eq:norm-as-inner-product},
\begin{align}
	c(\im\cP,\im\cQs)^2 &\leq  a(a-a_1) + 4a^2k!^3 q^{-2\ell} + 2a^2k!^5 q^{-4\ell}.
\end{align}
We have $\eps'=k^2/(2q^{2m}) \leq k^2/(2q^m) = \eps$, and hence 
\begin{align} 
a-a_1 = (1-\eps)^{-2}-\exp(-\eps') \leq (1-\eps)^{-2}-\exp(-\eps).
\end{align}
The assumption $k^2\leq q^m$ implies that $\eps \leq 1/2$, and in this regime $(1-\eps)^{-2}-\exp(-\eps) \leq 8\eps$.
Thus, 
\begin{align}
    a-a_1 \leq 8\eps = 4k^2q^{-m},
\end{align}
which concludes the proof.

\subsubsection{Proof of the subspace angle bound for multipartite crosstwirl protocol}\label{sec:proof-multipartite}

In this section we prove Proposition \ref{prop:subspace-angle-multi-crosstwirl} giving a bound on the subspace angle for the multipartite crosstwirl protocol described in Section \ref{sec:multipartite}.

	We first rewrite the subspace angle cosine in terms of the normalized Frobenius inner product $\nip{\cdot}{\cdot}$ using Lemma \ref{lem:projections-2-norm-inner-product} and \eqref{eq:dominance-multipartite}:
	\begin{align}
		c(\im\cP,\im\cQm)^2 = \sup_{X\neq 0}\frac{\left\| \cQm\cP(X)-\cR(X)\right\|^2}{\|X\|^2} = \sup_{X\neq 0} \frac{\nip{X}{\cP\cQm\cP(X)-\cR(X)}}{\nip{X}{X}}, \label{eq:norm-to-ip-multipartite}
	\end{align}
	and the supremum is achieved on an element $X = \cP(X)$, which by \eqref{eq:unitary-invariant-X-permutations} can be written as
	\begin{align}
		X = \cP(X) = \sum_{\vsigma} x_{\vsigma}\, \sigma^1_{A_1}\otimes \dots\otimes \sigma^P_{A_P}\label{eq:P(X)}
	\end{align}
	for some coefficients $x_{\vsigma}\in\mathbb{C}$, where we write $\vsigma = (\sigma^1,\dots,\sigma^P)$ with $\sigma^p\in\kS_K$ for $p=1,\dots,P$.\footnote{Note that the superscript in $\sigma^p$ is an index and not a power.}

    Recall that, by assumption of the proposition, $K^2 < 2q^{\ell_p} = \exp_q(\ell_p + \log_q(2))$ for all $p=1,\dots,P$.
    This implies $K< \exp_q(\frac{1}{2}(\ell_p + \log_q(2)) \leq q^{\ell_p}$ for all $p$, since $\log_q(2)\leq 1\leq \ell_p$ for $q\geq 2$.
    From $K\leq q^{\ell_p}$ we obtain both $K\leq q^{m_p}$ for all $p$ as well as $K\leq q^L$.
    We can therefore use Lemma \ref{lem:frobenius-as-ip-multipartite} proved in Sec.~\ref{sec:auxiliary-lemma} above to express the right-hand side of \eqref{eq:norm-to-ip-multipartite} as
	\begin{align}
		&\sup_{X\neq 0} \frac{\nip{X}{\cP\cQm\cP(X)-\cR(X)}}{\nip{X}{X}} \notag\\
		&= \sup_{\bx\neq 0} \frac{ \left\langle \bx, \left[\left(\bigotimes\nolimits_{p=1}^P\bG_{A_p^{m_p-\ell_p}}\right) \odot \left(\bN_{\ell_1,\dots,\ell_P}^T \bG_{\ell_1,\dots,\ell_P}^{-1} \bN_{\ell_1,\dots,\ell_P} \right) - \bN_{A_1\dots A_P}^T \bG_{A_1\dots A_P}^{-1} \bN_{A_1\dots A_P} \right]\bx \right\rangle } {\left\langle \bx, \left( \bigotimes\nolimits_{p=1}^P \bG_{A_p}\right) \bx \right\rangle}, \label{eq:frobenius-as-ip-multipartite}
	\end{align}
	where $\bG_{\ell_1,\dots,\ell_P}\equiv \bG_{A_{1}^{\ell_1}\dots A_{P}^{\ell_P}}$ is the Gram matrix of permutation operators acting on $A_{1}^{\ell_1}\dots A_{P}^{\ell_P}$, and the $(K!\times K!^P)$-matrices $\bN_{\ell_1,\dots,\ell_P}$ and $\bN_{A_1\dots A_P}$ are defined as
	\begin{align}
		\left(\bN_{\ell_1,\dots,\ell_P}\right)_{\pi,\vsigma} &= \nip{\pi_{A_1^{\ell_1}\dots A_P^{\ell_P}}}{\sigma^1_{A_1^{\ell_1}}\otimes \dots \otimes \sigma^P_{A_P^{\ell_P}}} \label{eq:N-ell}\\
		\left(\bN_{A_1\dots A_P}\right)_{\pi,\vsigma} &= \nip{\pi_{A_1\dots A_P}}{\sigma^1_{A_1}\otimes \dots \otimes \sigma^P_{A_P}}.\label{eq:N-multipartite}
	\end{align}
	
	Just as in the proof of Theorem \ref{prop:subspace-angle-swap}, we first eliminate the various Gram matrices in \eqref{eq:frobenius-as-ip-multipartite}.
	To this end, with $M=m_1 + \dots + m_P$ and $L = \ell_1 + \dots + \ell_P$ we define constants 
	\begin{align}
		\begin{aligned} 
        b_1 &= \prod_{p=1}^P \left(1-\frac{K^2}{2q^{m_p}}\right)^{-1} & b_2 &= \exp\left(-\frac{K^2}{2q^M}\right) \\ b_3 &= \left(1-\frac{K^2}{2q^L}\right)^{-1}
		& b_4&= \prod_{p=1}^P \exp\left(\frac{K^2}{2q^{m_p-\ell_p}}\right).
        \end{aligned}
        \label{eq:constants-crosstwirl}
	\end{align}
	We then use various operator inequalities that follow from the bounds on the largest and smallest eigenvalues of Gram matrices derived in \cite{harrow_approximate_2023-1}.
	First, $\bG_{A_p}  \geq \left(1-\frac{K^2}{2q^{m_p}}\right) \one_{K!}$ and hence
	\begin{align}
		\left\langle \bx, \left( \bigotimes\nolimits_{p=1}^P \bG_{A_p}\right) \bx \right\rangle^{-1} \leq b_1 \langle\bx,\bx\rangle^{-1}.
	\end{align}
	Second, $\bG_{A_1\dots A_P} \leq \exp\left(\frac{K^2}{2q^M}\right) \one_{K!}$ and hence
	\begin{align}
		{}-\big\langle \bx, \bN_{A_1\dots A_P}^T  \bG_{A_1\dots A_P}^{-1} \bN_{A_1\dots A_P} \bx \big\rangle \leq -b_2 \left\langle \bx, \bN_{A_1\dots A_P}^T \bN_{A_1\dots A_P} \bx\right\rangle.
	\end{align}
	To deal with the Hadamard product term in \eqref{eq:frobenius-as-ip-multipartite}, we use the following argument:
	From $\bG_{\ell_1,\dots,\ell_P} \geq \left(1-\frac{K^2}{2q^L}\right) \one_{K!}$ we get 
	\begin{align}
		\bN_{\ell_1,\dots,\ell_P}^T \bG_{\ell_1,\dots,\ell_P}^{-1} \bN_{\ell_1,\dots,\ell_P} \leq b_3 \bN_{\ell_1,\dots,\ell_P}^T \bN_{\ell_1,\dots,\ell_P},\label{eq:bound-G-ell_1-ell_p}
	\end{align}
	and $\bG_{A_p^{m_p-\ell_p}} \leq \exp\left(\frac{K^2}{2q^{m_p-\ell_p}}\right)\one_{K!}$ gives
	\begin{align}
		\bigotimes_{p=1}^P\bG_{A_p^{m_p-\ell_p}} \leq b_4 \one_{K!^P}. \label{eq:bound-G-prod-m_p-ell_p}
	\end{align}
	Now $X\odot Y\geq 0$ if $X,Y\geq 0$ \cite[Ch.~5]{horn_topics_1991} and $(X+Z)\odot Y = X\odot Z + X\odot Y$, which we can use together with \eqref{eq:bound-G-ell_1-ell_p} and \eqref{eq:bound-G-prod-m_p-ell_p} to obtain
	\begin{multline}
		\left\langle \bx, \left(\bigotimes\nolimits_{p=1}^P\bG_{A_p^{m_p-\ell_p}}\right) \odot \left(\bN_{\ell_1,\dots,\ell_P}^T \bG_{\ell_1,\dots,\ell_P}^{-1} \bN_{\ell_1,\dots,\ell_P} \right) \bx\right\rangle \\
		\leq b_3b_4  \left\langle \bx, \left(\one_{K!^P}\odot \bN_{\ell_1,\dots,\ell_P}^T  \bN_{\ell_1,\dots,\ell_P} \right) \bx\right\rangle.
	\end{multline}
	Note that $\one_{K!^P}\odot \bN_{\ell_1,\dots,\ell_P}^T  \bN_{\ell_1,\dots,\ell_P}$ is a diagonal matrix by the definition of the (element-wise) Hadamard matrix product $\odot$, which we will exploit further down in the proof.
	
	Using the above bounds in \eqref{eq:frobenius-as-ip-multipartite} gives
	\begin{align}
		&\sup_{X\neq 0} \frac{\nip{X}{\cP\cQm\cP(X)-\cR(X)}}{\nip{X}{X}}\notag\\ 
		&\eqspace \leq b_1 \sup_{\bx\neq 0} \frac{\left\langle \bx, \left(b_3b_4  \one_{K!^P}\odot \bN_{\ell_1,\dots,\ell_P}^T  \bN_{\ell_1,\dots,\ell_P} - b_2\bN_{A_1\dots A_P}^T\bN_{A_1\dots A_P}\right) \bx\right\rangle }{\langle\bx , \bx\rangle}\\
		& \eqspace = b_1 \left\| b_3b_4 \one_{K!^P}\odot \bN_{\ell_1,\dots,\ell_P}^T \bN_{\ell_1,\dots,\ell_P} - b_2 \bN_{A_1\dots A_P}^T\bN_{A_1\dots A_P} \right\|. \label{eq:bound-in-terms-of-op-norm-multipartite}
	\end{align}
	To bound the operator norm in the last line, we define two $(K!^P\times K!^P)$-matrices $\bD$ and $\bE$ via
	\begin{align}
		(\bD)_{\vsigma,\vomega} &= \begin{cases}\left(b_3b_4\bN_{\ell_1,\dots,\ell_P}^T \bN_{\ell_1,\dots,\ell_P} - b_2\bN_{A_1\dots A_P}^T\bN_{A_1\dots A_P}\right)_{\vsigma,\vsigma} & \text{if } \vsigma=\vomega,\\ 0 & \text{otherwise;} \end{cases}\\
		(\bE)_{\vsigma,\vomega} &= \begin{cases}0 &  \text{if } \vsigma=\vomega,\\ \left( b_2\bN_{A_1\dots A_P}^T\bN_{A_1\dots A_P}\right)_{\vsigma,\vomega} & \text{otherwise.}  \end{cases}
	\end{align}
	In other words, the matrices $\bD$ and $\bE$ are the diagonal and off-diagonal parts of the matrix $b_3b_4 \one_{K!^P}\odot \bN_{\ell_1,\dots,\ell_P}^T \bN_{\ell_1,\dots,\ell_P} - b_2 \bN_{A_1\dots A_P}^T\bN_{A_1\dots A_P}$, respectively, and
	\begin{align}
		\bD - \bE = b_3b_4 \one_{K!^P}\odot \bN_{\ell_1,\dots,\ell_P}^T \bN_{\ell_1,\dots,\ell_P} - b_2 \bN_{A_1\dots A_P}^T\bN_{A_1\dots A_P}.
	\end{align}
	This allows us to bound the subspace angle cosine as follows:
	\begin{align}
		c(\im\cP,\im\cQm)^2 &= \sup_{X\neq 0}\frac{\left\| \cQm\cP(X)-\cR(X)\right\|^2}{\|X\|^2}\\
		&=\sup_{X\neq 0} \frac{\nip{X}{\cP\cQm\cP(X)-\cR(X)}}{\nip{X}{X}}\\
		&\leq b_1 \|\bD - \bE\|\\
		&\leq b_1 \|\bD\| + b_1 \|\bE\|. \label{eq:subspace-angle-multi-bound}
	\end{align}
	
	We first bound the operator norm of $\bE$.
	This is a non-negative irreducible matrix, for which the generalized Perron-Frobenius Theorem states that $\|\bE\| \leq \max_{\vsigma} \sum_{\vomega} (\bE)_{\vsigma,\vomega}$ \cite[Ch.~1]{horn2012matrix}.
	Let $\vsigma=(\sigma^1,\dots,\sigma^P)$ denote the vector of permutations achieving this maximum.
	Then we upper-bound this expression by summing over all $\sigma^2,\dots,\sigma^P\in \kS_k$.
	This is a weaker upper bound, but in return we can calculate the latter sum explicitly as follows:
	\begin{align}
		\|\bE\| &\leq \sum_{\vomega} (\bE)_{\vsigma,\vomega}\\
		&\leq \sum_{\sigma^2,\dots,\sigma^P} \sum_{\vomega} (\bE)_{\vsigma,\vomega}\\
		&= \sum_{\sigma^2,\dots,\sigma^P} \left( \sum_{\vomega} (b_2\bN_{A_1\dots A_P}^T\bN_{A_1\dots A_P})_{\vsigma,\omega} - (b_2\bN_{A_1\dots A_P}^T\bN_{A_1\dots A_P})_{\vsigma,\vsigma} \right)\\
		&= b_2\sum_{\sigma^2,\dots,\sigma^P} \sum_{\vomega} (\bN_{A_1\dots A_P}^T\bN_{A_1\dots A_P})_{\vsigma,\omega} - b_2 \sum_{\sigma^2,\dots,\sigma^P} (\bN_{A_1\dots A_P}^T\bN_{A_1\dots A_P})_{\vsigma,\vsigma} \label{eq:K-op-norm-multi-calculation}
	\end{align}
	For the first sum,
	\begin{align}
		&\sum_{\sigma^2,\dots,\sigma^P} \sum_{\vomega} \left(\bN_{A_1\dots A_P}^T\bN_{A_1\dots A_P}\right)_{\vsigma,\vomega} \notag \\
		&= \sum_{\sigma^2,\dots,\sigma^P} \sum_{\vomega} \sum_\pi \prod_{p=1}^P \nip{\sigma^p_{A_p}}{\pi_{A_p}} \prod_{p=1}^P \nip{\pi_{A_p}}{\omega^p_{A_p}}\\
		&= \sum_\pi \nip{\sigma^1_{A_1}}{\pi_{A_1}} \sum_{\sigma^2} \nip{\sigma^2_{A_2}}{\pi_{A_2}} \dots \sum_{\sigma^P} \nip{\sigma^P_{A_P}}{\pi_{A_P}}\sum_{\omega^1} \nip{\pi_{A_1}}{\omega^1_{A_1}} \dots \sum_{\omega^P} \nip{\pi_{A_P}}{\omega^P_{A_P}}\label{eq:K-multi-sum05}
	\end{align}

	To evaluate the sums over permutations in \eqref{eq:K-multi-sum05}, we recall from \cite{harrow_approximate_2023-1} that for any $\omega\in\kS_K$ and $d\geq 1$ we have
	\begin{align}
			\sum_\omega d^{c(\omega)} = K!\binom{d+K-1}{K}.\label{eq:harrow-cycle-sum}
		\end{align}
	With $d \coloneqq q^{m_P}$, we can use this identity to evaluate the right-most sum in \eqref{eq:K-multi-sum05} as follows:
	\begin{align}
			\sum_{\omega^P} \nip{\pi_{A_P}}{\omega^P_{A_P}} = d^{-K} \sum_{\omega^P} d^{c(\pi^{-1}\omega^P)} = d^{-K} \sum_{\omega^P} d^{c(\omega^P)} = d^{-K} K! \binom{d + K-1}{K}, \label{eq:evaluate-cycle-sum}
		\end{align}
	where the second equality holds since $\omega^P \mapsto \pi^{-1}\omega^P$ for fixed $\pi$ is a bijection on $\kS_K$.
	Applying the identity \eqref{eq:evaluate-cycle-sum} successively to each of the $2P$ sums appearing in \eqref{eq:K-multi-sum05}  starting with the right-most one, and using $\left(\prod_{p=1}^P q^{-m_p K}\right)^2=q^{-2K\sum_{p=1}^Pm_p}=q^{-2MK}$, then gives 
	\begin{align}
		\sum_{\sigma^2,\dots,\sigma^P} \sum_{\vomega} \left(\bN_{A_1\dots A_P}^T\bN_{A_1\dots A_P}\right)_{\vsigma,\vomega}&= q^{-2MK} K!^{2P} \prod_{p=1}^P \binom{q^{m_p}+K-1}{K}^2.\label{eq:K-multi-sum1}
	\end{align}
	
	Similarly, for the second sum in \eqref{eq:K-op-norm-multi-calculation} we get
	\begin{align}
		\sum_{\sigma^2,\dots,\sigma^P} (\bN_{A_1\dots A_P}^T\bN_{A_1\dots A_P})_{\vsigma,\vsigma} &= \sum_{\sigma^2,\dots,\sigma^P} \sum_\pi \prod_{p=1}^P \nip{\pi_{A_p}}{\sigma^p_{A_p}}^2\\
		&= \sum_\pi \nip{\pi_{A_1}}{\sigma^1_{A_1}}^2 \sum_{\sigma^2} \nip{\pi_{A_2}}{\sigma^2_{A_2}}^2 \dots \sum_{\sigma^P} \nip{\pi_{A_P}}{\sigma^P_{A_P}}^2 \label{eq:K-multi-sum15}\\
		&= q^{-2MK} K!^P \prod_{p=1}^P \binom{q^{2m_p} + K - 1}{K},\label{eq:K-multi-sum2}
	\end{align}
	where we used \eqref{eq:harrow-cycle-sum} with $d=q^{2m_P}$ to obtain
	\begin{align}
		\sum_{\sigma^P} \nip{\pi_{A_P}}{\sigma^P_{A_P}}^2 = q^{-2m_PK} \sum_{\sigma^P} \left(q^{2m_P}\right)^{c(\pi^{-1}\sigma^P)} = q^{-2m_PK} K! \binom{q^{2m_P}+K-1}{K},
	\end{align}
	and similarly for the other sums in \eqref{eq:K-multi-sum15}.
	Substituting \eqref{eq:K-multi-sum1} and \eqref{eq:K-multi-sum2} in \eqref{eq:K-op-norm-multi-calculation} gives the bound
	\begin{align}
		\|\bE\| \leq b_2 q^{-2MK} K!^{2P} \prod_{p=1}^P \binom{q^{m_p}+K-1}{K}^2 - b_2 q^{-2MK} K!^P \prod_{p=1}^P \binom{q^{2m_p} + K - 1}{K}. \label{eq:K-op-norm-multi}
	\end{align}
	
	To bound the operator norm of the diagonal matrix $\bD$ in \eqref{eq:subspace-angle-multi-bound}, we first note that all diagonal entries of $\bD$ are positive:
	We have
	\begin{align}
		\frac{b_3b_4}{b_2} = \exp\left(\frac{K^2}{2q^{M}}\right) \prod_{p=1}^P\exp\left(\frac{K^2}{q^{m_p-\ell_p}}\right)  \left(1-\frac{K^2}{2q^{L}}\right)^{-1}  \geq 1, \label{eq:constants-relation}
	\end{align}
	and thus $b_3b_4\geq b_2$.
	Moreover, recall from \cite{harrow_approximate_2023-1} that for $\sigma,\tau\in\kS_K$ the normalized Frobenius inner product of two permutation operators $\sigma_X,\tau_X$ acting on a space $X^K$ is equal to 
	\begin{align}
		\nip{\sigma_X}{\tau_X} = d^{c(\sigma^{-1}\tau) - K},
	\end{align} 
	where $d=|X|$ is the dimension of $X$, and $c(\pi)$ denotes the number of cycles in $\pi\in\kS_K$.
	Since $c(\pi) \leq K$ for any $\pi\in\kS_K$ and $c(\pi)<K$ whenever $\pi\neq \operatorname{id}$ (where $\operatorname{id}\in\kS_K$ denotes the identity permutation), it then follows that, for two spaces $X,Y$ with $|X| < |Y|$,
	\begin{align}
		\nip{\sigma_X}{\tau_X} = |X|^{c(\sigma^{-1}\tau)-K} \geq |Y|^{c(\sigma^{-1}\tau)-K} = \nip{\sigma_Y}{\tau_Y},\label{eq:nip-dimension}
	\end{align}
	with strict inequality if $\sigma\neq \tau$.
	
	The inequality $b_3b_4\geq b_2$ and \eqref{eq:nip-dimension} imply that
	\begin{align}
		(\bD)_{\vsigma,\vsigma} &= \left(b_3b_4\bN_{\ell_1,\dots,\ell_P}^T \bN_{\ell_1,\dots,\ell_P} - b_2\bN_{A_1\dots A_P}^T\bN_{A_1\dots A_P}\right)_{\vsigma,\vsigma}\\
		&\geq b_2 \left(\bN_{\ell_1,\dots,\ell_P}^T \bN_{\ell_1,\dots,\ell_P} - \bN_{A_1\dots A_P}^T\bN_{A_1\dots A_P}\right)_{\vsigma,\vsigma}\\
		&>0.
	\end{align}
	The last strict inequality follows from the calculation
	\begin{align}
		\left(\bN_{\ell_1,\dots,\ell_P}^T \bN_{\ell_1,\dots,\ell_P}\right)_{\vsigma,\vsigma} &= \sum_{\pi} \nip{ \bigotimes\nolimits_{p=1}^P \sigma^p_{A_p^{\ell_p}} }{\pi_{A_1^{\ell_1}\dots A_P^{\ell_P}}} \nip{\pi_{A_1^{\ell_1}\dots A_P^{\ell_P}}}{\bigotimes\nolimits_{p=1}^P \sigma^p_{A_p^{\ell_p}} }\\
		&= \sum_\pi \prod_{p=1}^P \nip{\sigma^p_{A_p^{\ell_p}}}{\pi_{A_p^{\ell_p}}} \prod_{p=1}^P \nip{\pi_{A_p^{\ell_p}}}{\sigma^p_{A_p^{\ell_p}}}\\
		&> \sum_\pi \prod_{p=1}^P \nip{\sigma^p_{A_p}}{\pi_{A_p}} \prod_{p=1}^P \nip{\pi_{A_p}}{\sigma^p_{A_p}}\\
		&= \sum_{\pi} \nip{ \bigotimes\nolimits_{p=1}^P \sigma^p_{A_p } }{\pi_{A_1 \dots A_P }} \nip{\pi_{A_1 \dots A_P }}{\bigotimes\nolimits_{p=1}^P \sigma^p_{A_p } }\\
		&= \left(\bN_{A_1\dots A_P}^T\bN_{A_1\dots A_P}\right)_{\vsigma,\vsigma},
	\end{align}
	where we used \eqref{eq:nip-dimension} for the inequality, which is strict if at least one $\ell_p < m_p$.
	
	The operator norm of the diagonal matrix $\bD$ is thus equal to the largest diagonal entry, say, in row $\vsigma$.
	Similar to above, we then sum over $\sigma^2,\dots,\sigma^P$ to get an explicit upper bound on $\|\bD\|$:
	\begin{align}
		\|\bD\| &= (\bD)_{\vsigma,\vsigma}\\
		&\leq \sum_{\sigma^2,\dots,\sigma^P} (\bD)_{\vsigma,\vsigma}\\
		&= \sum_{\sigma^2,\dots,\sigma^P} \left(b_3b_4\bN_{\ell_1,\dots,\ell_P}^T \bN_{\ell_1,\dots,\ell_P} - b_2\bN_{A_1\dots A_P}^T\bN_{A_1\dots A_P}\right)_{\vsigma,\vsigma}\\
		&= b_3b_4 \sum_{\sigma^2,\dots,\sigma^P}  \left(\bN_{\ell_1,\dots,\ell_P}^T \bN_{\ell_1,\dots,\ell_P}\right)_{\vsigma,\vsigma} - b_2 \sum_{\sigma^2,\dots,\sigma^P} \left(\bN_{A_1\dots A_P}^T\bN_{A_1\dots A_P}\right)_{\vsigma,\vsigma} \label{eq:apply-sum-here}\\
		&= b_3b_4 q^{-2LK} K!^P \prod_{p=1}^P \binom{q^{2\ell_p}+K-1}{K} - b_2 q^{-2MK} K!^P \prod_{p=1}^P \binom{q^{2m_p}+K-1}{K},\label{eq:D-op-norm-multi}
	\end{align}
	where we used the identity \eqref{eq:K-multi-sum2} for the second sum in \eqref{eq:apply-sum-here}, and with the replacement $m_p \rightarrow \ell_p$ for the first sum therein.
	
	Substituting \eqref{eq:K-op-norm-multi} and \eqref{eq:D-op-norm-multi} in the bound \eqref{eq:subspace-angle-multi-bound} on the subspace angle cosine finally gives
	\begin{align}
		c(\im\cP,\im\cQm)^2 &\leq b_1b_3b_4 q^{-2LK} K!^P \prod_{p=1}^P \binom{q^{2\ell_p}+K-1}{K} - 2b_1b_2 q^{-2MK} K!^P \prod_{p=1}^P \binom{q^{2m_p}+K-1}{K} \notag\\
		& \quad{} + b_1b_2 q^{-2MK} K!^{2P} \prod_{p=1}^P \binom{q^{m_p}+K-1}{K}^2,
	\end{align}
	which concludes the proof of Proposition \ref{prop:subspace-angle-multi-crosstwirl}.

    \subsubsection{Open problem: Improving convergence bound for Twirl-Swap-Twirl protocol}
    \label{sec:technical-open-problem}

    The convergence in the construction of both the \hyperref[item:twirl-swap-twirl]{Twirl-Swap-Twirl} (Proposition~\ref{prop:subspace-angle-swap}) and the \hyperref[item:twirl-crosstwirl]{Twirl-Crosstwirl} protocols (Proposition~\ref{prop:subspace-angle-multi-crosstwirl}) is determined by the operator norm of the two matrices $\bX$ and $\bY$ defined in the previous sections:
\begin{align}
	\bX=a\bM^2-a_1 \bN^T\bN
\end{align}
with constants $a,a_1$ defined in \eqref{eq:constants-twirl-swap-twirl} and matrices $\bM,\bN$ defined in \eqref{eq:M-text} and \eqref{eq:N-text1}, respectively, and 
\begin{align}
	\bY=b_3b_4\one_{K!^P} \odot \bN_{\ell_1,\dots,\ell_P}^T\bN_{\ell_1,\dots,\ell_P}- b_2\bN_{A_1\dots A_P}^T\bN_{A_1\dots A_P}
\end{align}
with $\odot$ denoting the Hadamard product, the constants $b_2,b_3,b_4$ defined in \eqref{eq:constants-crosstwirl}, and the matrices $\bN_{\ell_1,\dots,\ell_P}, \bN_{A_1\dots A_P}$ defined in \eqref{eq:N-ell} and \eqref{eq:N-multipartite}, respectively.
The elements of both matrices are functions of permutation operators.
While the operator norm of $\bY$ can be controlled quite transparently using the generalized Perron-Frobenius Theorem, the different structure of the matrix $\bX$ defies this approach for \hyperref[item:twirl-swap-twirl]{Twirl-Swap-Twirl}, requiring a coarser analysis that leads to a weaker bound on the convergence of the design construction in this case.
A more refined analysis of $\|\bX\|$ may lead to an improved bound on this convergence mirroring that for the \hyperref[item:twirl-crosstwirl]{Twirl-Crosstwirl} protocol, and a general, more unified way of analyzing whether different bipartite interactions yield efficient designs.

\section{Log-depth Designs on Lattice Architectures}  \label{sec:lattices}

\begin{figure}
    \begin{minipage}{0.45\textwidth}
        \begin{subfigure}[t]{\textwidth}
            \includegraphics[width=\textwidth]{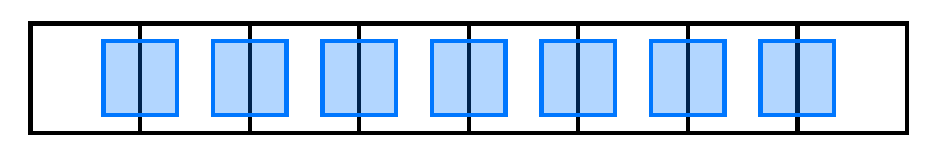}
            \caption{$1$-dimensional lattice.}
        \end{subfigure}\\
        \begin{subfigure}[t]{\textwidth}
            \includegraphics[width=\textwidth]{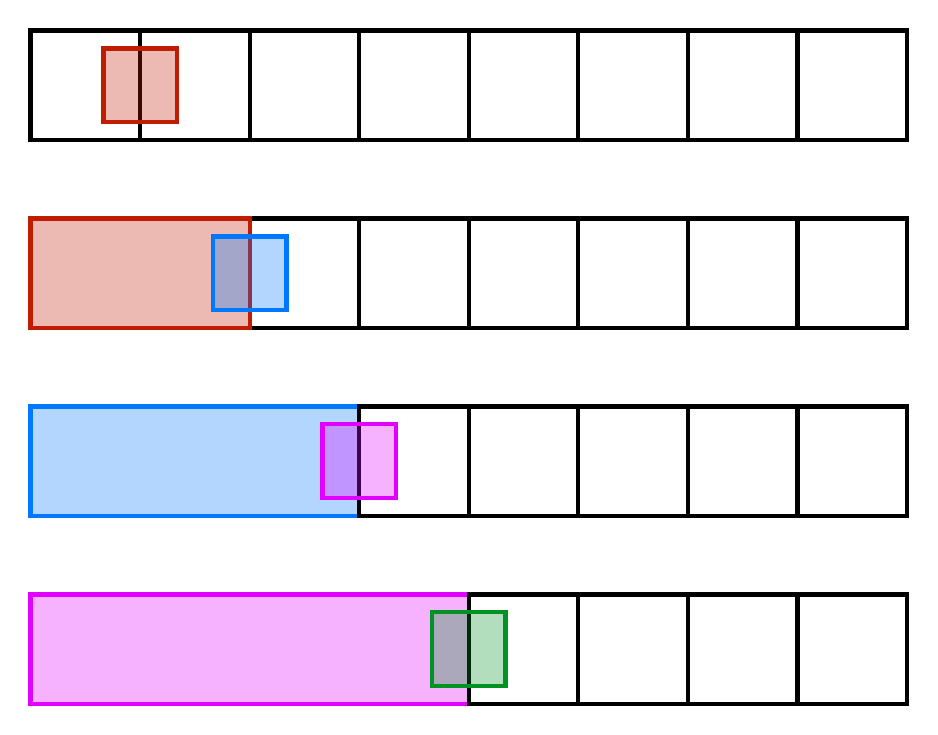}
            \caption{Sequence of iterative steps on a line.}
        \end{subfigure}
    \end{minipage}
    \hfill
    \begin{minipage}{0.45\textwidth}
        \begin{subfigure}[t]{\textwidth}
            \includegraphics[width=\textwidth]{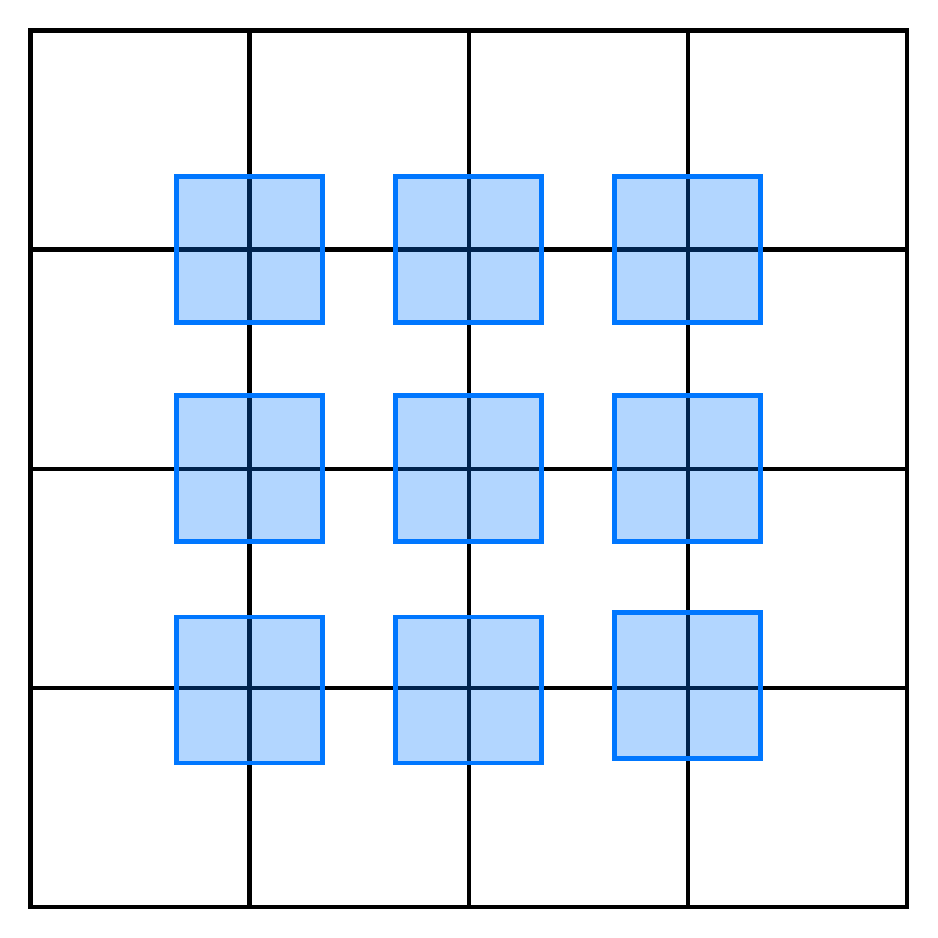}
            \caption{$2$-dimensional lattice.}
        \end{subfigure}
    \end{minipage}
    \caption{Illustrations of the 2-layer procedure in Protocol \ref{alg:2-layer}  designs on lattices in one spatial dimension (panels ({\scriptsize A}) and ({\scriptsize B})) and in two spatial dimensions (panel ({\scriptsize C})). In each case the blocks participating in the first step of Protocol \ref{alg:2-layer} are shown as white squares with black outlines, whereas the blocks participating in the second step are shown as colored squares.
    Panel ({\scriptsize B}) depicts the first four iterations of crosstwirls in the second step of Protocol \ref{alg:2-layer} on a $1$-dimensional lattice. In each step, a crosstwirl joins a new first-layer block (shown in white) to the previously merged block (shown at the left edge of the line in color).}
    \label{fig:lattice}
\end{figure}

As shown by Proposition \ref{prop:subspace-angle-multi-crosstwirl}, one may use the crosstwirl as a superchannel that takes a few approximate $k$-design channels in tensor product and returns an approximate $k$-design channel on the full system. In this section, we derive schemes that build up large $k$-designs from small $k$-designs by iterating the crosstwirl.

Given two large systems to which local twirls have been applied, Theorem \ref{thm:twirl-crosstwirl-relative} shows that twirling a relatively small system overlapping each results in the entire system being in a design. However, the original local twirls might already have been implemented this way on each local system, subdividing the whole into 4 parts. The local twirls in that layer may have themselves used instances of the twirl-crosstwirl protocol, initially starting with 8 parts. After iterating many more times, one may find that the original local twirls are of number proportional to the total number $m$ of qudits and of size almost independent of it. Because the crosstwirls are also small, the entire scheme is then a composition of twirls on systems of size growing only logarithmically with system size (a constraint enforced by approximately linear error buildup). Using existing linear-depth schemes \cite{brandao_local_2016, chen2024incompressibility}, each of these logarithmic-size twirls requires only logarithmic depth in $m$. They mostly parallelize, ultimately yielding a logarithmic-depth scheme approximating the full twirl. In this Section, we analyze a similar scheme that explicitly reduces to two layers, each applying local twirls to logarithmic-size subsystems. Hence we show as our culminating result that $O(\log m)$ depth suffices for unitary designs on $m$ qudits. This result addresses \cite[Section 1.5, Open Question 1]{harrow_mehraban2023approximate} and \cite[Section 1.5, Open Question 7]{harrow_mehraban2023approximate}, showing that relative error designs may form in sublinear depth even when that depth is much smaller than the graph diameter of an interaction lattice.


For the culminating technical result, consider the following 2-layer procedure on $m$ qudits arranged in a $D$-dimensional hyper-rectangle (see also Fig. ~\ref{fig:lattice}):

\begin{algo}[Two-step protocol]~
\label{alg:2-layer}
    Consider $m$ qudits arranged in a $D$-dimensional hyper-rectangle, and let $\ell \coloneqq \lceil \log_q (60 m k!^{4 D - 1} k D / \epsilon) \rceil$.
    \begin{enumerate}
        \item Starting from one corner, partition the hyper-rectangle into smaller hypercubes of size $\lceil \sqrt[D]{\ell} \rceil^D$ until hypercubes of this size no longer fit. When such a condition is reached, join remaining qudits into their nearest hypercubes, obtaining blocks of size at most $2^D \lceil \sqrt[D]{\ell} \rceil^D$. On each such block, draw a random unitary from an approximate relative-error $k$-design and apply it.
        \item Construct a block centered at each point where $2^D$ first-layer blocks meet in such a way that the new block overlaps with at least $\ell$ qudits of each first-layer block. For each such block, draw a random unitary from an approximate relative-error $k$-design, and apply it to that second-layer block.
    \end{enumerate}
\end{algo}

\begin{theorem} \label{thm:lattice-main}
    In fixed spatial dimension $D$, for any $\epsilon > 0$, Protocol \ref{alg:2-layer} yields an $\epsilon$-approximate relative $k$-design ensemble with unitaries of depth
    \begin{equation}
        O(\ell \times k \mathrm{polylog}(k) ) = O((k \log k + \log m + \log (1/\epsilon)) \times k \mathrm{polylog}(k) )
    \end{equation}
    that require at most 
    \begin{align} 
    O \big ( \#(\partial S) \times \ell \big )
    \label{eq:quantum-communication}
    \end{align}
    qudits of quantum communication between any lattice subregion $S$ and its complement.
    Here $\#(\partial S)$ denotes the number of qudits on the boundary of $S$.
\end{theorem}

Theorem \ref{thm:lattice-main} shows that logarithmic-depth designs can be obtained on any lattice connectivity. It also shows that up to multiplicative, logarithmic corrections, the entanglement entropy between subregions follows nearly an area law \cite{eisert2010colloquium}, analogously to the near-area law shown for quantum pseudoentanglement \cite{aaronson_quantum_2023}. To prepare the proof of Theorem \ref{thm:lattice-main}, we first prove the following auxiliary lemma.

\begin{lemma} \label{lemma:iterate-error-additive}
        Consider two families of channels $(\Phi_n)_{n=1}^N$ and $(\Psi_n)_{n=1}^N$ such that
    \begin{equation}
        (1-\epsilon_n) \Psi_n \prec \Phi_n \prec (1+\epsilon_n) \Psi_n
    \end{equation}
    for each $n = 1 \dots N$. With $\epsilon = \prod_n (1 + \epsilon_n) - 1$,
    \begin{align}
        (1-\epsilon) \prod_n \Psi_n \prec \prod_n \Phi_n \prec (1+\epsilon) \prod_n \Psi_n
    \end{align}
    Furthermore, assume that $(1-\delta) \Gamma \prec \prod_n \Psi_n \prec (1+\delta) \Gamma$ for some channel $\Gamma$. Then the total relative error of $
	\Phi$ with respect to $\Gamma$ is bounded from above by $(1+\delta)(1+\epsilon)$.
\end{lemma}
\begin{proof}
    Recall that for each $n$, $\Phi_n =  (1-\epsilon) \Psi_n + \epsilon \Theta_n$ for some channel $\Theta_n$.
    Therefore,
    \begin{align}
        \prod_n \Phi_n = \prod_n (1-\epsilon_n) \Psi_n + \Big ( 1 - \prod_n (1-\epsilon_n) \Big ) \Theta
    \end{align}
    for some channel $\Theta$ that is a convex combination of $\Theta_n$ for different $n$. Similarly, $(1+\epsilon)^{-1} \Psi_n = (1 - (1+\epsilon)^{-1}) \Phi_n + \epsilon \Theta_n'$ for some channel $\Theta_n'$, so
    \begin{align}
        \prod_n \Psi_n & = \prod_n (1+\epsilon_n)^{-1} \Phi_n + \Big ( 1 - \prod_n (1+\epsilon_n)^{-1} \Big ) \Theta'
    \end{align}
    for some channel $\Theta'$. Therefore,
    \begin{align}
        \prod_n (1-\epsilon_n) \Psi_n  \prec \prod_n \Phi_n \prec \prod_n (1+\epsilon_n) \Psi_n \pl.
    \end{align}
     Using Remark \ref{rem:generalbernoulli},
     \begin{equation}
         \prod_n (1-\epsilon_n) \geq 1 - \sum_n \epsilon_n \geq 1 - \bigg ( 1 - \prod_n (1 + \epsilon_n) \bigg ) \pl.
     \end{equation}
    Therefore, one obtains that with $\epsilon' = 1 - \prod_n (1 + \epsilon_n) \leq 1 + \sum_n \epsilon_n \leq \prod_n (1 + \epsilon_n)$,
    \begin{align}
        (1-\epsilon') \prod_n \Psi_n \prec \prod_n \Phi_n \prec (1+\epsilon) \prod_n \Psi_n
    \end{align}
    The final part of the Lemma follows from expanding $\prod_n \Phi_n$ as a convex combination involving the maps $\prod_n \Psi_n$.
\end{proof}
\begin{proof}[Proof of Theorem \ref{thm:lattice-main}]
    In Protocol \ref{alg:2-layer}, each second-layer block overlaps with $2 D$ first-layer blocks. 
    Let $\epsilon_1$ denote the error of the designs applied to the individual blocks in the first-layer, $\epsilon_2$ the corresponding error parameter for the second layer, and let $a$ and $b$ denote the respective upper and lower bounds on the sizes of first-layer blocks. Assume we have some chunk of blocks on which a `merged' design incorporating some of the first-layer blocks has already been applied, and assume that such a block has error at most $\epsilon_{\mathrm{old}}$. Using Lemma \ref{lemma:iterate-error-additive} with Theorem \ref{thm:twirl-crosstwirl-relative}, each second-layer block joins the first-layer blocks with which it overlaps into a $k$-design with relative error
    \begin{equation} \label{eq:block-merge}
        (1 + \epsilon_{\mathrm{new}}) \leq \Big ( 1 + 10 k!^{4 D - 1} k \frac{2 D}{q^{\ell}} \Big ) (1 + \epsilon_2) (1+\epsilon_{\mathrm{old}}) (1 + \epsilon_1)^{2 D - 1}  \pl.
    \end{equation}
    At the first step of iteration, $\epsilon_{\mathrm{old}} = \epsilon_1$, and the second-layer `crosstwirl' combines the first 2D first-layer blocks. At the next iteration, this `merged' block is combined with additional first-layer blocks of which there are at most $2 D - 1$ (and usually fewer). At each subsequent step, we choose a second-layer block that overlaps a first-layer block and a previously merged first-layer block. This is illustrated for one spatial dimension in panel ({\scriptsize B}) of Figure \ref{fig:lattice}. Iterating Equation \eqref{eq:block-merge}, we ultimately combine all first-layer blocks via crosstwirls with a total error $\delta$ for which
    \begin{equation}
        (1+\delta) \leq \Big ( 1 + 10 k!^{4 D - 1} k \frac{2 D}{q^{\ell}} \Big )^{2 D \lceil m / a \rceil} \times (1 + \epsilon_1)^{(2D - 1) \lceil m / a \rceil} \times (1 + \epsilon_2)^{\lceil m / a \rceil} \pl,
    \end{equation}
    which can be upper-bounded via the reverse Bernoulli's inequality and Remark \ref{rem:generalbernoulli} as
    \begin{equation}
        \begin{split} (1 + \delta) \leq \exp \bigg ( 1 + \left\lceil \frac{m}{a} \right\rceil \Big( 10 k!^{4 D - 1} k \frac{4 D^2}{q^{\ell}}
            + (2 D - 1) \epsilon_1 + \epsilon_2 \Big ) \bigg ) \pl.
        \end{split}
    \end{equation}
    To achieve a given value of $\epsilon$, we set $\ell$ as noted in Protocol \ref{alg:2-layer} and $\epsilon_1 = \epsilon / 3 (2 D - 1) m$, and $\epsilon_2 = \epsilon / 3m$. Therefore,
    \begin{equation}
        (1 + \delta) \leq \exp \bigg ( 1 + \left\lceil \frac{m}{a} \right\rceil \frac{\epsilon}{m} \bigg ) \pl.
    \end{equation}
    Since the bounds are for asymptotics in large $m > k$ with $\epsilon < 1$, we will assume that $\ell$ and therefore $a$ is large (if $m$ is not large compared to $k$, one may directly apply the bound from \cite{chen2024incompressibility}). Therefore, we may ignore the nonlinear terms in the exponential Taylor expansion and conclude that $\delta \leq \epsilon$.
    
    Using the results of \cite{chen2024incompressibility}, for a block of size $r$, local random unitaries achieve an $\epsilon_1$-approximate relative $k$-design in depth $O(((r k + \log (1/\epsilon) (\log k)^7 )$. Therefore, we may assume that after the first layer is applied, each first-layer block has undergone a unitary from such a $k$-design, and that the total circuit is of depth $O(((b k + \log (1/\epsilon) (\log k)^7 )$. We note that $b = O(\ell)$. We may also assume that the second layer applies random circuits from the scheme of \cite{chen2024incompressibility}, yielding unitaries from $\epsilon_2$-approximate $k$-designs on each second-layer block in depth $O(((\ell k + \log (1/\epsilon) (\log k)^7 )$. Although the scheme from \cite{chen2024incompressibility} applies to one-dimensional lattices, we may construct a one-dimensional path through higher-dimensional blocks and apply it therein as well.

    To see that the boundary entanglement is upper-bounded by \eqref{eq:quantum-communication}, consider a spatially contiguous (but not necessarily regular) subregion $S$ with boundary $\partial S$, defined as the set of qudits within $S$ that are adjacent to qudits not in $S$. Since the boundary contains $\#(\partial S)$ qudits, it intersects at most this many blocks at the first or second layer. To implement an arbitrary unitary on each block requires communicating a number of qudits at most equal to the block size to be transferred across any boundary through that block, and one may see that this also upper-bounds the amount of entanglement created by a unitary across any boundary through the block. The bound \eqref{eq:quantum-communication} follows from these considerations and the fact that the size of each first or second-layer block is $O(\ell)$ for fixed $D$. 
    %
    %
    %
    %
\end{proof}
\begin{rem}[Almost Brickwork] \normalfont \label{rem:almost-brickwork}
	A recent result \cite{chen2024incompressibility} showed that brickwork random circuits on a one-dimensional lattice are $k$-designs with depth $O(m k \polylog(k))$. On the 1-hypercube, the steps used in Theorem \ref{thm:lattice-main} can be implemented this way, yielding the Theorem statement with $g(k) = O(\polylog(k))$. To satisfy the criteria of Theorem \ref{thm:lattice-main}, the only deviation from overall brickwork is that random unitaries not cross boundaries separating different subsystems in the first (or second) layer. On higher-dimensional lattices, it is possible to iterate the scheme of \cite{harrow_mehraban2023approximate} to obtain relative error local twirls and crosstwirls in depth linear in their size, yielding again a logarithmic depth scheme with polynomial $k$-dependence that is nearly $D$-dimensional brickwork.
\end{rem}

\begin{rem}\label{rem:hamiltonian-vs-recursive} \normalfont
    A concurrent work \cite{schuster_random_2024} also shows that designs can be constructed on any geometry in logarithmic depth. That result uses a Hamiltonian path construction, in which any connected graph can efficiently approximate a one-dimensional layout. Via a detailed technical analysis, the other work also obtains an improved dependence on $k$. Straightforwardly counting the points at which a Hamiltonian path may cross a boundary suggests that subregion entanglement in such a construction scales as $\#(\partial S) \times O(\log k + \log m + \log \epsilon)$. 
\end{rem}

%

\printbibliography

\end{document}